\documentclass[a4paper,UKenglish,cleveref, autoref,thm-restate]{lipics-v2021}

\newtheorem{uccp}{Reduction Rule ECGP}
\newtheorem{buccp}{Reduction Rule BECGP}
\usepackage{xcolor}
\DeclareRobustCommand{\legendbox}[1]{%
  \tikz[baseline=-0.6ex]{\node[#1, minimum height=0.3cm, minimum width=0.3cm, inner sep=0pt] {};}%
}
\usepackage{todonotes}

\usepackage{xspace}
\usepackage{amsmath}
\usepackage{amssymb}

\newcommand{\Z}{\ensuremath{\mathbb{Z}}}

\newcommand{\vc}{\mathsf{vc}}
\newcommand{\fvs}{\mathsf{fvs}}
\newcommand{\fes}{\mathsf{fes}}
\newcommand{\tw}{\mathsf{tw}}
\newcommand{\pw}{\mathsf{pw}}
\newcommand{\td}{\mathsf{td}}
\newcommand{\cw}{\mathsf{cw}}
\newcommand{\mw}{\mathsf{mw}}
\newcommand{\nd}{\mathsf{nd}}
\newcommand{\vi}{\mathsf{vi}}
\newcommand{\vdc}{\mathsf{vdc}}
\newcommand{\cvd}{\mathsf{cvd}}
\newcommand{\vdp}{\mathsf{vdp}}
\newcommand{\vds}{\mathsf{vds}}
\newcommand{\mln}{\mathsf{mln}}

\usetikzlibrary{patterns}
\usetikzlibrary{patterns.meta}

\title{Parameterized Complexity of
Edge-Constrained Graph Partitioning} %TODO Please add

\author{Ajinkya Gaikwad}{Czech Technical University in Prague
Prague, Czech Republic}{ajinkya.gaikwad@fit.cvut.cz}{https://orcid.org/0000-0002-7514-0708}{}%TODO mandatory, please use full name; only 1 author per \author macro; first two parameters are mandatory, other parameters can be empty. Please provide at least the name of the affiliation and the country. The full address is optional. Use additional curly braces to indicate the correct name splitting when the last name consists of multiple name parts.

\author{Jan Pokorný}{Czech Technical University in Prague
Prague, Czech Republic}{jan.pokorny@fit.cvut.cz}{https://orcid.org/0000-0003-3164-0791}{}

\author{Tomáš Valla}{Czech Technical University in Prague
Prague, Czech Republic}{Tomas.Valla@fit.cvut.cz}{https://orcid.org/0000-0003-1228-7160}{}

\authorrunning{A. Gaikwad et al.} %TODO mandatory. First: Use abbreviated first/middle names. Second (only in severe cases): Use first author plus 'et al.'

\Copyright{Jane Open Access and Joan R. Public} %TODO mandatory, please use full first names. LIPIcs license is "CC-BY";  http://creativecommons.org/licenses/by/3.0/

 %TODO mandatory: Please choose ACM 2012 classifications from https://dl.acm.org/ccs/ccs_flat.cfm 

 %TODO mandatory; please add comma-separated list of keywords

\category{} %optional, e.g. invited paper

\relatedversion{} %optional, e.g. full version hosted on arXiv, HAL, or other respository/website
\funding{Research supported by the Czech Science Foundation Grant no.~24-12046S. This work was co-funded by the European Union under the project Robotics and Advanced Industrial Production (reg. no.~CZ.02.01.01/00/22\_008/0004590).}
\renewcommand{\linenumbers}{}
\nolinenumbers
\EventEditors{John Q. Open and Joan R. Access}
\EventNoEds{2}
\EventLongTitle{42nd Conference on Very Important Topics (CVIT 2016)}
\EventShortTitle{CVIT 2016}
\EventAcronym{CVIT}
\EventYear{2016}
\EventDate{December 24--27, 2016}
\EventLocation{Little Whinging, United Kingdom}
\EventLogo{}
\SeriesVolume{42}
\ArticleNo{23}
\begin{document}
\maketitle

\begin{abstract}
We study the \emph{Edge-Constrained Graph Partitioning Problem}
(\textsc{ECGP}), where the
the goal is to partition the vertices of a graph into \(r\) parts such that
every part has the number of edges at least \(\gamma\). We also
consider a balanced variant (\textsc{BECGP}), where all parts are required to
have equal size, and signed variants in which we measure the
difference between the number of positive and negative edges.

We first show that both \textsc{ECGP} and \textsc{BECGP} remain NP-hard for fixed constant values of $\gamma$, and that \textsc{BECGP} remains NP-hard for fixed constant values of $r$.
We therefore study their parameterized complexity with respect to the natural parameters \(r\), \(\gamma\), and several structural graph parameters. For the natural parameterization by \(r+\gamma\), we show that both \textsc{ECGP} and \textsc{BECGP} admit a polynomial kernel. We further consider structural graph parameters. On the positive side, we obtain FPT algorithms for \textsc{ECGP} and \textsc{BECGP} parameterized by maximum leaf number, by deletion distance to a clique, by cluster vertex deletion set plus $\gamma$ and by vertex integrity. We additionally show that \textsc{ECGP} is FPT when parameterized by vertex deletion to stars plus $\gamma$ and by vertex deletion to paths plus $\gamma$.

In contrast, we prove that \textsc{ECGP} and \textsc{BECGP} remain W[1]-hard when parameterized by \(r\) together with several structural parameters. In particular, the reductions for feedback edge set, vertex deletion to stars, vertex deletion to paths, and modular width establish hardness even when the corresponding parameter is \(0\), while we additionally prove W[1]-hardness parameterized by cluster vertex deletion set plus \(r\). Both the problems remain W[1]-hard when parameterized by cliquewidth even when $\gamma=3$. For signed graphs, we show NP-hardness for both variants even when \(r+\gamma=3\) and the input graph is a disjoint union of two cliques. Also we prove that the balanced signed variant is W[1]-hard when parameterized by treedepth plus $r$, even when \(\gamma=0\).
\end{abstract}

\section{Introduction}
Consider the task of dividing players into football teams. Some pairs of players cooperate well together, while others may have
poor compatibility or conflicting play styles. These relations may be captured
in a graph: the vertices represent the players
and an edge between two players indicates that they cooperate well.
One may require that every team individually achieves a minimum level of compatibility,
which corresponds to partitioning the graph into parts such that each part has a given minimum number of edges.
In many practical settings, additional constraints such as balanced team sizes are also important.
Furthermore, relations may be both positive and negative, which we model by assigning a sign to each edge. Such signed graph models have been studied in clustering and social network analysis, where the goal is typically to optimize a global objective, for instance by maximizing agreements or minimizing disagreements.
In this paper we study several versions of this coalition formation problem with
local constraints, mostly from the point of view of the parameterized complexity
theory.

Coalition formation is a fundamental problem in multi-agent systems, with applications spanning social, economic, and political domains, including team formation, collaborative networks, and community detection. In many such settings, individuals operate in groups rather than alone, and their satisfaction depends on the composition of the group to which they belong. 
A well-established framework for modeling such scenarios is that of hedonic games, introduced by Dr\`eze and Greenberg~\cite{69d4f077-72de-3815-b927-32755a2b5480} and further developed in a rich body of work. Early contributions focused on the existence and structure of stable coalition outcomes~\cite{banerjee2001core, BOGOMOLNAIA2002201, DEMANGE199445, 10.2307/1885113}, while subsequent research expanded the model to richer preference classes and provided comprehensive overviews of the area~\cite{Aziz_Savani_2016, 10.1007/978-3-642-35843-2_4}. More recent work has investigated computational and structural aspects of coalition formation, including graph-restricted and network-based settings~\cite{7e72472a-d8cc-392b-a7e5-767d9d2180f5, demange2004group}. In this model, outcomes correspond to partitions of agents into coalitions, where each agent’s utility depends solely on the members of its own coalition. These models provide a natural abstraction for analyzing how local interactions influence global group structures.

More formally, we study graph partitioning problems that enforce \emph{local constraints} on every part.
Given an undirected graph $G=(V,E)$ and integers $r$ and $\gamma$, the \emph{Edge-Constrained Graph Partitioning Problem} (\textsc{ECGP}) asks whether $V$ can be partitioned into $r$ parts such that each part induces at least $\gamma$ edges.
We also consider a balanced variant, the \emph{Balanced Edge-Constrained Graph Partitioning Problem} (\textsc{BECGP}),
in which all parts are required to have equal size, modeling scenarios such as forming equally sized teams or distributing workload evenly.

We further extend these formulations to signed graphs, where each edge is labeled as positive or negative.
%and the utility of a part is defined as the difference between the number of positive and negative edges it induces.
In this setting, each part must achieve a prescribed difference between the number of positive and negative edges it induces,
capturing a trade-off between cooperation and conflict.
These formulations can be viewed as constraint-based counterparts to classical hedonic models,
where instead of optimizing global objectives or enforcing stability, we require each coalition to satisfy a minimum utility, defined as the edge-count of each coalition.
Our goal is to understand the computational complexity of these problems under various structural parameterizations,
as well as with respect to the parameters $r$ and $\gamma$, which denote the number of parts and the required edge-count threshold, respectively.
We first show that both \textsc{ECGP} and \textsc{BECGP} remain NP-hard even for fixed constants \(r\) and \(\gamma\) (see Theorem~\ref{thm:np-hard 1} and Theorem~\ref{thm:np-hard 2}).
We therefore study their parameterized complexity with respect to the combination of natural parameters \(r\), \(\gamma\), and several structural graph parameters.

\noindent\textbf{Our Contributions.}
We provide a systematic study of the computational and parameterized complexity of \textsc{ECGP} and \textsc{BECGP}, along with their signed variants, showing several tractability and hardness results.

First, using the Expansion Lemma (Lemma~\ref{lem:expansion_kernel}) as the main ingredient, we obtain polynomial kernels for both \textsc{ECGP} and \textsc{BECGP} parameterized by $r+\gamma$. More precisely, \textsc{ECGP} admits a kernel with $\mathcal{O}(r\gamma^2)$ vertices and $\mathcal{O}((r\gamma)^2)$ edges (Theorem~\ref{thm:kernel_ru}), while \textsc{BECGP} admits a kernel with $\mathcal{O}(r\gamma^2+r^2)$ vertices (Theorem~\ref{thm:becgp-kernel}).

We then investigate the parameterized complexity of the problems with respect to structural graph parameters. 
On the positive side, we show (Theorem~\ref{thm:uccp-vin}) that both \textsc{ECGP} and \textsc{BECGP} are fixed-parameter tractable when parameterized by vertex integrity.
We also show (Theorems~\ref{thm:uccp-vdc} and \ref{thm:buccp-vdc}) that \textsc{ECGP} and \textsc{BECGP} are FPT when parameterized by vertex deletion to clique.
These tractability results rely on different ILP-based approaches.
For vertex integrity, we formulate the problem as an $N$-fold ILP and apply known FPT algorithms for $N$-fold integer programming. 
For vertex deletion to clique, we reduce the problem to an ILP with a bounded number of variables and apply Lenstra's theorem.
A central difficulty is that counting the edges induced by a part generally leads to quadratic constraints.
In the vertex-deletion-to-a-clique setting, we overcome this difficulty by showing that it suffices to guess certain variables from a bounded set, after which the relevant edge-counting constraints become linear. 
This argument relies crucially on the fact that the graph outside the deletion set is a single clique and does not extend directly to neighborhood diversity. Consequently, the parameterized complexity with respect to neighborhood diversity remains open.
We further obtain several additional tractability results for \textsc{ECGP}.
In particular, we show (Theorem~\ref{thm:vds-plus-u}) that the problem is FPT when parameterized by cluster vertex deletion set plus $\gamma$, maximum leaf number,
and also by  vertex deletion to paths plus $\gamma$ (see Theorem~\ref{thm:vdp-plus-u}).

We give an algorithm that solves \textsc{ECGP} in time $\mathcal{O}(n\cdot r^{k+1}\cdot (\gamma+1)^{2r})$ on a given nice tree decomposition of width $k$ (Theorem~\ref{thm:tw-dp-general}). 
Combined with our kernelization results, this yields algorithms running in $(r^{2r\gamma}(\gamma+1)^{2r})\cdot n^{\mathcal{O}(1)}$ time for \textsc{ECGP} and  an algorithm running in $(r^{2r\gamma}(\gamma+1)^{2r}(r\gamma^2+r^2+1)^{2r})\cdot n^{\mathcal{O}(1)}$ time for \textsc{BECGP} (Corollary~\ref{thm:fpt-ru-via-dp}).

On the negative side, we show that \textsc{ECGP} remains W[1]-hard even when parameterized by $r$ together with several structural parameters,
including vertex deletion to stars or paths, feedback edge set, modular width, and cluster vertex deletion set
(see Corollaries~\ref{cor:vds-hard} and \ref{cor:vdp-hard} and Theorem~\ref{thm:k-cvd-hard}).
Using a parameter-preserving reduction, these hardness results extend to \textsc{BECGP} as well (see Corollary~\ref{thm:fpt-ru-via-dp}). We also show that both the problems remain W[1]-hard when parameterized by cliquewidth even when $\gamma=3$ (see Theorem~\ref{thm:cw-hard-gamma-three}).
Our results leave open the intriguing question of whether \textsc{ECGP} and \textsc{BECGP} admit an FPT algorithm when parameterized by treewidth plus $\gamma$ and neighborhood diversity. Please refer to Figure~\ref{fig:uccp-overview} for overview of the results obtained for \textsc{ECGP} and \textsc{BECGP} as well as open cases.

Finally, we extend our study to signed graphs, where edges may represent both positive and negative  interactions.
In this setting, we show that the \textsc{SECGP} (see Theorem~\ref{thm:NP-hard signed ECGP}) and \textsc{SBECGP}
(see Corollary~\ref{cor:NP-hard signed BECGP}) are NP-hard even when $r+\gamma=3$ and input graph is disjoint union of two cliques.
We strengthen hardness results for the \textsc{SBECGP}, by showing W[1]-hardness when parameterized by $\mathsf{td}+r+\gamma$, even when $\gamma=0$ (see Corollary~\ref{cor:signed-buccp-td}).

\tikzset{
cell/.style={
    rectangle,
    draw,
    align=center,
    minimum height=6mm,
    minimum width=9mm,
    inner sep=1.5pt
},
fpt/.style={cell, draw=green!60!black, fill=green!20, pattern=crosshatch, pattern color=green!40},
whard/.style={cell, draw=yellow!80!black, fill=yellow!20, pattern=north east lines, pattern color=yellow!60},
nphard/.style={cell, draw=red!70!black, fill=red!20, pattern=dots, pattern color=red!40},
open/.style={cell, draw=black, fill=gray!10}, 
xpwhard/.style={cell, draw=orange!80!black, fill=orange!25, pattern=grid, pattern color=orange!40},
single/.style={
    rectangle,
    draw,
    align=center,
    minimum height=6mm,
    minimum width=13mm,
    inner sep=1.5pt
},
result/.style={line width=1.2pt}
}

\begin{figure}[h]
\centering
\begin{tikzpicture}[>=stealth, scale=1.2]   

\node[fpt,result] (ml1) at (-5,-3) {$\mln$};

\node[fpt] (vc1) at (-1,-3) {$\vc$};

\node[open] (nd1) at (1.5,-1) {$\nd$};

\node[nphard,result]  (mw1a) at (-0.3,3) {$\mw$};
\node[whard,result, right=0pt of mw1a] (mw1b) {$+r$};
\node[open, right=0pt of mw1b] (mw1c) {$+\gamma$};

\node[nphard]  (cw1a) at (-1.5,5) {$\cw$};
\node[whard, right=0pt of cw1a] (cw1b) {$+r$};
\node[whard, result, right=0pt of cw1b] (cw1c) {$+\gamma$};

\node[nphard]  (pw1a) at (-1.5,2) {$\pw$};
\node[xpwhard, right=0pt of pw1a] (pw1b) {$+r$};
\node[open, right=0pt of pw1b]    (pw1c) {$+\gamma$};

\node[nphard]  (fvs1a) at (-4.7,1) {$\fvs$};
\node[xpwhard, right=0pt of fvs1a] (fvs1b) {$+r$};
\node[open, right=0pt of fvs1b]    (fvs1c) {$+\gamma$};

\node[nphard]  (tw1a) at (-4.1,3.5) {$\tw$};
\node[xpwhard,result, right=0pt of tw1a] (tw1b) {$+r$};
\node[open, right=0pt of tw1b]    (tw1c) {$+\gamma$};

\node[whard]    (cvd1a) at (2,1.5) {$\cvd$};
\node[whard,result, right=0pt of cvd1a] (cvd1b) {$+r$};
\node[fpt,result, right=0pt of cvd1b]     (cvd1c) {$+\gamma$};

\node[nphard]    (td1a) at (-2.1,1) {$\td$};
\node[xpwhard, right=0pt of td1a] (td1b) {$+r$};
\node[open, right=0pt of td1b]    (td1c)  {$+\gamma$};

\node[nphard,result]  (vdp1a) at (-4.4,-1.5) {$\vdp$};
\node[xpwhard,result, right=0pt of vdp1a] (vdp1b) {$+r$};
\node[fpt, result, right=0pt of vdp1b]    (vdp1c) {$+\gamma^{*}$};

\node[nphard,result]  (fes1a) at (-5.9,-0.5) {$\fes$};
\node[xpwhard,result, right=0pt of fes1a] (fes1b) {$+r$};
\node[open,right=0pt of fes1b]    (fes1c) {$+\gamma$};

\node[nphard,result]  (vds1a) at (-2.9,-0.5) {$\vds$};
\node[xpwhard,result, right=0pt of vds1a] (vds1b) {$+r$};
\node[fpt,result, right=0pt of vds1b]     (vds1c) {$+\gamma^{*}$};

\node[fpt,result] (vi1) at (0,-0.5) {$\vi$};

\node[fpt,result] (vdc1) at (2,-3) {$\vdc$};

\draw[->] (ml1) -- (fes1b);
\draw[->] (ml1) -- (vdp1b);
\draw[->] (fes1b) -- (fvs1b);
\draw[->] (vi1) -- (td1b);
\draw[->] (vc1) -- (vi1);
\draw[->] (nd1) -- (mw1b);
\draw[->] (mw1b) -- (cw1b);
\draw[->] (fvs1b) -- (tw1b);
\draw[->] (pw1b) -- (tw1b);
\draw[->] (tw1b) -- (cw1b);

\draw[->, bend right=25] (cvd1b) to (cw1c);
\draw[->, bend right=5] (vc1) to (cvd1a);
\draw[->] (td1b) -- (pw1b);

\draw[->] (vc1) -- (nd1);
\draw[->] (vc1) -- (vds1b);
\draw[->] (vc1) -- (vdp1b);
\draw[->] (vdc1) -- (nd1);
\draw[->] (vdc1) -- (cvd1b);
\draw[->] (vdp1b) -- (fvs1b);
\draw[->] (vds1b) -- (td1b);
\draw[->] (vds1b) -- (fvs1b);

\end{tikzpicture}

\caption{All results shown apply to both \textsc{ECGP} and \textsc{BECGP}, except the cell marked $^{*}$, which is currently known only for \textsc{ECGP}. For selected structural parameters, the three adjacent boxes indicate the parameter alone, the parameter together with $r$, and the parameter together with $\gamma$, respectively. 
For the full name of the parameters please refer to \Cref{def:params}.
\legendbox{fpt} indicates FPT, \legendbox{whard} indicates W[1]-hardness, \legendbox{xpwhard} indicates W[1]-hardness together with XP membership, \legendbox{nphard} indicates NP-hardness, and \legendbox{open} indicates open cases. Cells with bold frame indicate a theorem result.}
\label{fig:uccp-overview}
\end{figure}
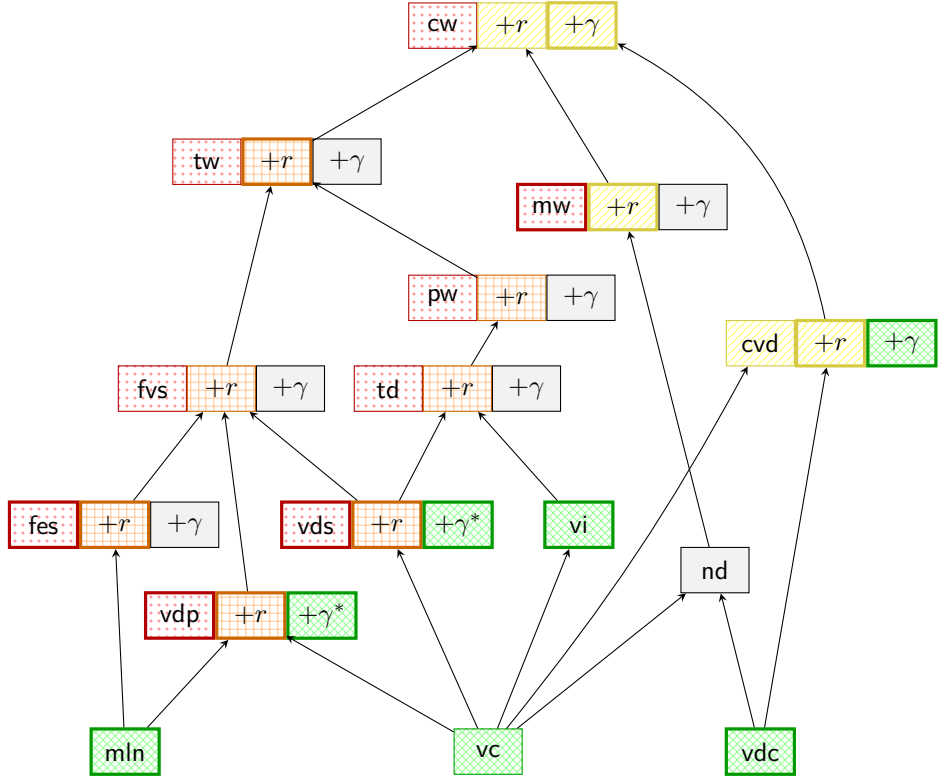

\noindent \textbf{Related Work.}\\
\noindent\textbf{Hedonic Games and Coalition Formation.}
Coalition formation has been extensively studied in the framework of hedonic games, where each agent's utility depends solely on the members of its coalition \cite{69d4f077-72de-3815-b927-32755a2b5480, banerjee2001core, BOGOMOLNAIA2002201}. In this model, agents have preferences over coalitions, and the goal is to compute partitions that are stable under deviations. A central line of work focuses on stability notions such as core, Nash, and individual stability \cite{Aziz_Savani_2016,10.1007/978-3-642-35843-2_4,CECHLAROVA2004333, SUNG2007155}, capturing both group and unilateral deviations.

From a computational perspective, numerous works study the complexity of deciding and computing stable partitions, revealing both polynomial-time algorithms and hardness results depending on the preference domain and structural restrictions \cite{BALLESTER20041,FANELLI2025104394,10.1007/978-3-642-16170-4_16, AZIZ2013316,Peters_2016}.

Our work differs fundamentally from this line of research. Rather than studying stability or preference satisfaction, we consider \emph{feasibility-based} coalition formation, where each coalition must satisfy a prescribed local utility constraint. This shifts the focus from equilibrium concepts to structural and combinatorial properties of feasible partitions.

\smallskip
\noindent\textbf{Graph-Restricted Coalition Formation.}
In many applications, feasible coalitions are constrained by an underlying network structure.
This setting was formalized in graph-restricted games \cite{7e72472a-d8cc-392b-a7e5-767d9d2180f5}, where only connected subsets of agents can form coalitions.
Subsequent work has explored coalition formation under additional structural restrictions,
including networks, hierarchical organizations, and preference-based constraints \cite{demange2004group, DEMANGE199445, jackson2001network, igarashi2016hedonicgamesgraphrestrictedcommunication}.
These models capture limitations on coalition formation arising from communication, coordination, or interaction constraints among agents.
In contrast, our model does not restrict coalitions based on connectivity, but instead enforces \emph{edge constraints} on each coalition, measured via induced subgraphs.
This leads to a different set of computational challenges.

\smallskip
\noindent\textbf{Related Clustering and Graph Partitioning Problems.}
Graph partitioning and clustering problems have been widely studied, including classical formulations that partition vertices into $k$ parts under structural or optimization constraints \cite{GDOWNEY2003209, Demaine2003FixedParameterAF}.
A rich body of work considers clustering under various objectives and parameterizations \cite{10.1007/978-3-319-08783-2_24, pmlr-v162-ganian22a, Ganian_Kanj_Ordyniak_Szeider_2020, SCHAEFFER200727}.
In particular, correlation clustering and related models study clustering in signed graphs, where the goal is to maximize agreements or minimize disagreements between vertices \cite{bansal2004correlation, CHARIKAR2005360,SHAMIR2004173, 10.1145/1411509.1411513}.
These approaches typically optimize a \emph{global objective function}.

\smallskip
\noindent\textbf{Parameterized Complexity of Graph Partitioning.}
Parameterized complexity has been successfully applied to a wide range of graph partitioning and clustering problems,
yielding both fixed-parameter tractability and hardness results under structural parameters such as treewidth, vertex cover, and deletion distance \cite{marekcygan, Demaine2003FixedParameterAF, Ganian_Kanj_Ordyniak_Szeider_2020,
10.1007/978-3-030-67731-2_23,10.1007/978-3-030-64843-5_6,10.1007/978-3-642-11269-0_10,GAIKWAD2026103820,gaikwad2026hardnesstractabilityth1freeedge,doi:10.1137/110855247,10.1609/aaai.v37i5.25771,blazej_et_al:LIPIcs.MFCS.2024.29}.
These approaches often exploit structural decompositions or parameterized frameworks to identify tractable cases, while also establishing tight lower bounds for more general settings.
Our work builds on this line of research by studying partitioning under local constraints.

\section{Preliminaries}
The set of integers from $1$ to $n$ is denoted as $[n]$ and $[n]_0 = [n] \cup \{0\}$. 
We consider simple undirected graphs $G=(V,E)$ with $n=|V|$ vertices and $m=|E|$ edges.
For a set $S \subseteq V$, we denote by $G[S]$ the subgraph induced by $S$, and by $E(G[S])$ its edge set.
We write $|E(G[S])|$ for the number of edges induced by $S$. 
For two disjoint vertex sets $A,B\subseteq V(G)$, we denote by  $E_G(A,B):=\{uv\in E(G):u\in A \text{ and } v\in B\}$  the set of edges having one endpoint in $A$ and the other in $B$. When the underlying graph is clear from the context, we write $E(A,B)$ instead of $E_G(A,B)$
A partition of $V$ into $r$ parts is a collection of pairwise disjoint sets $V_1,\dots,V_r$ such that $\bigcup_{i=1}^r V_i = V$.

\medskip

\noindent
\fbox{%
\parbox{0.97\linewidth}{%
\textsc{Edge-Constrained Graph Partitioning Problem (ECGP)}\\
\textbf{Input:} An undirected graph $G=(V,E)$ and integers $r,\gamma$.\\
\textbf{Question:} Does there exist a partition of $V$ into $r$ sets $V_1,\dots,V_r$ such that
$|E(G[V_i])| \ge \gamma$ for all $i \in [r]$?
}%
}

\medskip

\noindent
\textbf{Balanced variant.}
In the \textsc{Balanced Edge-Constrained Graph Partitioning Problem} (\textsc{BECGP}), we are additionally given that $r$ divides $n$, and require that each part has equal size, i.e., $|V_i| = n/r$ for all $i \in [r]$.

\medskip
A signed graph is a graph $G=(V,E)$ together with a labeling $\sigma : E \to \{+,-\}$ that assigns each edge to be either positive (green) or negative (red).
For a set $S \subseteq V$, let $E^+(S)$ and $E^-(S)$ denote the sets of positive and negative edges induced by $S$, respectively.

\medskip

\noindent
\fbox{%
\parbox{0.97\linewidth}{%
\textsc{Signed Edge-Constrained Graph Partitioning Problem (Signed ECGP)}\\
\textbf{Input:} A signed graph $G=(V,E,\sigma)$ and integers $r,\gamma$.\\
\textbf{Question:} Does there exist a partition of $V$ into $r$ sets $V_1,\dots,V_r$ such that
$|E^+(V_i)| - |E^-(V_i)| \ge \gamma$ for all $i \in [r]$?
}%
}

\medskip

\noindent
\textbf{Balanced signed variant.}
The balanced version (\textsc{Signed BECGP}) additionally requires that $|V_i| = n/r$ for all $i \in [r]$.

\noindent \textbf{Structural Parameters.}
We now introduce the structural graph parameters.

\begin{definition}\label{defvc}
A set $S \subseteq V(G)$ is a \emph{vertex cover} of $G$ if every edge in $E(G)$ has at least one endpoint in $S$. The size of a smallest vertex cover of $G$ is called the \emph{vertex cover number}.
\end{definition}

\begin{definition}
A \emph{feedback vertex set} (resp.\ \emph{feedback edge set}) of a graph $G$ is a set of vertices (resp.\ edges) whose removal results in a forest. The minimum size of such a set is called the \emph{feedback vertex set number} (resp.\ \emph{feedback edge set number}).
\end{definition}

A rooted forest is a disjoint union of rooted trees. Given a rooted forest $Y$, its \emph{closure} is the graph $H$ with $V(H)=V(Y)$, where two distinct vertices are adjacent if and only if one is an ancestor of the other in $Y$.

\begin{definition}\label{deftd}
The \emph{treedepth} of a graph $G$ is the minimum height of a rooted forest $Y$ whose closure contains $G$ as a subgraph. It is denoted by $td(G)$.
\end{definition} 

We now recall the notion of tree decompositions.

\begin{definition}
A \emph{tree decomposition} of a graph $G=(V,E)$ is a tree $T$ together with a family of subsets $(X_t)_{t \in V(T)}$ of $V$ (called \emph{bags}) such that $\bigcup_{t \in V(T)} X_t = V$ and the following conditions hold:
(1) for every edge $uv \in E(G)$, there exists $t \in V(T)$ with $\{u,v\} \subseteq X_t$, and
(2) for every $v \in V$, the set of nodes $\{t \in V(T) : v \in X_t\}$ induces a connected subtree of $T$.
\end{definition}

\begin{definition}\label{deftw}
The \emph{width} of a tree decomposition is $\max_{t \in V(T)} |X_t| - 1$. The \emph{treewidth} of $G$, denoted $tw(G)$, is the minimum width over all tree decompositions of $G$.
\end{definition}

\begin{definition}\label{defpw}\rm
    If the tree $T$ of a tree decomposition is a path, then we say that the tree decomposition 
    is a {\it path decomposition}. The \emph{pathwidth}  ${\pw}(G)$ of a graph $G$  is the  minimum width among all possible path decompositions of $G$.
\end{definition}

\begin{definition}\label{cvddef}
The \emph{cluster vertex deletion number} of a graph $G$ is the minimum number of vertices whose removal results in a disjoint union of complete graphs.
\end{definition}

\begin{definition}
    The clique-width of a graph $G$, denoted by ${\cw}(G)$, is the minimum number of labels needed 
to construct $G$ 
using the following four operations:
\begin{enumerate}
    \item Create a new graph with a single vertex $v$ with label $i$.
    \item Take the disjoint union of two labelled graphs $G_1$ and $G_2$.
    \item Add an edge between every vertex with label $i$ and every vertex with label~$j$,
$i\neq j$.
\item Relabel every vertex with label $i$ to have label $j$.
\end{enumerate}
\end{definition}

\begin{definition}
    The \emph{vertex integrity number} of a graph $G$ is the minimal integer $\vi$ such that there exist a set of vertices $X$ of size at most $\vi$ that if removed, each connected component of $G\setminus X$ is of size at most $\vi$. 
\end{definition}

\begin{definition}
    The \emph{maximum leaf number} of a graph $G$ is the maximum number of leafs in a spanning tree of $G$. It is denoted as $\mln(G)$.
\end{definition}

\begin{definition}\label{def:params}
Let $\mathcal{G}$ be a graph class. The \emph{vertex deletion distance} of a graph $G$ to $\mathcal{G}$ is the minimum number of vertices whose removal transforms $G$ into a graph in $\mathcal{G}$.
Throughout the paper, we use the following abbreviations for graph parameters:
$\vc$ (vertex cover),
$\fvs$ (feedback vertex set),
$\fes$ (feedback edges set),
$\tw$ (treewidth),
$\pw$ (pathwidth),
$\td$ (treedepth),
$\cw$ (clique-width),
$\mw$ (modular width),
$\nd$ (neighborhood diversity),
$\vi$ (vertex integrity),
$\vdc$ (vertex deletion distance to a clique),
$\cvd$ (cluster vertex deletion),
$\vdp$ (vertex deletion distance to disjoint unions of paths), 
$\mln$ (maximum leaf number) and
$\vds$ (vertex deletion distance to disjoint unions of stars).
\end{definition}

\section{Basic Complexity}
We begin with simple observations on the classical complexity of both \textsc{BECGP} and \textsc{ECGP}.

\begin{theorem}
For $r=1$, both \textsc{BECGP} and \textsc{ECGP} can be solved in polynomial time.
\end{theorem}

\begin{proof}
If $r=1$, then the unique part is $V(G)$ itself. Hence the instance is a YES-instance if and only if
$|E(G)| \ge \gamma$,
which can be checked in polynomial time.
\end{proof}

\begin{theorem}\label{thm:np-hard 1}
\textsc{BECGP} is NP-hard for $r=2$.
\end{theorem}

\begin{proof}
We reduce from \textsc{Minimum Bisection} on cubic graphs. The input is a cubic graph $G=(V,E)$ with $|V|=n$ even and an integer $k$, and the question is whether $V$ can be partitioned into two sets $A,B$ such that $|A|=|B|=n/2$ and $|E(A,B)| \le k$.
Let $m:=|E|$. We construct the \textsc{BECGP} instance $(G,2,\gamma)$ with
$\gamma := \left\lceil \frac{m-k}{2} \right\rceil$.

\medskip
\noindent
\emph{Forward direction.}
Let $(A,B)$ be a bisection with $|E(A,B)|=c \le k$. Since $G$ is cubic and $|A|=|B|=n/2$, we have
$2|E(G[A])| + c = m$, and hence $|E(G[A])| = \frac{m-c}{2} \ge \frac{m-k}{2}$. Similarly, $|E(G[B])| \ge \frac{m-k}{2}$. Since these values are integers, both are at least $\gamma$. Thus $(A,B)$ is a valid solution.

\medskip
\noindent
\emph{Backward direction.}
Let $(A,B)$ be a valid solution to \textsc{BECGP}. Then $|A|=|B|=n/2$ and $|E(G[A])| \ge \gamma$. Let $c:=|E(A,B)|$. Again, $2|E(G[A])| + c = m$, so
$\frac{m-c}{2} \ge u \ge \frac{m-k}{2}$, which implies $c \le k$. Hence $(A,B)$ is a valid bisection.
This completes the reduction.
\end{proof}

\begin{theorem}
For $\gamma=1$, both \textsc{BECGP} and \textsc{ECGP} can be solved in polynomial time.
\end{theorem}

\begin{proof}
Observe that a partition of $V(G)$ into $r$ parts satisfies $|E(G[V_i])| \ge 1$ for every $i \in [r]$ if and only if each part contains at least one edge.

\medskip
\noindent
\emph{Forward direction.}
Let $(V_1,\dots,V_r)$ be a feasible solution. For each $i \in [r]$, since $|E(G[V_i])| \ge 1$, there exists an edge $e_i \in E(G[V_i])$. As the parts are pairwise disjoint, the edges $e_1,\dots,e_r$ are pairwise vertex-disjoint. Hence they form a matching of size $r$.

\medskip
\noindent
\emph{Backward direction.}
Let $M=\{e_1,\dots,e_r\}$ be a matching of size $r$, where $e_i=\{u_i,v_i\}$ for each $i \in [r]$. We construct a feasible partition.

For \textsc{ECGP}, initialize $r$ parts by setting $V_i := \{u_i,v_i\}$ for all $i \in [r]$. Then assign every remaining vertex of $V(G) \setminus \bigcup_{i=1}^r V_i$ arbitrarily to any part. Since each $V_i$ already contains the edge $e_i$, we have $|E(G[V_i])| \ge 1$ for all $i$.

For \textsc{BECGP}, note that $|V(G)|=n$ is divisible by $r$. After initializing $V_i := \{u_i,v_i\}$, we distribute the remaining $n-2r$ vertices arbitrarily among the $r$ parts so that each part ends up with exactly $n/r$ vertices. This is always possible since $n \ge 2r$ and the remaining vertices can be assigned to satisfy the size constraints. As before, each part contains the edge $e_i$, and hence induces at least one edge.

\medskip

Thus, there exists a feasible solution if and only if $G$ contains a matching of size at least $r$. Since a maximum matching can be computed in polynomial time, the claim follows.
\end{proof}

\begin{theorem}\label{thm:np-hard 2}
Both \textsc{BECGP} and \textsc{ECGP} are NP-hard for $\gamma=2$.
\end{theorem}

\begin{proof}
We reduce from the \emph{$P_3$-Partition} problem, which is NP-complete even on bipartite graphs of maximum degree~3~\cite{Monnot2007ThePP}.
Given a graph $G=(V,E)$ with $|V|$ divisible by $3$, the task is to partition $V$ into triples each inducing a path on three vertices.

We use the same graph and set $\gamma=2$. Since the graph is bipartite, it contains no triangles, and hence any three vertices induce at most two edges. Therefore, a set of three vertices induces at least two edges if and only if it induces a $P_3$.

Thus, a valid partition into parts each inducing at least two edges corresponds exactly to a $P_3$-partition of the graph.
This proves NP-hardness for both \textsc{BECGP} and \textsc{ECGP}.
\end{proof}

\section{Algorithmic Results}

\subsection{Kernelization Algorithm}

We begin with \textsc{ECGP}, where the vertex set is to be partitioned into exactly $r$ parts, each inducing at least $\gamma$ edges.

\begin{theorem}\label{thm:kernel_ru}
\textsc{ECGP} admits a polynomial kernel with $\mathcal{O}(r\gamma^2)$ vertices and $\mathcal{O}((r\gamma)^{2})$  edges when parameterized by $r+\gamma$.
\end{theorem}

Let $(G=(V,E),r,\gamma)$ be an instance of \textsc{ECGP}.
Note that an isolated vertex contributes no edges to any part and can be added to any part without affecting feasibility.

\begin{uccp}\label{isolated vertices}
  If $v \in V$ is isolated, delete $v$.  
\end{uccp}

\noindent Next, we provide a reduction rule for graph with large matching.

\begin{uccp}\label{large matching}
If $G$ contains a matching of size at least $r\gamma$, return a YES-instance.
\end{uccp}

\begin{lemma}\label{lem:large matching}
Reduction Rule~\ref{large matching} is correct.
\end{lemma}
\begin{proof}
Suppose that $G$ contains a matching $M$ of size at least $r\gamma$. Choose any submatching $M' \subseteq M$ of size exactly $r\gamma$. Since the edges of $M'$ are pairwise disjoint, they cover exactly $2r\gamma$ distinct vertices.
Partition the edge set $M'$ into $r$ groups $M_1,\dots,M_r$, each containing exactly $\gamma$ edges. For every $j \in [r]$, let $V_j$ be the set of endpoints of the edges in $M_j$. Then the sets $V_1,\dots,V_r$ are pairwise disjoint, because the edges of $M'$ are pairwise disjoint.
Moreover, for every $j \in [r]$, the graph $G[V_j]$ contains all $\gamma$ edges of $M_j$. Hence $|E(G[V_j])| \ge \gamma$.

Now assign every vertex of $V(G) \setminus \bigcup_{j=1}^r V_j$ arbitrarily to one of the sets $V_1,\dots,V_r$. This yields a partition of $V(G)$ into exactly $r$ parts. Since adding vertices to a part cannot delete edges already present inside that part, each resulting part still induces at least the $\gamma$ edges coming from its corresponding group $M_j$.
Therefore, the resulting partition is a valid solution to \textsc{ECGP}. It follows that every instance satisfying the premise of Reduction Rule~\ref{large matching} is a YES-instance, and hence the rule is correct.
\end{proof}

\noindent Hence, in the reduced instance, every matching has size less than $r\gamma$.
Let $M$ be a maximal matching. Then $|M| < r\gamma$, and the endpoints of $M$ form a vertex cover $C$ with $|C| \le 2r\gamma$. Let $I := V \setminus C$.
Since isolated vertices have been removed, every vertex in $I$ has at least one neighbor in $C$. Thus, the graph induced by edges between $C$ and $I$ is bipartite with no isolated vertices in $I$. Next, we use the following lemma to construct the next reduction rule.

\begin{definition}
Let $G$ be a bipartite graph with vertex bipartition $(A,B)$. 
  For a positive integer $q$, a
set of edges $M\subseteq E(G)$ is called by a $q$-expansion of $A$ into $B$ if
every vertex of $A$ is incident to exactly $q$ edges of $M$;
$M$ saturates exactly $q|A|$ vertices in $B$.
\end{definition}

\begin{lemma}[Expansion Lemma~\cite{10.1145/1721837.1721848}]\label{expansion lemma}
\label{lem:expansion_kernel}
Let $q \ge 1$ and let $G$ be a bipartite graph with bipartition $(A,B)$ such that $|B| \ge q|A|$ and there are no isolated vertices in $B$. Then there exist nonempty sets $X \subseteq A$ and $Y \subseteq B$ such that:
\begin{itemize}
    \item there is a $q$-expansion from $X$ into $Y$,
    \item and $N(Y) \subseteq X$.
\end{itemize}
Furthermore, such sets can be found in polynomial time.
\end{lemma}

\begin{uccp}\label{u-expansion}
    Suppose there exist nonempty sets $X \subseteq C$ and $Y \subseteq I$ such that:
    \begin{itemize}
        \item there is a $\gamma$-expansion from $X$ into $Y$,
        \item and $N(Y) \subseteq X$.
    \end{itemize} 
Then:
\begin{itemize}
    \item if $|X| \ge r$, return a YES-instance;
    \item otherwise replace $(G,r,\gamma)$ by $(G-(X \cup Y),\, r-|X|,\, \gamma)$.
\end{itemize}
\end{uccp}

\begin{lemma}
    Reduction Rule~\ref{u-expansion} is correct.
\end{lemma}
\begin{proof}
We prove both directions.

\medskip
\noindent
\emph{Forward direction.}
Assume that $(G,r,\gamma)$ is a YES-instance, and let $\mathcal{P}$ be a partition of $V(G)$ into $r$ parts, each inducing at least $\gamma$ edges.

Let $B_1,\dots,B_p$ be the parts that intersect $X$, and let $A_1,\dots,A_{r-p}$ be the remaining parts. Since each $B_j$ contains at least one vertex of $X$, we have $p \le |X|$.
We may assume without loss of generality that all vertices of $Y$ are assigned to parts intersecting $X$. Indeed, since $N(Y) \subseteq X$, any vertex $y \in Y$ has all its neighbors inside $X$, and hence contributes edges only to parts containing vertices of $X$. Therefore, moving vertices of $Y$ into parts intersecting $X$ does not decrease the number of edges induced by any parts.

Using the $\gamma$-expansion, we construct $|X|$ parts entirely from $X \cup Y$ as follows. For each $x \in X$, select a set $U_x \subseteq Y$ of $\gamma$ distinct neighbors, disjoint over all $x$. Then each set $V_x := \{x\} \cup U_x$ induces at least $\gamma$ edges and forms a valid part.
Let $R := (\bigcup_{j=1}^p B_j) \setminus (X \cup Y)$ be the set of remaining vertices originally placed in parts intersecting $X$.
If $|X| \ge r$, then the sets $V_x$ for any $r$ choices of $x \in X$ already give a valid solution, and the remaining vertices can be assigned arbitrarily. Thus $(G,r,\gamma)$ is a YES-instance, and the rule correctly returns YES.
Otherwise, $|X| < r$. Then $r-p > 0$, and hence there exists at least one part among $A_1,\dots,A_{r-p}$. Distribute the vertices of $R$ arbitrarily among these parts. Since each of these parts already induces at least $\gamma$ edges, their validity is preserved.
Thus, we obtain a partition of $V(G) \setminus (X \cup Y)$ into $r-p$ parts, each inducing at least $\gamma$ edges. Since $r-p \ge r-|X|$, we can merge parts if necessary to obtain exactly $r-|X|$ parts, each still inducing at least $\gamma$ edges.
Therefore, $(G-(X \cup Y), r-|X|, \gamma)$ is a YES-instance.

\medskip
\noindent
\emph{Reverse direction.}
Assume that $(G-(X \cup Y), r-|X|, \gamma)$ is a YES-instance. Then there exists a partition of $V(G) \setminus (X \cup Y)$ into $r-|X|$ parts, each inducing at least $\gamma$ edges.
Since there is a $\gamma$-expansion from $X$ into $Y$, for every $x \in X$ we can select a set $U_x \subseteq Y$ of exactly $\gamma$ distinct neighbors of $x$, such that the sets $U_x$ are pairwise disjoint over all $x \in X$.
For each $x \in X$, consider the set $V_x := \{x\} \cup U_x$. Since every vertex in $U_x$ is adjacent to $x$, the induced subgraph $G[V_x]$ contains at least $\gamma$ edges, and hence $V_x$ forms a valid part.
Assign the remaining vertices of $Y$ arbitrarily to these $|X|$ parts. Since adding vertices cannot remove edges, each such part still induces at least $\gamma$ edges.
Together with the $r-|X|$ parts from $G-(X \cup Y)$, this yields a partition of $V(G)$ into exactly $r$ parts, each inducing at least $\gamma$ edges. Hence $(G,r,\gamma)$ is a YES-instance.
This completes the proof.
\end{proof}

\begin{proof}[Proof of Theorem~\ref{thm:kernel_ru}]
Given an instance $(G=(V,E),r,\gamma)$ of \textsc{ECGP}, we apply Reduction Rules~\ref{isolated vertices}, \ref{large matching}, and~\ref{u-expansion} exhaustively. By the correctness of these rules, every application preserves equivalence. Hence the final reduced instance is equivalent to the original one. If at some point Reduction Rule~\ref{large matching} applies, then we immediately return a YES-instance. Likewise, if Reduction Rule~\ref{u-expansion} applies with $|X| \ge r$, then we again return a YES-instance. Thus, we may assume that neither of these situations occurs.

Assume that Reduction Rule~\ref{isolated vertices}, \ref{large matching} and \ref{u-expansion} are applied exhaustively. Let $M$ be a maximal matching in the reduced graph. Since Reduction Rule~\ref{large matching} is no longer applicable, we have $|M| < r\gamma$. Let $C$ be the set of endpoints of the edges of $M$, and let $I := V(G)\setminus C$. Since $M$ is maximal, $C$ is a vertex cover of $G$, and hence $I$ is an independent set. Moreover, $|C| \le 2|M| < 2r\gamma$.

As Reduction Rule~\ref{isolated vertices} has been applied exhaustively, every vertex of $I$ has at least one neighbor in $C$. Consider the bipartite graph induced by edges between $C$ and $I$.
If $|I| \ge \gamma|C|$, then by Lemma~\ref{expansion lemma} there exist nonempty sets $X \subseteq C$ and $Y \subseteq I$ such that there is a $\gamma$-expansion from $X$ into $Y$ and $N(Y)\subseteq X$. This would make Reduction Rule~\ref{u-expansion} applicable, contradicting that the instance is reduced. Hence $|I| < \gamma|C|$.
Therefore, $|I| < \gamma|C| < 2r\gamma^2$, and thus
$|V(G)| = |C| + |I| < 2r\gamma + 2ru^2 = \mathcal{O}(r\gamma^2)$.

We now bound the number of edges. Since $I$ is an independent set, every edge of $G$ has at least one endpoint in $C$. Thus,
$|E(G)| \le |E(G[C])| + |E(C,I)|$.
We have $|E(G[C])| \le \binom{|C|}{2} = \mathcal{O}((r\gamma)^2)$. Moreover, since every vertex in $I$ has all its neighbors in $C$, we have
$|E(C,I)| \le |C|\cdot|I| < |C| \cdot \gamma|C| = \mathcal{O}((r\gamma)^2)$.
Therefore, $|E(G)| = \mathcal{O}((r\gamma)^2)$
\end{proof}

\begin{theorem} \label{thm:becgp-kernel} 
\textsc{BECGP} parameterized by $r+\gamma$ admits a polynomial kernel with $\mathcal{O}(r\gamma^2+r^2)$ vertices. 
\end{theorem}
\begin{proof} Let $(G,r,\gamma)$ be an instance of \textsc{BECGP}, where $G=(V,E)$ and $n:=|V|$. Recall that in \textsc{BECGP} we ask whether $V$ can be partitioned into exactly $r$ parts $V_1,\ldots,V_r$ such that every part has size exactly $n/r$ and induces at least $\gamma$ edges. \\

\noindent \textbf{Step 1: Preliminary reduction rules and a bounded vertex cover.} We begin with an elementary reduction rule. 

\begin{buccp}\label{BUCCP large matching}
 If $G$ contains a matching of size at least $r\gamma$, then return a trivial YES-instance.   
\end{buccp}

\begin{lemma}
    Reduction Rule \ref{BUCCP large matching} is correct.
\end{lemma}
\begin{proof}
    We now prove the correctness of this rule. 
    Suppose that $G$ contains a matching $M$ of size at least $ru$.
    Choose a submatching $M'\subseteq M$ of size exactly $r\gamma$.
    Since $M'$ is a matching, its edges are pairwise vertex-disjoint.
    Partition the edge set $M'$ into $r$ groups  $M_1,\ldots,M_r$  such that each $M_i$ contains exactly $u$ edges. 
    For every $i\in[r]$, let $S_i$ be the set of endpoints of the edges in $M_i$. 
    Then  $|S_i| = 2\gamma$  and $|E(G[S_i])| \ge \gamma$. Therefore each set $S_i$ can be extended to a set $V_i$ of size exactly $B$ by adding arbitrary unused vertices. Since $rB=n$, all remaining vertices can be distributed among $V_1,\ldots,V_r$ so that every final set has size exactly $B$. Adding vertices to a part cannot destroy edges already induced by that part. Thus each resulting part $V_i$ still satisfies $|E(G[V_i])| \ge \gamma$. Hence we obtain a balanced partition of $V(G)$ into $r$ parts, each of size $B$ and each inducing at least $\gamma$ edges. Therefore $(G,r,\gamma)$ is a YES-instance, and Reduction Rule \ref{BUCCP large matching} is correct.
\end{proof}

After applying the Reduction Rule~\ref{BUCCP large matching}, we may assume that $G$ contains no matching of size $r\gamma$. Let $M$ be a maximal matching in $G$, and let $C$ be the set of endpoints of the edges of $M$. Since $G$ has no matching of size $r\gamma$, we have $|M| < r\gamma$,  and therefore  $|C| = 2|M| < 2r\gamma$.  Moreover, since $M$ is maximal, the set $C$ is a vertex cover of $G$. Hence  $I := V\setminus C$ is an independent set. This completes the first step of the kernelization.

\medskip \noindent \textbf{Step 2: Iterative expansion marking.} We now work with the vertex cover $C$ and the independent set $I:=V\setminus C$ obtained in Step 1, where $|C|<2r\gamma$. Recall that all vertices of $I$ have all their neighbors in $C$. We maintain three sets $A\subseteq C$, $Y^\star\subseteq I$, and $Z^\star\subseteq I$. The set $A$ stores marked vertices of the vertex cover, the set $Y^\star$ stores the independent-set vertices used as expansion witnesses, and $Z^\star$ stores vertices of $I$ that become isolated after the marked vertices are ignored. Initially, $A=Y^\star=Z^\star=\emptyset$. At any point, define the active sets $C_{\mathrm{act}}:=C\setminus A$ and $I_{\mathrm{act}}:=I\setminus (Y^\star\cup Z^\star)$. We consider the bipartite graph induced by the edges between $C_{\mathrm{act}}$ and $I_{\mathrm{act}}$. If there exist nonempty sets $X\subseteq C_{\mathrm{act}}$ and $Y\subseteq I_{\mathrm{act}}$ such that there is a $\gamma$-expansion from $X$ into $Y$ and $N(Y)\cap C_{\mathrm{act}}\subseteq X$, then we mark this expansion block. More precisely, for every $x\in X$, we fix a set $Y_x\subseteq Y$ of exactly $\gamma$ private neighbors of $x$, such that the sets $Y_x$ are pairwise disjoint over all $x\in X$. This is possible by the definition of a $\gamma$-expansion. We may assume, by replacing $Y$ with $\bigcup_{x\in X}Y_x$, that $|Y|=u|X|$. Indeed, the reduced set still witnesses a $\gamma$-expansion from $X$, and all vertices of the reduced set have their active neighborhood contained in $X$. We then update $A:=A\cup X$ and $Y^\star:=Y^\star\cup Y$. After marking $X\cup Y$, we additionally mark the active independent-set vertices whose entire neighborhood is now contained in the marked cover set $A$. That is, we define \[ Z_X := \{v\in I_{\mathrm{act}}\setminus Y : N(v)\subseteq A\}, \] and update $Z^\star:=Z^\star\cup Z_X$. We call the vertices in $Z_X$ pseudo-isolated vertices created by the expansion block $X\cup Y$. We repeat this procedure as long as such a $\gamma$-expansion exists and $|A|<r$. 

\medskip 
\noindent \textbf{Immediate YES case.} If at any point we obtain $|A|\ge r$, then the instance is a YES-instance, unless the whole instance is already bounded by a function of $r+\gamma$. Indeed, suppose first that $B=n/r\ge \gamma+1$. Choose any $r$ vertices $x_1,\ldots,x_r\in A$. For each chosen vertex $x_j$, let $Y_{x_j}$ be its reserved set of $\gamma$ private neighbors, and define $S_j:=\{x_j\}\cup Y_{x_j}$. Then $|S_j|=\gamma+1$ and, since every vertex of $Y_{x_j}$ is adjacent to $x_j$, we have $|E(G[S_j])|\ge \gamma$. Moreover, the sets $S_1,\ldots,S_r$ are pairwise disjoint, because the private neighborhoods chosen for the expansion blocks are pairwise disjoint, and every vertex of $A$ is marked only once. Since $B\ge \gamma+1$, each set $S_j$ can be extended to a set $V_j$ of size exactly $B$ by adding arbitrary unused vertices. As $rB=n$, all remaining vertices can be distributed among $V_1,\ldots,V_r$ so that every set has size exactly $B$. Adding vertices to a part cannot destroy already induced edges. Hence each final part still induces at least $\gamma$ edges. Therefore we obtain a feasible balanced partition, and the instance is a YES-instance. On the other hand, if $B<\gamma+1$, then $n=rB<r(\gamma+1)$, and hence the instance already has fewer than $r(\gamma+1)$ vertices. In this case the instance is already bounded by a polynomial in $r+\gamma$. Thus, from now on, we may assume that the iterative marking procedure stops with $|A|<r$. Since each expansion block contributes exactly $\gamma$ vertices of $I$ for every marked vertex of $C$, we have $|Y^\star|\le u|A|<r\gamma$. In particular, the total number of marked expansion vertices is bounded by $|A|+|Y^\star|<r+r\gamma$. It remains to control the number of pseudo-isolated vertices in $Z^\star$ and the size of the active remainder. This will be done in the next step.

\medskip \noindent \textbf{Step 3: Compressing pseudo-isolated vertices.} We now reduce the number of pseudo-isolated vertices marked during the expansion procedure. Recall that every pseudo-isolated set $Z_X$ created in Step 2 satisfies $Z_X\subseteq I$ and $N(Z_X)\subseteq A$, where $A$ is the set of marked vertices of the vertex cover at the moment $Z_X$ is created. In particular, vertices of $Z_X$ have no neighbors outside the already marked cover vertices. We first observe that if $B\le 2\gamma$, then $n=rB\le 2r\gamma$, and the instance already has at most $2r\gamma$ vertices. Hence, in this case, the instance is already bounded by a polynomial in $r+\gamma$. Thus, for the remainder of this step, assume that $B\ge 2\gamma+1$. 

\begin{buccp}\label{bound pseudo isolated vertices}
 Let $Z_X$ be one of the pseudo-isolated sets created during Step 2. If $|Z_X|\ge r$, then delete arbitrary $r$ vertices from $Z_X$.   
\end{buccp}

\begin{lemma} Reduction Rule~\ref{bound pseudo isolated vertices} is correct. \end{lemma} \begin{proof} We assume $\gamma\ge 1$, since for $\gamma=0$ every instance with $r\mid n$ is trivially a YES-instance. Let $R\subseteq Z_X$ be the set of $r$ vertices deleted by the rule, and let $G':=G-R$. Since $|R|=r$ and $r\mid n$, we have $|V(G')|=n-r$. Thus the new balanced part size is $B':=(n-r)/r=B-1$. Since the rule is applied only when $B\ge 2\gamma+1$, we have $B'\ge 2\gamma$. We prove that $(G,r,\gamma)$ is a YES-instance if and only if $(G',r,\gamma)$ is a YES-instance. 

\medskip \noindent \emph{Reverse direction.} Suppose that $(G',r,\gamma)$ is a YES-instance. Then there is a partition $V(G')=V'_1 \cup \cdots \cup V'_r$ such that $|V'_i|=B'$ and $|E(G'[V'_i])|\ge \gamma$ for every $i\in[r]$. Since $|R|=r$, add one distinct vertex of $R$ to each part $V'_i$. Each resulting part has size $B'+1=B$. Moreover, adding vertices to a part cannot decrease the number of edges induced by that part. Hence every resulting part still induces at least $\gamma$ edges. Therefore $(G,r,\gamma)$ is a YES-instance.

\medskip \noindent \emph{Forward direction.} Suppose that $(G,r,\gamma)$ is a YES-instance, and let $V_1,\ldots,V_r$ be a feasible balanced partition of $G$. Thus $|V_i|=B$ and $|E(G[V_i])|\ge \gamma$ for every $i\in[r]$. We will construct a feasible balanced partition of $G'$. The first goal is to construct pairwise disjoint sets $W_1,\ldots,W_r\subseteq V(G')$ such that $|W_i|\le 2\gamma$ and $|E(G'[W_i])|\ge \gamma$ for every $i\in[r]$. Recall that $R\subseteq Z_X$, and by the definition of pseudo-isolated vertices we have $N(R)\subseteq A$. Also, for every marked cover vertex $x\in A$, the expansion marking step fixed a set $Y_x$ of exactly $\gamma$ private neighbors of $x$. These sets $Y_x$ are pairwise disjoint over all marked vertices $x\in A$, and they are disjoint from every pseudo-isolated set, in particular from $Z_X$. Let \[ \mathcal I_A:=\{i\in[r] : V_i\cap A\neq\emptyset\}. \] For every $i\in\mathcal I_A$, choose one vertex $x_i\in V_i\cap A$ and define $W_i:=\{x_i\}\cup Y_{x_i}$.  Since $|Y_{x_i}|=\gamma$, we have $|W_i|=u+1\le 2\gamma$. Moreover, every vertex of $Y_{x_i}$ is adjacent to $x_i$, so $G[W_i]$ contains at least the $u$ edges between $x_i$ and the vertices of $Y_{x_i}$. Hence $|E(G[W_i])|\ge \gamma$. We also have $W_i\cap R=\emptyset$. Indeed, $x_i\in A\subseteq C$, while $R\subseteq Z_X\subseteq I$, so $x_i\notin R$. Furthermore, $Y_{x_i}\subseteq Y^\star$, and the reserved expansion-witness sets are disjoint from the pseudo-isolated sets; hence $Y_{x_i}\cap R=\emptyset$. Therefore $W_i\subseteq V(G')$ and $|E(G'[W_i])|\ge \gamma$. The sets $W_i$ for $i\in\mathcal I_A$ are pairwise disjoint. The vertices $x_i$ lie in distinct original parts and hence are distinct, and the sets $Y_{x_i}$ are pairwise disjoint by construction. Now let \[ Q:=\bigcup_{i\in\mathcal I_A} W_i. \] We next define $W_j$ for every $j\notin\mathcal I_A$. Fix such an index $j$. Then $V_j\cap A=\emptyset$. We claim that every vertex of $(R\cup Q)\cap V_j$ is isolated inside $G[V_j]$. First, if $v\in R\cap V_j$, then $N(v)\subseteq A$, while $V_j\cap A=\emptyset$. Hence $v$ has no neighbor in $V_j$. Second, if $v\in Q\cap V_j$, then $v$ cannot be one of the vertices $x_i$, since all such vertices lie in $A$ and $V_j\cap A=\emptyset$. Thus $v$ belongs to some reserved set $Y_{x_i}$. By construction of the expansion marking step, every vertex in such a reserved set has all its neighbors in the marked cover set $A$. Again, since $V_j\cap A=\emptyset$, the vertex $v$ has no neighbor in $V_j$. This proves the claim. Since $G[V_j]$ induces at least $\gamma$ edges and the vertices of $(R\cup Q)\cap V_j$ are isolated inside $G[V_j]$, none of these isolated vertices is needed as an endpoint of an induced edge in $G[V_j]$. Therefore we may choose $u$ induced edges in $G[V_j]$ whose endpoints avoid $R\cup Q$. 

Let $W_j$ be the set of endpoints of these $\gamma$ edges. Then $W_j\subseteq V(G')$, $W_j\cap Q=\emptyset$, $|W_j|\le 2\gamma$, and $|E(G'[W_j])|\ge \gamma$. Doing this for every $j\notin\mathcal I_A$, we obtain sets $W_1,\ldots,W_r\subseteq V(G')$. They are pairwise disjoint: the sets for indices in $\mathcal I_A$ are pairwise disjoint by the private-neighbor construction; the sets for indices outside $\mathcal I_A$ are chosen inside distinct original parts; and, by construction, every such set avoids $Q$, the union of the already chosen witness sets. Thus, for every $i\in[r]$, we have $|W_i|\le 2\gamma$ and $|E(G'[W_i])|\ge \gamma$. Since $B'\ge 2\gamma$, each $W_i$ can be extended to a set of size exactly $B'$ by adding arbitrary unused vertices of $G'$. Finally, $|V(G')|=rB'$. Hence all remaining vertices of $G'$ can be distributed among the sets $W_1,\ldots,W_r$ so that each resulting part has size exactly $B'$. Adding vertices cannot destroy the already witnessed $\gamma$ induced edges. Therefore each final part induces at least $\gamma$ edges, and we obtain a feasible balanced partition of $G'$. Hence $(G',r,\gamma)$ is a YES-instance. This proves the correctness of Reduction Rule~\ref{bound pseudo isolated vertices}. 
\end{proof}

We apply Reduction Rule 4 exhaustively to every pseudo-isolated set created during Step 2. After this exhaustive application, every such set has size less than $r$. Since the expansion marking procedure stops with $|A|<r$, there are fewer than $r$ expansion rounds. Hence the total number of pseudo-isolated vertices that remain after applying Reduction Rule 4 is less than $r\cdot r=r^2$. That is, $|Z^\star|<r^2$. The remaining task is to bound the size of the active remainder. This will be done in the next step using the Expansion Lemma.

\medskip \noindent \textbf{Step 4: Bounding the active remainder.} It remains to bound the number of active vertices, that is, the vertices that were neither used in expansion blocks nor marked as pseudo-isolated. Recall that after Step 2 we have sets $A\subseteq C$, $Y^\star\subseteq I$, and $Z^\star\subseteq I$, where $A$ is the set of marked cover vertices, $Y^\star$ is the set of independent-set vertices reserved as expansion witnesses, and $Z^\star$ is the set of pseudo-isolated vertices. Define $C_{\mathrm{act}}:=C\setminus A$ and $I_{\mathrm{act}}:=I\setminus (Y^\star\cup Z^\star)$. We claim that $|I_{\mathrm{act}}|<u|C_{\mathrm{act}}|$. Suppose, for contradiction, that $|I_{\mathrm{act}}|\ge u|C_{\mathrm{act}}|$. Consider the bipartite graph induced by the edges between $C_{\mathrm{act}}$ and $I_{\mathrm{act}}$. By the definition of $Z^\star$, no vertex of $I_{\mathrm{act}}$ has all its neighbors contained in the already marked cover set $A$. Since $C$ is a vertex cover, every vertex of $I$ has all its neighbors in $C$. Hence every vertex of $I_{\mathrm{act}}$ has at least one neighbor in $C_{\mathrm{act}}$. Therefore, the bipartite graph between $C_{\mathrm{act}}$ and $I_{\mathrm{act}}$ has no isolated vertices on the $I_{\mathrm{act}}$ side. Since $|I_{\mathrm{act}}|\ge \gamma|C_{\mathrm{act}}|$, the Expansion Lemma yields nonempty sets $X\subseteq C_{\mathrm{act}}$ and $Y\subseteq I_{\mathrm{act}}$ such that there is a $\gamma$-expansion from $X$ into $Y$ and $N(Y)\cap C_{\mathrm{act}}\subseteq X$. But this means that the iterative expansion marking procedure from Step 2 could have continued, contradicting the fact that it stopped. Thus $|I_{\mathrm{act}}|<\gamma|C_{\mathrm{act}}|$. Since $C_{\mathrm{act}}\subseteq C$ and $|C|<2r\gamma$, we get $|I_{\mathrm{act}}|<\gamma|C_{\mathrm{act}}|\le \gamma|C|<2r\gamma^2$. 

\medskip \noindent \textbf{Step 5: Kernel-size bound.} We now bound the total number of vertices remaining in the instance. The vertex set is partitioned as $V(G)=C\cup I_{\mathrm{act}}\cup Y^\star\cup Z^\star$. We bound each term separately. First, from Step 1, $|C|<2r\gamma$. Second, by Step 4, $|I_{\mathrm{act}}|<2r\gamma^2$. Third, since the expansion marking procedure stopped with $|A|<r$, and each marked cover vertex reserves exactly $\gamma$ private expansion neighbors, we have $|Y^\star|\le \gamma|A|<r\gamma$. Fourth, by exhaustive application of Reduction Rule 4, every pseudo-isolated set has size less than $r$. Since the expansion marking procedure has fewer than $r$ rounds, we have $|Z^\star|<r^2$. Combining the above bounds, we obtain \[ \begin{aligned} 
|V(G)| &\le |C| + |I_{\mathrm{act}}| + |Y^\star| + |Z^\star| \\ &< 2r\gamma + 2r\gamma^2 + r\gamma + r^2 \\ &= 2r\gamma^2 + 3r\gamma + r^2. \end{aligned} \] 
Hence the reduced instance has $\mathcal{O}(r\gamma^2+r^2)$ vertices. All reduction rules are polynomial-time computable and preserve equivalence, except when they correctly return a trivial YES- or NO-instance. Hence this gives a polynomial kernel for \textsc{BECGP} parameterized by $r+\gamma$ with $\mathcal{O}(r\gamma^2+\gamma^2)$ vertices.
This completes the proof of Theorem~\ref{thm:becgp-kernel}.
\end{proof}

\subsection{ILP-Based FPT Algorithms}

Our algorithmic results in this subsection are based on formulating the problem as an ILP with a bounded number of variables or as $N$-fold. 
We start with designing ILPs for this classical result.

\begin{proposition}[\cite{Lenstra83}, \cite{Kannan87}, \cite{FrankTardos87}]
    \label{pro:pilp}
    There is an algorithm that solves an input \textup{ILP} instance $\mathcal{I}$ with $p$ variables in time
    $p^{\mathcal{O}(p)}\cdot |\mathcal{I}|$.
\end{proposition}

\noindent  We apply it in the parameterization by vertex deletion to clique and in the parameterization by maximum leaf number.

\begin{theorem}
\label{thm:uccp-vdc}
\textsc{ECGP} is fixed-parameter tractable when parameterized by $\mathsf{vdc}$.
\end{theorem}

\begin{proof}
Let $(G=(V,E),r,\gamma)$ be an instance of \textsc{ECGP}. Let $X\subseteq V$ be a vertex deletion set to a clique with $|X|=k=\mathsf{vdc}(G)$, and let $C:=V\setminus X$, so that $C$ induces a clique.
Define
$s:=\min\left\{\ell\ge 0:\binom{\ell}{2}\ge \gamma\right\}$.
Observe that every subset of $C$ of size at least $s$ induces at least $\gamma$ edges.

\medskip
\noindent
\emph{Guessing the parts intersecting $X$.}
Let $p\le \min\{r,k\}$ denote the number of parts intersecting $X$. We first guess an ordered partition
$A_1,\ldots,A_p$
of $X$, where $A_j$ is the subset of $X$ contained in the $j$-th part.
Next, partition the clique vertices according to their neighborhood in $X$. For every
$T\subseteq X$, let
$C_T:=\{v\in C:N(v)\cap X=T\}$,
and let $n_T:=|C_T|$. Since $|X|=k$, there are at most $2^k$ such types.

\medskip
\noindent
\emph{A natural ILP formulation.}
For every $j\in[p]$ and every type $T\subseteq X$, introduce a variable
$x_{j,T}\in\mathbb Z_{\ge0}$,
denoting the number of vertices of type $T$ assigned to the part containing $A_j$.
Let
$m_j:=\sum_{T\subseteq X}x_{j,T}$
be the number of clique vertices assigned to the $j$-th part.
The variables satisfy the availability constraints
\[
\sum_{j=1}^{p}x_{j,T}\le n_T
\qquad
\text{for every }T\subseteq X.
\]
Moreover, the $j$-th part induces
\[|E(G[A_j])|
+\binom{m_j}{2}
+\sum_{T\subseteq X}|A_j\cap T|x_{j,T}\]
edges. Thus a natural edge constraint is
\[
|E(G[A_j])|
+\binom{m_j}{2}
+\sum_{T\subseteq X}|A_j\cap T|x_{j,T}
\ge \gamma.
\]
Unfortunately, this is not an ILP, since the term
$\binom{m_j}{2}$ is quadratic in the variables.
\begin{claim}
For every $j\in[p]$, it is sufficient to consider
$m_j\in
\{\max(0,s-|A_j|),\ldots,s\}$.
Consequently, each value $m_j$ can be guessed from at most
$|A_j|+1\le k+1$
possibilities.
\end{claim}
\begin{proof}
Suppose first that
$m_j+|A_j|<s$.
Then, even if every possible edge among the vertices of the part were present,
\[
|E(G[A_j])|
+|A_j|m_j
+\binom{m_j}{2}
\le
\binom{|A_j|+m_j}{2}
<\gamma.
\]
Hence such a part cannot satisfy the edge constraint.
Conversely, suppose a feasible solution contains a part intersecting $X$ with
$m_j>s$.
Since $C$ is a clique, any $s$ of these clique vertices already induce at least $\gamma$ edges. Therefore we may retain any $s$ clique vertices in this part and move the remaining clique vertices to arbitrary other parts. Since adding vertices to a part cannot decrease the number of induced edges, feasibility is preserved. Thus every feasible solution can be transformed into one satisfying
$m_j\le s$.
\end{proof}
We now guess the value of $m_j$ for every $j\in[p]$. After this guess,
$\binom{m_j}{2}$ becomes a constant, and the edge constraints become linear.
Finally, let
$L:=|C|-\sum_{j=1}^{p}m_j$
be the number of clique vertices not assigned to parts intersecting $X$.
The remaining $r-p$ parts lie entirely inside the clique $C$. Such a part satisfies the edge constraint if and only if it contains at least $s$ vertices. Hence these parts exist if and only if
$L\ge (r-p)s$.
Any remaining clique vertices can then be distributed arbitrarily among the existing parts.

\medskip
\noindent
\emph{Running time.}
%The number of ordered partitions of $X$ depends only on $k$. By the claim, each value $m_j$ has at most $k+1$ possibilities, and hence the total number of guesses is bounded by a function of $k$.
For every fixed guess, the resulting ILP contains at most $p2^k\le k2^k$
variables and only linear constraints. Therefore, by using \Cref{pro:pilp}, each ILP can be solved in FPT time. Therefore, we have an FPT algorithm.
\end{proof}

\begin{theorem}
\label{thm:buccp-vdc}
\textsc{BECGP} is fixed-parameter tractable when parameterized by $\mathsf{vdc}$.
\end{theorem}
\begin{proof}
Let $(G=(V,E),r,\gamma)$ be an instance, and let $X \subseteq V$ be a minimum vertex deletion set to a clique, so $|X|=\mathsf{vdc}(G)$ and $C:=V \setminus X$ induces a clique.

\medskip

\noindent
\emph{Step 1: Guess parts intersecting $X$.}
Let $p \le |X|$ be the number of parts intersecting $X$. We guess an ordered partition $A_1,\dots,A_p$ of $X$, where each $A_j$ corresponds to one part.

\medskip

\noindent
\emph{Step 2: Types of clique vertices.}
For every subset $T \subseteq X$, define
\[
C_T := \{ v \in C : N(v)\cap X = T \}, \quad n_T := |C_T|.
\]
There are at most $2^{|X|}$ types.

\medskip

\noindent
\emph{Step 3: Part sizes.}
In \textsc{BECGP}, each part has size exactly $B := |V|/r$ (assume w.l.o.g.\ that $r \mid |V|$).
For each part $A_j$, we have $|A_j| = B$, and hence the number of clique vertices assigned to it is
$m_j := B - |A_j|$.
Note that $m_j$ is completely determined.

\medskip

\noindent
\emph{Step 4: ILP variables.}
For every $j \in [p]$ and every type $T \subseteq X$, introduce a variable
$x_{j,T} \in \mathbb{Z}_{\ge 0}$,
denoting how many vertices of type $T$ are assigned to part $A_j$.

\medskip

\noindent
\emph{Step 5: Constraints.}

\begin{itemize}

\item \textbf{Part size constraints:}
\[
\sum_{T \subseteq X} x_{j,T} = m_j
\quad \text{for all } j \in [p].
\]

\item \textbf{Availability constraints:}
\[
\sum_{j=1}^{p} x_{j,T} \le n_T 
\quad \text{for all } T \subseteq X.
\]

\item \textbf{Edge constraints for each part $A_j$:}

Since $C$ induces a clique and $m_j$ is fixed, the number of edges inside part $j$ must be at least $\gamma$. Therefore, we require:
\[
|E(G[A_j])| + \binom{m_j}{2} + \sum_{T \subseteq X} |A_j \cap T| \cdot x_{j,T} \ge \gamma.
\]
\end{itemize}

\medskip

\noindent
\emph{Step 6: Remaining parts.}
Let
$L := |C| - \sum_{j=1}^p m_j$
be the number of clique vertices not assigned to $X$-parts.
Since each remaining part has size exactly $B$, the remaining $r-p$ parts can be formed if
$L = (r-p)\cdot B$. Also we need ${B \choose 2}\geq \gamma$.
\medskip

\noindent
\emph{Running time.}
The number of variables is at most $|X| \cdot 2^{|X|}$, depending only on $\mathsf{vdc}$. We guess a partition of $X$, which depends only on $|X|$. The resulting ILP has a bounded number of variables and only linear constraints. By Lenstra's algorithm, each instance can be solved in FPT time parameterized by $\mathsf{vdc}$.
\end{proof}

\begin{theorem}
    Both \textsc{ECGP} and \textsc{BECGP} are FPT parameterized by $\mln$. 
\end{theorem}

We use the following characterization of the maximum leaf number. 
\begin{proposition}[\cite{bouland_2011, KleitmanW91, FellowsLMMRS09}]
    Suppose $G$ is a graph with maximum leaf number $\mln$. Then $G$ is a subdivision of a graph $H$ with at most $4\mln-6$ vertices.
\end{proposition}
In this proof, we define $k := 4\mln-6$ as the parameter.
Given an edge $e = \{a,b\} \in E(H)$, we denote the corresponding $ab$-path in $G$ minus its endpoints $a$ and $b$ as the \emph{open branch} $B_e$. 
Thus, $V(B_e)$ consists purely of the internal vertices created by subdivision.

We further define a \emph{nice solution} to be a solution in which 
\begin{itemize}
    \item For every part $V_i$ and branch $e \in E(H)$, $V_i \cap V(B_e)$ forms a single contiguous subpath of $B_e$.
    \item If $V_i \cap V(H) \neq \emptyset$, then for any branch $e = \{u, v\}$ where $u \in V_i \cap V(H)$ and $v \not \in V_i \cap V(H)$, the subpath $V_i \cap V(B_e)$ is contiguous to $u$.
    \item No two parts are present in the same two or more branches of distinct edges $e_1, e_2 \in E(H)$.
\end{itemize}

%\todoC[inline]{formalize a bit more as well}
\begin{lemma}\label{lem:mln_nice_solution}
    If \textsc{ECGP} instance admits a solution $V_1, V_2, \dots V_r$, it also admits a nice solution $V'_1, V'_2, \dots V'_r$ with $|V_i| = |V'_i|$.
\end{lemma}
\begin{proof}
    If $V_i$ is disconnected inside one branch, we can move one connected part to the other such that, it still is a solution.
    Other vertices in the branch will be moved in the opposite direction and no connected part will be split. 
    Similarly if a part of part is not connected to the vertex of $H$ we connect them. 

    Lastly, if two parts share two or more branches, we can take the smallest induced path from the two shared branches and swap it with the same number of vertices from the second part in the second branch. 
    By connecting the parts again, both of the parts have the same or more edges. 
    Applying these two steps exhaustively we have a solution $V'_1, V'_2, \dots, V'_r$.
\end{proof}

A nice solution allows us to bound the number of parts that are present in two or more branches by $\binom{|E(H)|}{2}$.

We first guess which parts occupy the vertices of $H$.
There are at most $k^k$ many assignments.
If a vertex adjacent to such an assigned vertex belongs to the same coalition, the coalition gains one edge.
There are two such edges if the whole branch belongs to the same part. 
We also guess for each branch how many of these extra edges are there.
As we already know the solution on $V(H)$, we know to which parts these edges belong.
If the incident vertices of a branch are from the same part there are three possibilities, if they are from different parts there are 4 possibilities.
This is noted as $\partial_{i,e} \in \{0,1,2\}$ and there are at most $4^{\binom{k}{2}}$ possibilities in total.  

We group parts intersecting multiple branches with parts containing at least one vertex from $V(H)$ and call them \emph{special}.
Let $W_1, W_2, \dots, W_l, l \in \mathcal{O}(k^4)$ be the partition of $V(H)$ into the special parts.
The rest of the parts are contained each inside just one branch, we call such parts \emph{well-behaved}.
The well behaved parts contain just $\gamma$ edges and $\gamma+1$ vertices. 
The extra vertices in each branch will go towards the special parts.

For each special part $i \in [l]$ we further guess inside which branches it has at least one vertex.
We denote it as $F_i \subseteq E(H)$.
There are ${{2^{|E(H)|}}^l} \in (2^{k^2})^{\mathcal{O}(k^4)} = 2^{\mathcal{O}(k^6)}$ such possibilities.
Together, all guesses of $W$, $\partial$, and $F$ add up to $k^k \cdot 4^{\binom{k}{2}} \cdot 2^{\mathcal{O}(k^6)} = 2^{\mathcal{O}(k^6)}$. 
For every guess, we run an ILP and return YES iff at least one ILP outputs YES.

We build a ILP that checks whether each path has the correct number of vertices. 

\noindent
\textbf{Variables}:
\begin{itemize}
    \item $\forall e \in E(H): x_{e}$ - how many well-behaved parts are contained in the $ab$-branch.
    \item $\forall i \in [l], \forall e \in E(H): y_{i,e} \in [\partial_{i, e}, n]$ - how many vertices are in $B_e$
\end{itemize}

\noindent
\textbf{Constraints}:
\begin{enumerate}
    \item \label{thm:mln:constr:length_of_paths} We count each vertex at most once within each branch.
    $$\forall e \in E(H): (\gamma+1)x_{e} + \sum_{i \in [l]} y_{i,e} \le |V(B_e)|$$
    \item \label{thm:mln:constr:edge-count_of_special_parts} To satisfy edge constraint for special parts.
    $$\forall i \in [l]: |E(G[W_i])| + \sum_{e \in F_i} (y_{i,e}-1 + \partial_{i,e}) \ge \gamma$$
    \item \label{thm:mln:constr:extra} If we guessed there to be two extra edges counted for one part within one branch, then all of the vertices inside the branch belong to that part. 
    $$\forall i \in [l], \forall e \in E(H): \partial_{i,e} = 2 \implies y_{i,e} = |V(B_e)|$$
    \item \label{thm:mln:constr:num_parts} There are exactly $r$ parts in total.
    $$ l + \sum_{e \in E(H)} x_e = r$$

\end{enumerate}

\noindent
\emph{The Running time.}
There are $\binom{k}{2} + \binom{k}{2}\cdot \mathcal{O}(k^4) = \mathcal{O}(k^6)$ variables.
Therefore, using \Cref{pro:pilp}, we can solve the ILP in $k^{\mathcal{O}(k^6)}n$.
Together with the guessing part we have a runtime of $k^{\mathcal{O}(k^6)}n$.
As the relationship between $k$ and $\mln$ is linear, the asymptotic complexity is $\mln^{\mathcal{O}(\mln^6)}\cdot n$.

\begin{claim}\label{lem:correctness of mln}
    $\mathcal{I}$ is a YES-instance of \textsc{ECGP} if and only if the algorithm returns YES.
\end{claim}
\begin{claimproof}
    \emph{Forward direction.}
    Let $V_1, V_2, \dots, V_r$ be a nice solution that we get by applying \Cref{lem:mln_nice_solution} to $\mathcal{I}$.

    For each $B_e$ branch we set $x_e$ to the number of parts that are fully contained inside the branch.
    There are at most $k + \binom{k}{2}$ unaccounted parts. 
    For each of them, we set $y_{i, e} := |V'_i \cap B_e|$.
    We set $W_i := V'_i \cap V(H)$ and $\partial_{i,e}$ we set to zero, one or two if at the corresponding branch the part $i$ has respectively 0, at least 1 or all of the vertices in the branch.

    We now have to check the ILP conditions hold.
    In \Cref{thm:mln:constr:length_of_paths} each part is either fully in the branch, and so it contains at least $\gamma+1$ vertices, or it is  counted within the $y$ variables.
    Every special part present inside a branch forms a path on $y_{i,e}$ vertices and one less edge. 
    If the part contains one of the ends of the branch, then the path is connected there and there is one more edge.
    If both of the branch ends are in the part two edges are added only in the case of the whole branch being in the part.
    This is exactly how the $\partial_{i,e}$ was chosen and thus \Cref{thm:mln:constr:edge-count_of_special_parts} holds.
    All \Cref{thm:mln:constr:extra} conditions are satisfied as $\partial_{i,e}$ was set up to 2 if all of the vertices of a branch were in the same part. 
    Lastly \Cref{thm:mln:constr:num_parts} is satisfied as each part was identified as one of the $l$ special parts or it is fully contained inside a $e$-branch and so it is counted in $x_e$.

\emph{Backward direction.} 
    Given an ILP solution, we construct the parts.
    We know the $W_1, \dots, W_l$ and $\partial_{i,e}$.
    We put all the vertices of $W_i$ into $V_i$.
    For each branch we distribute the vertices separately.
    
    First if $\partial_{i,e} = 1$, then one of the endpoints is in $V_i$.
    We put the closest $y_{i, e}$ many vertices into $V_i$.
    There could have been only two such part according to definition of $\partial$.
    If $\partial_{i,e} = 2$, then $V(B_e) \subseteq V_i$.
    The rest of the vertices we split into $x_e$ parts, each being a path with $\gamma+1$ vertices.
    There are enough vertices in the path thanks to \Cref{thm:mln:constr:length_of_paths}.
    Any extra vertices are put into $V_i$ and so $V_1, V_2, \dots, V_r$ is a partition.
    Now it remains to argue, that each special part has utility at least $\gamma$.
    As \Cref{thm:mln:constr:length_of_paths} is satisfied and we managed to put vertices inducing $y_{i,e} + \partial_{i,e}$ edges into $V_i$, then $|V_i| \ge \gamma$ for every $i$.
\end{claimproof}

This gives us an FPT algorithm for \textsc{ECGP} when parameterized by maximum leaf number.
We slightly modify the ILP to get an FPT algorithm for \textsc{BECGP} as well.

We assume that $(\gamma + 1) \le \frac{n}{r}$ otherwise there is no well behaved and we set all $x_{e}$ to $0$.
We add a condition that limits the number of vertices in each special part.
\begin{enumerate}
    \item[5] \label{thm:mln:vertex-count_of_special_parts}
    The vertices of $V(H)$ in the $i$-th part and the vertices in each branch add up to at most~$\frac{n}{r}$.
    $$\forall i \in [l]: |W_i| + \sum_{e \in E(H)} y_{i,e} \le \frac{n}{r}$$
\end{enumerate}  

We did not add any new variables, so the updated algorithm has the same asymptotic complexity as the original.

We have to argue that the extra constraint does not change the correctness.
In the forward direction we use the same assignment of variables.
As $y_{i,e}$ already represents the number of vertices in part $i$ within $e$-branch and no vertex is counted twice in another variable or $|W_i|$, the left side is at most $|V_i| \le \frac{n}{r}$.
Therefore, the constraints in \Cref{thm:mln:vertex-count_of_special_parts} are also satisfied.
For the backward direction, after the vertices are distributed according to the variables, there may be extra unassigned vertices due to \Cref{thm:mln:constr:length_of_paths} being inequality. 
However, so far in each part there are at most $\frac{n}{r}$ vertices due to \Cref{thm:mln:vertex-count_of_special_parts}.
Therefore, the rest of the vertices are distributed to fill each part to exactly $\frac{n}{r}$.

This concludes the FPT algorithms parameterized by the maximum leaf number.

\subsection{FPT via \(N\)-Fold ILP}

In this work, we use the so-called \emph{$N$-fold integer programming} formulation. 
Here, the problem is to minimize a linear objective over a set of linear constraints with a very restricted structure. 
In particular, the constraints are as follows. 
We use $x^{(i)}$ to denote a set of~$t_i$ variables (a so-called \emph{brick}).
\begin{align}
	D_1 x^{(1)} + D_2 x^{(2)} + \cdots + D_N x^{(N)} &= \textbf{b}_0    \label{eq:NFold:linking}  \\
	A_i x^{(i)}                        &= \textbf{b}_i    & \forall i \in [N]       \\
	\textbf{0} \le x^{(i)}                 &\le \textbf{u}_i  & \forall i \in [N]
\end{align}
Where we have $D_i \in \Z^{\rho \times \tau_i}$ and $A_i \in \Z^{\sigma_i \times \tau_i}$
Let us denote $\sigma = \max_{i \in [N]} \sigma_i$, $\tau = \max_{i \in [N]} \tau_i$, and let the dimension be~$d$, i.e., $d = \sum_{i \in [N]} \tau_i \le N\tau$.
Constraints~\eqref{eq:NFold:linking} are the so-called \emph{linking constraints} and the rest are the \emph{local constraints}. In the analysis of our algorithms, we use the following result of Eisenbrand et al.~\cite{EisenbrandHKKLO19}.

\begin{proposition}[{\cite[Corollary~97]{EisenbrandHKKLO19}}]\label{prop:n_fold_algo}
	$N$-fold IP can be solved in $a^{\rho^2s+\rho\sigma^2} \cdot d \cdot \log(d) \cdot L$ time, where
	\begin{itemize}
		\item $L$ is the maximum feasible value of the objective and
		\item $a= \rho\cdot\sigma \cdot \max_{i \in [N]} \left( \max ( \|D_i\|_\infty, \|A_i\|_\infty ) \right)$.
	\end{itemize} 
\end{proposition}

\begin{theorem}
\label{thm:uccp-vin}
    There is an FPT algorithm for \textsc{ECGP} and \textsc{BECGP} parameterized by $\vi$. 
\end{theorem}

The algorithm consists first of brute-forcing all the possible solutions in the modulator and then verifying whether any solution can be extended to the whole graph using $N$-fold ILP.
Let $X \subseteq V$ of size at most $\vi$ be the vertex integrity modulator. 
Let $X = X_1 \cup X_2 \cup \dots \cup X_k$ be the partitioning of $X$ into at most $\vi$ different global parts.
Each connected component of $G\setminus X$ has size of at most $\vi$.

We assign each component $C$ of $G \setminus X$ a \emph{type} $T\in \mathcal{T}$ according to $G[X \cup V(C)]$ -- $C_1$ and $C_2$ have the same type if and only if there is an isomorphism of $G[X \cup V(C_1)]$ and $G[X \cup V(C_2)]$ that acts as an identity on $X$. 
We denote by $n_T$ the total number of components of type $T$ present in $G \setminus X$.
A \emph{pattern} $p \in \mathcal{P}$ extends a type $T$ by specifying a partition of a component's vertices $V(C)$ into the modulator parts and anonymous non-modulator parts.
Formally, $p = (T, (e_1^p, e_2^p, \dots, e_k^p), (m_0^p, m_1^p, \dots, m_{h}^p))$, where $e_i^p = |E(G[X_i \cup V_i^C])| - |E(G[X_i])|$ denotes the number of edges added to the modulator part $X_i$ by this pattern, and $m_t^p$ counts the number of chunks $W_j \in \{W_1, \dots, W_m\}$ whose internally induced edge-count $|E(G[W_j])|$ is exactly $t$.
Note that, $t$ is bounded by $h := \binom{\vi}{2}+\vi^2$ the maximum edges of $G[X\cup C]\setminus E(G[X])$.
For a single type, there are at most $h^k\cdot h^{h} \in \vi^{\mathcal{O}(\vi^2)}$ patterns extending it.

We can now state the ILPs.

\noindent
\textbf{Variables}:
\begin{enumerate}
    \item $x_p$: number of components with pattern $p$.
    \item $y_{t, i}$: number of chunks with edge-count $t$ within the part $i$.
\end{enumerate}

\noindent
\textbf{Constraints}:
\begin{enumerate}
    \item \label{thm:vi:enum:modulator_parts} The parts intersecting the modulator have edge-count of at least $u$ each:
    
        $$\forall i \in [k]: |E(G[X_i])| + \sum_{p \in \mathcal{P}} e_i^p \cdot x_p \ge \gamma$$
    
    \item \label{thm:vi:enum:other_parts} Each other part has edge-count of at least $u$:
    
        $$\forall i \in [k+1, r]: \sum_{t \in [h]_0} t \cdot y_{t,i} \ge \gamma$$
    
    \item \label{thm:vi:enum:type_counts} The number of components with a pattern that extends a type of a component is exactly the same:

        $$\forall T \in \mathcal{T}: \sum_{p \text{ extends } T} x_p = n_T$$
    
    \item \label{thm:vi:enum:chunks_count} The number of chunks across components is the same as it is assigned to the parts:

        $$\forall t \in [h]_0: \sum_{i} y_{t,i} = \sum_{p \in \mathcal{P}} m_t^p x_p$$
    
\end{enumerate}

\begin{claim}\label{lem:running_time_vin}
  The above $N$-fold ILP can be solved in time $2^{2^{\mathcal{O}(\vi^2)}} r \log(r)$.
\end{claim}
\begin{claimproof}
    The number of rows of \Cref{thm:vi:enum:modulator_parts} and \Cref{thm:vi:enum:chunks_count} is at most $\mathcal{O}(\vi)$ and $\mathcal{O}(\vi^2)$ respectively.
    The number of types is also bounded by a function of $\vi$ and that bounds the number of constraints in \Cref{thm:vi:enum:type_counts} 
    All of \Cref{thm:vi:enum:modulator_parts,thm:vi:enum:chunks_count,thm:vi:enum:type_counts} are linking constraints for $N$-fold.
    However, there can be $r$ many constraints in \Cref{thm:vi:enum:other_parts}. 
    Fortunately, each row of \Cref{thm:vi:enum:other_parts} forms a brick in the matrix as each variable $y_{t,i}$ appears only in the row of part $i$ and in the linking constraint.
    The last structural property of the matrix we have to check is the size of the coefficients. 
    All of the coefficients in the matrix are bounded by the number of edges in a graph with $\mathcal{O}(\vi)$ vertices.

    With this distinction we can now calculate the running time of the $N$-fold ILP using \Cref{prop:n_fold_algo}.
    The number of types in $\mathcal{T}$ is at most $2^{2\vi^2}$. Note that 
    $\rho \in 2^{\mathcal{O}(\vi^2)}$,
    $\sigma=1$, 
    $a \in \rho \cdot \sigma \cdot \vi^2$, and 
    $d = |\mathcal{P}|+ r|\mathcal{T}| \in ({\mathcal{O}(\vi^2)}\cdot \vi^\mathcal{\vi^2})\cdot(r {\mathcal{O}(\vi^2)}) = r \vi^\mathcal{\vi^2}$.
    Therefore, the total running time of a single ILP instance is
    $2^{2^{\mathcal{O}(\vi^2)}} r \log(r)$. 
\end{claimproof}

In the algorithm there are run at most $\vi^\vi$ ILP instances and so the total running time is $2^{2^{\mathcal{O}(\vi^2)}} r \log(r) + n$.

\begin{claim}\label{clm:N-fold_ILP_vin}
       \textsc{ECGP} admits a solution if and only if at least one $N$-fold is feasible.
\end{claim}
\begin{proof}
    \emph{Forward direction.}
    Assume there exists a valid partition $V = V_1 \cup V_2 \dots \cup V_r$ such that $|E(G[V_i])| \ge \gamma$ for all $i \in [r]$.
    The algorithm brute-forces all partitions of $X$. Thus, one iteration will exactly match the intersection of the true partition with $X$.
    Each part induces patterns in components of $C$ of $G \setminus X$.
    We set $x_p$ to be the exact number of components in $G \setminus X$ that exhibit pattern $p$.
    Any part of a component $C$ assigned to a part $V_j$ ($j > k$) that doesn't touch the modulator acts as an anonymous chunk. 
    If the edge-count of this chunk is exactly $t$, it contributes to the count of $y_{t,j}$.

    Now we can check that the constraints are satisfied. 
    For $i \in [k]$, $|E(G[V_i])|$ consists of internal edges in $X_i$, internal edges within the component parts assigned to $i$, and edges between $X_i$ and those parts. 
    This is precisely counted in \Cref{thm:vi:enum:modulator_parts} by $|E(G[X_i])| + \sum_{p} e_i^p \cdot x_p \ge \gamma$.
    
    For each non-modulator part $i \in \{k+1, \dots, r\}$, all its edges come from the anonymous chunks assigned to it. 
    Since each chunk with edge-count $t$ contributes exactly $t$ edges and contributions between chunks can be only positive, then $\sum_{t} t \cdot y_{t,i} = |E(G[V_i])| \ge \gamma$ and so \Cref{thm:vi:enum:other_parts}.
    
    \Cref{thm:vi:enum:type_counts} holds as each component has a pattern. 
    Similarly \Cref{thm:vi:enum:chunks_count} holds as the $y$ variables were exactly set according the actual number of chunks which was determined from $x_p$.

\emph{Backward direction.} Assume one of the ILPs yields a feasible non-negative integer solution $(x_p, y_{t,i})$. 
    We reconstruct the partition of $V$ as follows.
    Place the vertices of $X$ into parts $V_1, \dots, V_k$ according to the current brute-forced partition $X = X_1 \cup \dots \cup X_k$.
    
    \Cref{thm:vi:enum:type_counts} ensures that every component with type $T \in \mathcal{T}$ was partitioned accordingly to patterns.
    We can arbitrarily match the components of type $T$ to the chosen patterns $p$.
    For each component, we distribute its vertices into $V_1, \dots, V_k$ according to its pattern $p$.
    The chunks can be distributed according to $y_{t,i}$ arbitrarily as no two chunks interfere negatively (no two chunks share edges, edges are only positive).
    Due to \Cref{thm:vi:enum:chunks_count} there are enough chunks to be assigned.
    Each vertex is assigned and the assignment is valid as \Cref{thm:vi:enum:type_counts} holds.
    As \Cref{thm:vi:enum:other_parts,thm:vi:enum:modulator_parts} holds each part has at least $\gamma$ induced edges.   
\end{proof}

The ILP can be modified to work for \textsc{BECGP} as well.

We need to track not only the edge-counts of chunks, but vertex-counts as well.
For that we define \emph{balanced pattern} $b \in \mathcal{B}$ such that $$b := (T, (e_1^b, e_2^b, \dots, e_k^b), (v_1^b, v_2^b, \dots, v_k^b), (m_{1,0}^b, m_{2,0}^b, m_{2,1}^b, \dots, m_{\vi, h}^b)),$$ where $v_i^b$ is the number of vertices assigned to the $i$-th part and $m_{c,t}^b$ is the number of chunks with exactly $c$ vertices and $t$ edges. 
The rest of the elements are defined the same as for pattern of a chunk. 

In a similar manner the variables are further distinguished with the vertex-counts.
\begin{enumerate}
    \item[2'] $y_{c, t, i}$: number of chunks with $c$ vertices, edge-count $t$ within the part $i$.
\end{enumerate}
We drop $y_{t,i}$ as a variable and each such occurrence is replaced as $y_{t, i} \to \sum_{c \in [\vi]} y_{c, t, i}$.
The $x_b$ variables are the same just are indexed with a balanced pattern instead.

We also have to further apply the division of chunks to \Cref{thm:vi:enum:chunks_count} to have the correct number of them.
And we add additional constraints.
\begin{enumerate}
    \item[4'] \label{thm:vi:enum:chunks_count_2} 
    The chunks are aggregated by the number of vertices $c$ and edges $t$ according to the number of patterns.
    $$\forall c \in [\vi], \forall t \in [h]_0: \sum_{i \in [k+1, r]} y_{c,t,i} = \sum_{b \in \mathcal{B}} m_{c,t}^b x_b$$

    \item[5] \label{thm:vi:enum:balanced_modulator}
    Each modulator part has $n/r$ vertices:
    $$ \forall i \in [k]: |X_i| + \sum_{b\in \mathcal{B}} v_i^b \cdot x_b = n/r$$

    \item[6] \label{thm:vi:enum:balanced_other}
    Similarly non-modulator parts:
    $$ \forall i \in [k+1, r]: \sum_{c \in [\vi], t} c \cdot y_{c, t, i} = n/r$$
\end{enumerate}

\begin{claim}\label{lem:running_time_vin_balanced}
  The updated $N$-fold ILP can be solved in time $2^{2^{\mathcal{O}(\vi^2)}} r \log(r)$.
\end{claim}   
\begin{claimproof}
    The constraints in \Cref{thm:vi:enum:modulator_parts,thm:vi:enum:type_counts} remain linking constraints.
    Their amount is still bounded with $2^{\mathcal{O}(\vi^2)}$ even after patterns were updated to balanced patterns as they are either created for types, or the first $k$ parts.
    The updated constraints in \Cref{thm:vi:enum:chunks_count_2} grow only to $\mathcal{O}(\vi^3)$.
    The constraints in \Cref{thm:vi:enum:balanced_modulator} are also linking constraints and their amount is bounded by $\vi$.

    Together constraints in \Cref{thm:vi:enum:other_parts,thm:vi:enum:balanced_other} are the local constraints,
    This time for each $i \in [k+1, r]$ all the local constraints with $y_{c,t,i}$ for any $c$ and $t$ form a brick. 
    There are exactly $2$ of such constraints in a single brick.
    Variables with different part index appear only in one brick. 

    Therefore the ILP program is still $N$-fold with the following parameters regarding \Cref{prop:n_fold_algo}.
    $\rho \in 2^{\mathcal{O}(\vi^2)}$,
    $\sigma \in \mathcal{O}(\vi^2)$,
    $a \in \rho\cdot \sigma \cdot \vi^2$
    and
    $d = |\mathcal{P}|+ r|\mathcal{T}| \in ({\mathcal{O}(\vi^2)}\cdot \vi^\mathcal{\vi^2})\cdot(r {\mathcal{O}(\vi^2)}) = r \vi^\mathcal{\vi^2}$.

    Therefore, the total running time of a single ILP is
    $2^{2^{\mathcal{O}(\vi^2)}} r \log(r)$. 
\end{claimproof}
Again the algorithm uses at most $\vi^\vi$ ILP instances and so the total running time is $2^{2^{\mathcal{O}(\vi^2)}} r \log(r) + n$.

\begin{claim}\label{clm:N-fold_ILP_vin_balanced}
       \textsc{BECGP} admits a solution if and only if at least one $N$-fold is feasible.
\end{claim}
\begin{claimproof}
     \emph{Forward direction.}
    There is a partition $V = V_1 \cup V_2 \dots \cup V_r$ such that $|E(G[V_i])| \ge \gamma$ and $|V_i| = \frac{n}{r}$ for all $i \in [r]$.
    We already know $X_i := X \cap V_i$.
    We set $x_b$ to be the exact number of components in $G \setminus X$ that exhibit balanced pattern $b$.
    For each anonymous chunk in $b$, if the number of vertices is exactly $c$ and edge-count is exactly $t$, it is counted in $y_{c,t,i}$.
    
    Now we can check that the constraints are satisfied. 
    The constraints \Cref{thm:vi:enum:modulator_parts,thm:vi:enum:other_parts,thm:vi:enum:chunks_count_2,thm:vi:enum:type_counts} are satisfied with the same argument from \Cref{clm:N-fold_ILP_vin}.

    The constraints \Cref{thm:vi:enum:balanced_modulator,thm:vi:enum:balanced_other} are satisfied as each part $V_i$ has size $\frac{n}{r}$ and the variables reflect the number of types and chunks with exactly that many non modulator vertices with together with the modulator gives us is exactly  $\frac{n}{r}$.

\emph{Backward direction.} 
    Assume one of the ILPs yields a feasible non-negative integer solution $(x_b, y_{c,t,i})$. 
    The same procedure as in \Cref{clm:N-fold_ILP_vin} gives us a \textsc{ECGP} solution $V_1, V_2, \dots V_r$ as \Cref{thm:vi:enum:chunks_count} is tightened to  \Cref{thm:vi:enum:chunks_count_2} and more constraints are added but not removed.
    We only have to argue that each $V_i$ has size of exactly $\frac{n}{r}$.
    Each part was reconstructed according to assigned patterns and brute-forced modulator.
    The number of vertices in each part is captured with $y_{c,t,i}$ and forced to be $\frac{n}{r}$ by \Cref{thm:vi:enum:balanced_modulator,thm:vi:enum:balanced_other}.
\end{claimproof}

This finishes the FPT algorithms when parameterized by vertex integrity.

\begin{theorem} 
\label{thm:uccp-cvdn_plus_u}
    Both \textsc{ECGP} and \textsc{BECGP} are FPT parameterized by $\cvd+\gamma$.
\end{theorem}
\begin{proof}
    We will use $N$-fold to design a FPT algorithm.
We denote the modulator as $X$ and the set of all clusters as $\mathcal{K}$.

First, we notice that if a part has more than $\gamma+1$ vertices in a single clique, those vertices can be moved to other parts, while satisfying the edge count constraint.
We define \emph{vertex type} of a vertex $v$ to be the $t_v = N(v) \cap X$.
Further we call a set of vertices in the same clique $A$ a \emph{chunk}.
The \emph{chunk type} $T$ of a chunk $A$ inside a clique is a multiset of vertex types of the vertices in the chunk.
The count of vertex types $t$ in $T$ is denoted as $T(t)$.
The set of all chunk types is $\mathcal{T}$.
We consider chunks of size at most $\gamma+1$, in that case the number of chunk types is bounded by $(2^\cvd)^{\gamma+1} \in 2^{\mathcal{O}(\gamma\cdot\cvd)}$. 

We use the $N$-fold ILP to both distribute the vertices into the chunk types and then form the parts from the chunk types. 
Again we guess the partition of the modulator and then run many $N$-fold ILPs and if any of those return YES then we have a YES-instance and NO-instance otherwise.
Assigning the modulator only parts $1$ to $\cvd$ ensures there are at most $\cvd^\cvd$ guesses.
The set of vertices from the modulator that belong to the $i$-th part is denoted as $X_i$.

Now follows the ILP formulation.

\noindent
\textbf{Variables}
\begin{itemize}
    \item $\forall T \in \mathcal{T}, K \in \mathcal{K}: x_{T, K} \in \mathbb{N}$ -- the number of chunks of chunk type $T$ in the cluster $K$.
    \item $\forall T \in \mathcal{T}, \forall i \in [r]: y_{T,i} \in \mathbb{N}$ -- the number of chunks of chunk type $T$ in the part $i$.
\end{itemize}

\noindent
\textbf{Constraints}
\begin{enumerate}
    \item \label{thm:cvd_u:enum:types_across_vertices} For each vertex type and cluster, there is at least the number of vertices of this type as there are demanded by the chunk types.
    $$\forall t \subseteq X, \forall K \in \mathcal{K}: \sum_{T \in \mathcal{T}} T(t)\cdot x_{T,K} \le |V_{t,K}|$$
    where $T(t)$ is the number occurrences of the vertex type $t$ inside the chunk type $T$ and $V_{t,K}$ is the set of vertices inside the cluster $K$ and vertex type $t$.

    \item \label{thm:cvd_u:enum:types_across_cliques_and_parts} For any chunk type, the number of chunks with that type has to be the same if we count them from the point of view of cliques and parts.
    $$\forall T \in \mathcal{T}: \sum_{K \in \mathcal{K}} x_{T,K} = \sum_{i \in [r]} y_{T,i}$$

    \item \label{thm:cvd_u:enum:edge-count} For each part $i \in [r]$, the edge-count has to be at least $\gamma$.
    $$\forall i \in [r]: \sum_{T \in \mathcal{T}} e(T, X_i) \cdot y_{T,i} + |E(G[X_i])| \ge \gamma$$
    where $e(T, X_i):= \binom{|T|}{2} + \sum_{t\subseteq X} T(t)\cdot |t\cap X_i|$ is the number of edges induced by a chunk of chunk type $T$ plus the number of edges between the chunk and $X_i$.

\end{enumerate}

\emph{The running time.} 
The constraints \Cref{thm:cvd_u:enum:types_across_cliques_and_parts} we classify as linking constraints. 
There are at most $2^{\mathcal{O}(\gamma\cdot\cvd)}$ many chunk types, so there are at most that many linking constraints.
The rest of the constraints are local constraints.
In the case of \Cref{thm:cvd_u:enum:types_across_vertices} constraints for the same cluster $K$ are grouped together to form a brick.
There are at most $2^\cvd$ of such constraints in each brick, as that is the number of different vertex types.
The constraints in \Cref{thm:cvd_u:enum:edge-count} each form a brick on its own.
No two bricks share variables as they are grouped by $K$ in $x_{T,K}$ and $i$ in $y_{T,i}$.
The largest coefficient that appears in a constraint is either $T(t)$ or $e(T, X_i)$, which are bounded by $\gamma$ and $\gamma^2 + \gamma \cdot \cvd$.

Substituting into \Cref{prop:n_fold_algo} we get the following runtime of a single ILP.
$\rho \in 2^{\mathcal{O}(\cvd \cdot \gamma)}$,
$\sigma= 2^\cvd$, 
$a \in \rho\cdot \sigma\cdot \gamma (\gamma + \cvd)$, and
$d = |\mathcal{T}| \cdot |\mathcal{K}| + |\mathcal{T}|\cdot r \in 2^{\mathcal{O}(\gamma\cdot\cvd)}n$.
Therefore, the total running time is ${(\gamma^2 + 2^\cvd)}^{2^{\mathcal{O}(\cvd \cdot \gamma)}} \cdot n \cdot \log(n)$.

\begin{claim}\label{clm:N-fold_ILP_cvd_plus_u}
       \textsc{ECGP} admits a solution if and only if at least one $N$-fold is feasible.
\end{claim}
\begin{claimproof}
     \emph{Forward direction.}
    Assume there exists a valid partition $V = V_1 \cup V_2 \dots \cup V_r$ such that $|E(G[V_i])| \ge \gamma$ for all $i \in [r]$.
    % The partition fixes $X_1, X_2, \dots X_r$.
    For each part $V_i$ and clique $K$ we let $A := V_i \cap K$ be a chunk. 
    If $A$ has more than $\gamma+1$ vertices, we remove any vertices until it has $\gamma+1$ vertices.
    The chunk $A$ has chunk type $T_A$.
    We set $x_{T_A,K}$ and $y_{T_A, i}$ as a count of all such $T_A$ we have created in this way.
    
    Now we can check that the constraints are satisfied. 
    
    Each chunk $A$ was counted once in $x_{T_A,K}$ and once in $y_{T_A,i}$, therefore \Cref{thm:cvd_u:enum:types_across_cliques_and_parts} is satisfied.
    The union of chunks associated with the $i$-th part is a subset of $V_i$ as each chunk contains at most $\gamma+1$ vertices from $V_i \cap K$.
    Every vertex is present only in one part and therefore in \Cref{thm:cvd_u:enum:types_across_vertices} each vertex was accounted at most once, so it holds.  
    The number of edges $|E(G[V_i])|$ is at least $\gamma$. 
    The chunks corresponding to $V_i$ may have fewer edges in total than $|E(G[V_i])|$, but only in case that there are more than $\gamma+1$ vertices in $V_i$ from the same clique.
    With $\gamma+1$ vertices in a clique, there are $\binom{\gamma+1}{2} \ge \gamma$ edges for $\gamma \ge 0$ and so \Cref{thm:cvd_u:enum:edge-count} is satisfied.
    
    \emph{Backward direction.}
    The algorithm brute-forces all partitions of $X$. 
    Thus, one iteration will exactly match the intersection of the true partition with $X$.
    
    We go through each type $T \in \mathcal{T}$.
    As \Cref{thm:cvd_u:enum:types_across_cliques_and_parts} holds we can pair each occurrence of type $T$ in $V_i$ with an occurrence in a clique $K$. 
    Then we assign a set of unassigned vertices with types $T$ in $K$ to $V_i$.
    There are enough vertices as \Cref{thm:cvd_u:enum:types_across_vertices} holds.
    Some of the vertices were unassigned as \Cref{thm:cvd_u:enum:types_across_vertices} only forces inequality. 
    We put all of them into $V_1$.
    This cannot decrease the edge-count in $G[V_i]$.
    
    Now $V_1, V_2, \dots, V_r$ is a partition. 
    Each part has at least $\gamma$ edges as the number of edges is determined just by the chunk type and their sum is forced in \Cref{thm:cvd_u:enum:edge-count}.
\end{claimproof}

To establish FPT for \textsc{BECGP} with respect to $\cvd + \gamma$ we just add the following constraints.

\begin{enumerate}
    \item[4] \label{thm:cvd_u:vertices_bound}
    Each part has at most $\frac{n}{r}$ many vertices:
    $$\forall i \in [r]: \sum_{T \in \mathcal{T}} |T| \cdot y_{T,i} + |X_i| \le \frac{n}{r}$$
\end{enumerate}

The constraints in \Cref{thm:cvd_u:vertices_bound} are all local constraints. 
Each of them is added to a different brick of the part $i$, where there is one constraint with exactly the same variables. 
The coefficients $|T|$ are at most $\gamma+1$ which is less than previously stated upper bound.
Therefore each updated $N$-fold has the same asymptotic runtime as before.

To argue the correctness the constraint in \Cref{thm:cvd_u:vertices_bound} limits that there are at most $\frac{n}{r}$ vertices from the chunks and the modulator.
This condition will be satisfied for any solution as vertices not in a chunk are not counted.
For the other direction it holds $\sum_{T, i} y_{T, i} = \sum_{T, K} x_{T, K}$.
Therefore, it is enough to assign each unassigned vertex such that there are exactly $\frac{n}{r}$ vertices in each part.

This concludes the FPT algorithms when parameterized by the cluster vertex deletion number and the edge count combined.

\end{proof}

\subsection{FPT via Vertex Integrity}

\begin{theorem}
\label{thm:paths-u}
\textsc{ECGP} is fixed-parameter tractable on graphs that are disjoint unions of paths when parameterized by $\gamma$.
\end{theorem}

\begin{proof}
Let $(G,r,\gamma)$ be an instance of \textsc{ECGP}, where connected components of $G$ are paths.

\begin{uccp}\label{paths+u RR}
   If a connected component $P=v_1,\dots,v_t$ has $t \ge \gamma+1$, then delete the vertices $v_1,\dots,v_{\gamma+1}$ and decrease $r$ by $1$. 
\end{uccp}

\begin{lemma}\label{lem:paths+u RR}
    Reduction Rule \ref{paths+u RR} is correct.
\end{lemma}
\begin{proof}
   Let $W:=\{v_1,\dots,v_{\gamma+1}\}$. Since $G[W]$ is a path on $\gamma+1$ vertices, it induces exactly $\gamma$ edges.

\smallskip
\noindent
\emph{Forward direction.}
Suppose $(G,r,\gamma)$ is a YES-instance. Let $P=v_1,\dots,v_t$ be a path component and let $W=\{v_1,\dots,v_{\gamma+1}\}$.
Among all feasible partitions of $V(G)$ into $r$ parts, fix one that minimizes the number of parts that contain vertices of $W$.
If all vertices of $W$ lie in a single part, we are done.
Otherwise, there exist indices $1 \le a \le b < c \le d \le \gamma+1$ such that:
\begin{itemize}
    \item $\{v_a,\dots,v_b\} \subseteq C_1$,
    \item $\{v_c,\dots,v_d\} \subseteq C_2$,
\end{itemize}
for two distinct parts $C_1 \neq C_2$, and $b+1=c$ (i.e., these two sets are consecutive along the path).

\medskip
\noindent
\emph{Step 1: Merge along the path.}
Move all vertices $v_c,\dots,v_d$ from $C_2$ to $C_1$.
Since $b+1=c$, the edge $v_bv_c$ becomes internal to $C_1$. Hence the number of edges induced by $C_1$ increases by $1$.

\medskip
\noindent
\emph{Step 2: Reconstruct parts.}
Let $S$ be the set of vertices originally in $C_1 \cup C_2$.
We now repartition $S$ into two parts $C_1'$ and $C_2'$ as follows.
Initialize $C_1' := \{v_a,\dots,v_d\}$. Let $R := S \setminus C_1'$.

We add vertices from $R$ to $C_1'$ in the following manner. Since $G$ is a disjoint union of paths, $R$ induces a collection of vertex-disjoint paths. Consider any ordering of these path components, say $Q_1,Q_2,\dots,Q_{l}$.
Iteratively add the vertices of $Q_1,Q_2,\dots,Q_{l}$ to $C_1'$ until the induced subgraph $G[C_1']$ has at least $\gamma$ edges. Let $Q_i$ be the first path whose addition makes the number of edges at least $\gamma$.

Since each $Q_i$ is a path, there exists a prefix $Q_i'$ of $Q_i$ such that adding $Q_i'$ to $C_1'$ results in exactly $u$ edges.
Set
\[
C_1' := C_1' \cup Q_i', \qquad
C_2' := (R \setminus Q_i') \cup (Q_i \setminus Q_i').
\]
\noindent Since $C_1$ and $C_2$ are feasible parts, we have
$|E(G[C_1])| \ge \gamma$ and $|E(G[C_2])| \ge \gamma$,
and hence
\[
|E(G[C_1])| + |E(G[C_2])| \ge 2\gamma.
\]

When we move the vertices $v_c,\dots,v_d$ from $C_2$ to $C_1$, the edge $v_bv_c$ (with $b+1=c$) becomes internal to $C_1'$. Thus the total number of internal edges increases by $1$.

In the reconstruction step, we split at most one path $Q_i$ into two consecutive parts $Q_i'$ and $Q_i \setminus Q_i'$. Since $Q_i$ is a path, there is exactly one edge between these two parts, and this edge is not counted in either part. Hence at most one edge is lost.
Therefore, the total number of internal edges in $C_1'$ and $C_2'$ satisfies
\begin{align*}
|E(G[C_1'])| + |E(G[C_2'])|
&\ge |E(G[C_1])| + |E(G[C_2])| + 1 - 1 \\
&\ge 2\gamma.
\end{align*}
By construction, $C_1'$ induces exactly $\gamma$ edges. Consequently,
\[
|E(G[C_2'])| \;\ge\; 2\gamma - \gamma \;=\; \gamma.
\]
Thus both $C_1'$ and $C_2'$ induce at least $\gamma$ edges, and the modified partition is feasible.
Thus we obtain a valid partition into $r$ parts in which fewer parts intersect $W$, contradicting minimality.
Therefore, in some feasible solution, all vertices of $W$ lie in a single part.

\smallskip
\noindent
\emph{Reverse direction.}
If the reduced instance $(G',r-1,\gamma)$ is a YES-instance, then we can add $W$ as one part. This part induces exactly $\gamma$ edges, and all other parts remain valid. Hence $(G,r,\gamma)$ is a YES-instance.
\end{proof}

Apply this rule exhaustively. Let $(G',r',\gamma)$ be the resulting instance.
In $G'$, every connected component has at most $\gamma$ vertices. Otherwise, the reduction rule would still apply. Since $G'$ is a disjoint union of paths, this implies
$\mathsf{vi}(G') \le \gamma$.
Due to Theorem~\ref{thm:uccp-vin}, we know that \textsc{ECGP} is FPT when parameterized by $\mathsf{vi}$. Applying this algorithm to $(G',r',\gamma)$ yields an FPT algorithm.
\end{proof}

\begin{theorem} \label{thm:vdp-plus-u} 
\textsc{ECGP} is fixed-parameter tractable when parameterized by $\vdp+ \gamma$.  
\end{theorem}
\begin{proof}
Let $X$ be the given vertex deletion set of size $k$ such that $G-X$ is a disjoint union of paths, and let $\mathcal{P}$ denote the set of connected components of $G-X$. 

\medskip \noindent \textbf{Step 1: Guessing and normalizing the parts intersecting $X$.}
A part of a partition of $V(G)$ is called \emph{special} if it intersects $X$. Since the parts are pairwise disjoint and every special part contains at least one vertex of $X$, any partition contains at most $k$ special parts.

Consider a feasible partition $V(G)=V_1\dot\cup\cdots\dot\cup V_r$, and let $p$ be the number of its special parts. After relabeling the parts, we may assume that $V_1,\ldots,V_p$ are precisely the special parts. For every $i\in[p]$, let $X_i:=V_i\cap X$. The sets $X_1,\ldots,X_p$ form a partition of $X$ into nonempty sets. Thus $0\le p\le\min\{k,r\}$, where $p=0$ is possible only when $X=\emptyset$.

The algorithm branches over all partitions $X=X_1\dot\cup\cdots\dot\cup X_p$  of $X$ into $p$ nonempty sets, including the unique empty partition when $X=\emptyset$. In the branch corresponding to $(X_1,\ldots,X_p)$, we seek a solution in which $X_i$ is the intersection of the $i$-th special part with $X$. The number of branches is bounded by a function of $k$; for instance, it is at most $k^k$ when $k\ge1$.

The following lemma shows that each special part can be certified using at most $2\gamma$ vertices outside $X$.

\begin{lemma} \label{lem:vdp-special-part-normalization} 
Let $V_1,\ldots,V_r$ be a feasible partition of $V(G)$, and suppose that $V_1,\ldots,V_p$ are precisely the parts intersecting $X$.
For every $i\in[p]$, let $X_i:=V_i\cap X$.
Then there exist pairwise disjoint sets $R_1,\ldots,R_p\subseteq V(G)\setminus X$ such that, for every $i\in[p]$, 
\[ R_i\subseteq V_i,\qquad |R_i|\le 2\gamma,\qquad\text{and}\qquad |E(G[X_i\cup R_i])|\ge \gamma. \] 
Consequently, $X_i\cup R_i$ is a certificate of feasibility for the $i$-th special part, and every vertex of $V_i\setminus(X_i\cup R_i)$ is irrelevant to the feasibility of that part.
\end{lemma}
\begin{proof} Fix $i\in[p]$. Since $V_i$ is feasible, $G[V_i]$ contains at least $\gamma$ edges. 
Choose a set $F_i\subseteq E(G[V_i])$ of exactly $u$ edges, and let $Q_i$ be the set of their endpoints. 
Then $|Q_i|\le 2\gamma$. 
Define $R_i:=Q_i\setminus X_i$. 
Since $Q_i\subseteq V_i$ and $V_i\cap X=X_i$, we have $R_i\subseteq V_i\setminus X$ and $|R_i| \le 2\gamma$.

Moreover, $Q_i\subseteq X_i\cup R_i$. Hence every edge of $F_i$ is induced by $X_i\cup R_i$, and therefore $|E(G[X_i\cup R_i])|\ge |F_i|=\gamma$. 
Since the parts $V_1,\ldots,V_p$ are pairwise disjoint and $R_i\subseteq V_i$, the sets $R_1,\ldots,R_p$ are pairwise disjoint. Finally, removing vertices of $V_i\setminus(X_i\cup R_i)$ from the $i$-th special part does not remove any edge of $F_i$. Thus $X_i\cup R_i$ remains feasible. Such removed vertices may subsequently be assigned to any already feasible part, since adding vertices cannot decrease the number of induced edges. Therefore, at most $2\gamma$ vertices of $G-X$ are needed to certify the feasibility of each special part. \end{proof}

For a fixed branch $X=X_1\dot\cup\cdots\dot\cup X_p$, we call the vertices of $R_i$ the \emph{certificate vertices} of the $i$-th special part. 
Let $R:=\bigcup_{i=1}^{p}R_i$  be the set of all certificate vertices. By Lemma~\ref{lem:vdp-special-part-normalization}, we have $|R|=\sum_{i=1}^{p}|R_i|\le 2p\gamma\le 2k\gamma$. 

A path $P\in\mathcal{P}$ is called \emph{affected} if $V(P)\cap R\neq\emptyset$, and \emph{unaffected} otherwise. Since every affected path contains at least one certificate vertex and the paths in $\mathcal{P}$ are pairwise vertex-disjoint, the number of affected paths is at most $|R|\le 2k\gamma$.

If an affected path contains $s$ certificate vertices, then deleting these vertices produces at most $s+1$ residual path segments.
Let  $\mathcal{P}_R := \{P\in\mathcal{P}:V(P)\cap R\neq\emptyset\}$  denote the set of affected paths. 
The total number of residual path segments produced by the affected paths is at most  $\sum_{P\in\mathcal{P}_R} \bigl(|V(P)\cap R|+1\bigr) = |R|+|\mathcal{P}_R| \le 2|R| \le 4k\gamma$.  
Thus, for the fixed partition of $X$, every feasible solution admits certificates for its special parts using at most $2ku$ vertices of $G-X$.
These certificate vertices belong to at most $2k\gamma$ path components, and their deletion produces at most $4k\gamma$ residual path segments.
In the next step, we represent the possible assignments of certificate vertices to the special parts by bounded signatures, without enumerating their positions explicitly.

\medskip \noindent \textbf{Step 2: Baseline decompositions and signatures of affected paths.}
Fix a branch corresponding to a partition $X=X_1\dot\cup\cdots\dot\cup X_p$, where $p\le k$. By Lemma~\ref{lem:vdp-special-part-normalization}, it suffices to select at most $2u$ certificate vertices for each special part. 
Hence at most $2p\gamma \le 2k\gamma$ vertices of $G-X$ are selected as certificate vertices in total.

We first associate with every path $P\in\mathcal{P}$ a fixed \emph{baseline decomposition}.
This is the decomposition that $P$ would contribute if none of its vertices were selected as certificate vertices.
At this point, we do not determine whether $P$ is affected. The baseline is only a fixed reference contribution. 
Later, the dynamic program either retains this baseline or replaces it with the contribution of an affected configuration.

Let $m:=\gamma+1$. 
For every path $P\in\mathcal{P}$, let $n_P:=|V(P)|$ and write  $n_P=a_Pm+b_P$, and $0\le b_P<m$. 
Equivalently, $a_P=\lfloor n_P/(\gamma+1)\rfloor$ and $b_P=n_P\bmod (\gamma+1)$.

By exhaustive application of Reduction Rule~\ref{paths+u RR}, the path $P$ produces $a_P$ completed feasible parts and leaves, if $b_P>0$, one residual path on $b_P$ vertices.
Indeed, each application removes $\gamma+1$ consecutive vertices, which induce exactly $u$ path edges, and decreases the number of required parts by one. 
We call $a_P$ the \emph{baseline block contribution} of $P$ and $b_P$ its \emph{baseline remainder}.

Define $A_0:=\sum_{P\in\mathcal{P}}a_P$  and, for every $\ell\in[\gamma]$, define  $c^0_\ell := \bigl|\{P\in\mathcal{P}:b_P=\ell\}\bigr|$.
Thus, under the baseline decomposition, the paths of $G-X$ produce $A_0$ completed feasible parts, while the residual graph is the disjoint union of $c^0_\ell$ copies of the path $P_\ell$ for every $\ell\in[\gamma]$.

The values $A_0,c^0_1,\ldots,c^0_u$ are determined entirely by the orders of the path components of $G-X$ and can be computed in polynomial time. They are not guessed or stored in the dynamic-programming states. Instead, the dynamic program stores only the bounded corrections caused by paths that are selected as affected.

\medskip 
\noindent \emph{Configurations of an affected path.} Let $P=v_1,\ldots,v_t$ be a path in $\mathcal{P}$. 
A \emph{configuration} of $P$ is a mapping $\varphi:V(P)\rightarrow\{0,1,\ldots,p\}$. 
The value $\varphi(v)=0$ means that $v$ is ordinary, whereas $\varphi(v)=i\in[p]$ means that $v$ is selected as a certificate vertex for the $i$-th special part.
A configuration is called \emph{affected} if at least one vertex receives a nonzero value.

For every $i\in[p]$, define  $s_i(P,\varphi) := \bigl|\{v\in V(P):\varphi(v)=i\}\bigr|$. 
We consider only configurations satisfying $s_i(P,\varphi)\le2\gamma$ for every $i\in[p]$.
Let  $s(P,\varphi):=\sum_{i=1}^{p}s_i(P,\varphi)$  denote the total number of certificate vertices selected from $P$.

For every $i\in[p]$, define the edge count contribution of $P$ to the $i$-th special part as \[ \begin{aligned} \mu_i(P,\varphi) &:= \sum_{\substack{v\in V(P)\\ \varphi(v)=i}} |N_G(v)\cap X_i| + \bigl| \{v_jv_{j+1}\in E(P): \varphi(v_j)=\varphi(v_{j+1})=i\} \bigr|. \end{aligned} \] 
The first term counts edges between certificate vertices assigned to the $i$-th special part and vertices of $X_i$. 
The second term counts path edges whose endpoints are both assigned to the $i$-th special part. 
We truncate this contribution at $\gamma$ and define $\alpha_i(P,\varphi):=\min\{\gamma,\mu_i(P,\varphi)\}$. 
Edges induced entirely by $X_i$ are not included in $\alpha_i(P,\varphi)$ and will be counted separately.

\medskip 
\noindent \emph{Residual segments of an affected path.}
Let $\varphi$ be an affected configuration of $P$, and let $s:=s(P,\varphi)\ge1$. Deleting the $s$ certificate vertices selected by $\varphi$ produces at most $s+1$ ordinary path segments.
Let $g_0,\ldots,g_s$ denote their numbers of vertices, allowing zero-length end gaps. Since $P$ has $t$ vertices, we have $\sum_{j=0}^{s}g_j=t-s$. 
For every $j\in\{0,\ldots,s\}$, write  $g_j=h_jm+\rho_j$, and $0\le\rho_j<m$.  Once the certificate vertices have been assigned to the special parts, the ordinary segments are used exclusively for non-special parts. Hence Reduction Rule~\ref{paths+u RR} may be applied exhaustively to each such segment.
The segment on $g_j$ vertices produces $h_j$ completed feasible parts and leaves, if $\rho_j>0$, one residual path on $\rho_j\le \gamma$ vertices.

Define the block contribution of $(P,\varphi)$ as  $a(P,\varphi) := \sum_{j=0}^{s}h_j = \sum_{j=0}^{s} \left\lfloor\frac{g_j}{\gamma+1}\right\rfloor$.
For every $\ell\in[\gamma]$, define $c_\ell(P,\varphi) := \bigl|\{j\in\{0,\ldots,s\}:\rho_j=\ell\}\bigr|$. 
Thus, $c_\ell(P,\varphi)$ is the number of residual paths on exactly $\ell$ vertices produced by the configuration.

We compare these values with the baseline contribution of $P$. 
Define the \emph{block correction} $\delta(P,\varphi):=a(P,\varphi)-a_P$  and, for every $\ell\in[\gamma]$, the \emph{remainder correction} \[ d_\ell(P,\varphi) := c_\ell(P,\varphi)-\mathbf{1}[b_P=\ell], \] where $\mathbf{1}[\mathcal{Q}]$ is $1$ if $\mathcal{Q}$ holds and is $0$ otherwise.

The following lemma bounds these corrections in terms of the number of certificate vertices selected from $P$, independently of the order of $P$. 
\begin{lemma} \label{lem:vdp-bounded-path-correction} 
Let $P$ be a path, and let $\varphi$ be an affected configuration of $P$ selecting exactly $s$ certificate vertices.
Then $-s\le\delta(P,\varphi)\le0$. Moreover, for every $\ell\in[\gamma]$,  $-1\le d_\ell(P,\varphi)\le s+1$,  and $\sum_{\ell=1}^{\gamma}\max\{0,d_\ell(P,\varphi)\}\le s+1$. \end{lemma} 
\begin{proof} 
Let $t:=|V(P)|$ and $m:=\gamma+1$. By definition, 
\[ a_P=\left\lfloor\frac{t}{m}\right\rfloor \qquad\text{and}\qquad a(P,\varphi) = \sum_{j=0}^{s} \left\lfloor\frac{g_j}{m}\right\rfloor, \]
where $\sum_{j=0}^{s}g_j=t-s$. 
Since 
\[ \sum_{j=0}^{s} \left\lfloor\frac{g_j}{m}\right\rfloor \le \left\lfloor \frac{\sum_{j=0}^{s}g_j}{m} \right\rfloor = \left\lfloor\frac{t-s}{m}\right\rfloor \le \left\lfloor\frac{t}{m}\right\rfloor, \] 
we obtain $a(P,\varphi)\le a_P$, and hence $\delta(P,\varphi)\le0$.

For the lower bound, write $t=a_Pm+b_P$, where $0\le b_P<m$, and $g_j=h_jm+\rho_j$, where $0\le\rho_j<m$. 
Since $\sum_{j=0}^{s}g_j=t-s$, we have  $a(P,\varphi)m+\sum_{j=0}^{s}\rho_j = a_Pm+b_P-s$.  
Therefore, \[ (a_P-a(P,\varphi))m = s+\sum_{j=0}^{s}\rho_j-b_P. \] Using $\rho_j\le m-1$ and $b_P\ge0$, we obtain \[ \begin{aligned} (a_P-a(P,\varphi))m &\le s+(s+1)(m-1) \\ &=sm+(m-1). \end{aligned} \] Consequently, \[ a_P-a(P,\varphi) \le s+\frac{m-1}{m} <s+1. \] Since the left-hand side is an integer, $a_P-a(P,\varphi)\le s$, and hence $\delta(P,\varphi)\ge-s$.

Deleting $s$ vertices from a path produces at most $s+1$ nonempty residual segments. Therefore, $c_\ell(P,\varphi)\le s+1$ for every $\ell\in[\gamma]$, while the baseline contains at most one residual path. It follows that \[ -1 \le c_\ell(P,\varphi)-\mathbf{1}[b_P=\ell] \le s+1. \] 
Finally, $\sum_{\ell=1}^{\gamma}c_\ell(P,\varphi)\le s+1$.  Subtracting the single baseline residual path, if one exists, cannot increase the sum of the positive coordinates. Hence \[ \sum_{\ell=1}^{\gamma}\max\{0,d_\ell(P,\varphi)\}\le s+1. \] \end{proof}
We associate with an affected configuration $(P,\varphi)$ the \emph{signature} 
\[ \sigma(P,\varphi) := \left( (s_i(P,\varphi))_{i\in[p]}, (\alpha_i(P,\varphi))_{i\in[p]}, \delta(P,\varphi), (d_\ell(P,\varphi))_{\ell\in[\gamma]} \right). \] 
The signature records precisely the information relevant outside $P$: the numbers of certificate vertices assigned to the special parts, their edge count contributions, and the corrections to the baseline numbers of completed parts and residual short paths.
Let \[ \Sigma(P) := \{\sigma(P,\varphi): \varphi\text{ is a valid affected configuration of }P\} \] denote the set of realizable signatures of $P$. 
Thus, $\sigma\in\Sigma(P)$ means that at least one affected configuration of $P$ has the effect recorded by $\sigma$. 
Distinct configurations may yield the same signature and need not be distinguished by the global algorithm.

Consider a collection $\mathcal{A}\subseteq\mathcal{P}$ of affected paths, and suppose that, for every $P\in\mathcal{A}$, an affected configuration $\varphi_P$ has been chosen. 
Assume that these configurations select at most $2k\gamma$ certificate vertices in total, that is, $\sum_{P\in\mathcal{A}} s(P,\varphi_P)\le 2k\gamma$.  
Define the \emph{accumulated block correction} by
$\Delta A := \sum_{P\in\mathcal{A}}\delta(P,\varphi_P)$,  
and, for every $\ell\in[\gamma]$, define the \emph{accumulated remainder correction} by 
$\Delta c_\ell := \sum_{P\in\mathcal{A}}d_\ell(P,\varphi_P)$.  
Since every affected path contains at least one certificate vertex, we have $|\mathcal{A}|\le 2k\gamma$. 
Moreover, by Lemma~\ref{lem:vdp-bounded-path-correction}, \[ \begin{aligned} \Delta A &\ge -\sum_{P\in\mathcal{A}}s(P,\varphi_P) \ge -2k\gamma, \end{aligned} \] while $\delta(P,\varphi_P)\le0$ for every $P\in\mathcal{A}$. 
Hence $-2k\gamma\le\Delta A\le0$.

For every $\ell\in[\gamma]$, each affected path can remove at most one baseline residual path of order $\ell$. 
Therefore, \[ \Delta c_\ell\ge-|\mathcal{A}|\ge-2k\gamma. \] 
On the other hand, the configurations represented by $\mathcal{A}$ produce at most \[ \sum_{P\in\mathcal{A}} \bigl(s(P,\varphi_P)+1\bigr) \le 2k\gamma+|\mathcal{A}| \le 4k\gamma \] residual path segments in total. 
Consequently,  $\Delta c_\ell\le4ku$ for every $\ell\in[\gamma]$. Thus, we get $-2k\gamma\le\Delta c_\ell\le4k\gamma$  for every $\ell\in[\gamma]$.

The values $\Delta A$ and $\Delta c_\ell$ describe the changes relative to the fixed baseline caused by the selected affected configurations. 
In particular, the resulting total number of completed feasible parts is  $A=A_0+\Delta A$, and the resulting number of residual paths on $\ell$ vertices is $c_\ell=c^0_\ell+\Delta c_\ell$ for every $\ell\in[\gamma]$.

Therefore, although the baseline values $A_0,c^0_1,\ldots,c^0_u$ may depend on the input size, all corrections stored by the global dynamic program are bounded by a function of $k+\gamma$. In the next step, we compute the sets $\Sigma(P)$ using a finite-state dynamic program and combine the resulting signatures over all path components.

\medskip
\noindent \textbf{Step 3: Computing the signatures of a path.}

Fix a branch $X=X_1\dot\cup\cdots\dot\cup X_p$, where $p\le k$, and let $P=v_1,\ldots,v_t$ be a path component of $G-X$. We describe a left-to-right dynamic program that computes the set $\Sigma(P)$ of all realizable signatures of affected configurations of $P$.

Recall that a configuration of $P$ is a mapping $\varphi:V(P)\rightarrow\{0,1,\ldots,p\}$.
The value $\varphi(v)=0$ means that $v$ is ordinary, whereas $\varphi(v)=i\in[p]$ means that $v$ is selected as a certificate vertex for the $i$-th special part.
We consider only configurations satisfying $s_i(P,\varphi)\le 2u$ for every $i\in[p]$ and $s(P,\varphi)\le 2k\gamma$. 
Let $m:=\gamma+1$. For every $j\in\{0,\ldots,t\}$, let $P_j:=P[\{v_1,\ldots,v_j\}]$ denote the prefix of $P$ on its first $j$ vertices, with $P_0$ being the empty graph. 
A state after processing $P_j$ is a tuple \[ \eta= \left( (\alpha_i)_{i\in[p]}, (s_i)_{i\in[p]}, c, \lambda, z, (q_\ell)_{\ell\in[\gamma]} \right). \] 
The entries have the following meanings.

\begin{enumerate}
\item For every $i\in[p]$, the value $\alpha_i\in\{0,\ldots,\gamma\}$ is the edge count contributed by the processed prefix to the $i$-th special part, truncated at $u$. It counts edges between certificate vertices assigned to the $i$-th special part and vertices of $X_i$, together with path edges in $P_j$ whose endpoints are both assigned to the $i$-th special part. 
\item For every $i\in[p]$, the value $s_i\in\{0,\ldots,2\gamma\}$ is the number of vertices of $P_j$ selected as certificate vertices for the $i$-th special part. 
\item The value $c\in\{0,1,\ldots,p\}$ records the assignment of the last processed vertex $v_j$. The value $c=0$ means that $v_j$ is ordinary, while $c=i\in[p]$ means that $v_j$ is assigned to the $i$-th special part. For $j=0$, we set $c=0$. 
\item The value $\lambda\in\{0,\ldots,\gamma\}$ records the residual order of the currently open ordinary segment after extracting every possible complete block of $\gamma+1$ consecutive ordinary vertices. More precisely, if the currently open ordinary segment has $L$ vertices, then $L=h(\gamma+1)+\lambda$ for some number $h$ of already extracted complete blocks. Thus $\lambda=L\bmod (\gamma+1)$. If $v_j$ is a certificate vertex, then $\lambda=0$. 
\item Let $a_j$ be the number of complete blocks of $\gamma+1$ ordinary vertices extracted from the ordinary segments of $P_j$. We store the correction \[ z:=a_j-\left\lfloor\frac{j}{\gamma+1}\right\rfloor \] relative to the baseline block contribution of the prefix $P_j$. 
\item For every $\ell\in[\gamma]$, the value $q_\ell$ is the number of already closed ordinary segments whose residual order, after extracting all complete blocks, is exactly $\ell$. The currently open ordinary segment, represented by $\lambda$, is not included in these counters. 
\end{enumerate}

A state is retained only if $\sum_{i=1}^{p}s_i\le 2k\gamma$. 

\medskip 
\noindent \emph{Initial state.} For $j=0$, the table contains only the state  $\left( (0)_{i\in[p]}, (0)_{i\in[p]}, 0, 0, 0, (0)_{\ell\in[\gamma]} \right)$.

\medskip 
\noindent \emph{Transitions.} 
Suppose that \[ \eta= \left( (\alpha_i)_{i\in[p]}, (s_i)_{i\in[p]}, c, \lambda, z, (q_\ell)_{\ell\in[\gamma]} \right) \] is reachable after processing $P_j$, where $0\le j<t$. Define \[ \Delta_j := \left\lfloor\frac{j+1}{\gamma+1}\right\rfloor - \left\lfloor\frac{j}{\gamma+1}\right\rfloor. \] Thus $\Delta_j\in\{0,1\}$ indicates whether the baseline decomposition completes a new block when the prefix is extended from $P_j$ to $P_{j+1}$.

\smallskip 
\noindent \emph{Ordinary transition.} 
Suppose that $v_{j+1}$ is declared ordinary. If $\lambda<\gamma$, no new complete ordinary block is formed, and we set $\lambda':=\lambda+1$ and $z':=z-\Delta_j$. If $\lambda=\gamma$, then $v_{j+1}$ completes a block of $\gamma+1$ consecutive ordinary vertices. We extract this block and set $\lambda':=0$ and $z':=z+1-\Delta_j$. In both cases, the vectors $(\alpha_i)_{i\in[p]}$, $(s_i)_{i\in[p]}$, and $(q_\ell)_{\ell\in[\gamma]}$ remain unchanged, and we set $c':=0$.

\smallskip 
\noindent \emph{Certificate transition.} Suppose that $v_{j+1}$ is assigned to the $i$-th special part, where $i\in[p]$. This transition is allowed only if $s_i<2\gamma$ and $\sum_{h=1}^{p}s_h<2k\gamma$.

We set $s'_i:=s_i+1$ and $s'_h:=s_h$ for every $h\neq i$. The vertex $v_{j+1}$ contributes exactly  $|N_G(v_{j+1})\cap X_i|$ edges between the path and $X_i$. Moreover, if $j\ge1$ and $c=i$, then the edge $v_jv_{j+1}$ also belongs to the $i$-th special part. Hence we set \[ \alpha'_i := \min\left\{ \gamma,\, \alpha_i+|N_G(v_{j+1})\cap X_i|+\mathbf{1}[c=i] \right\}, \] and $\alpha'_h:=\alpha_h$ for every $h\neq i$.

If $\lambda>0$, then assigning $v_{j+1}$ to a special part closes the currently open ordinary segment. Its residual order is $\lambda$, since all complete blocks have already been extracted. Therefore, for every $\ell\in[\gamma]$, we set \[ q'_\ell:=q_\ell+\mathbf{1}[\ell=\lambda]. \] If $\lambda=0$, we set $q'_\ell:=q_\ell$ for every $\ell\in[\gamma]$. Finally, we set $c':=i$, $\lambda':=0$, and $z':=z-\Delta_j$. No complete ordinary block is created by this transition, while the baseline may complete one.

\medskip \noindent \emph{Bounds on the local states.}
We retain only states satisfying \[ -2k\gamma\le z\le0 \qquad\text{and}\qquad \sum_{\ell=1}^{\gamma}q_\ell\le2k\gamma. \] These restrictions do not discard any state that can lead to a relevant configuration. Indeed, consider a configuration of a prefix $P_j$ selecting $s\le2k\gamma$ certificate vertices. Applying Lemma~\ref{lem:vdp-bounded-path-correction} to $P_j$ shows that the difference between the number of complete ordinary blocks and the baseline block contribution lies in $[-s,0]$. Hence $-2k\gamma\le z\le0$.

Furthermore, every nonempty closed ordinary segment ends immediately before a certificate vertex. Thus a prefix containing $s$ certificate vertices has at most $s$ closed ordinary segments. The possible ordinary segment following the last certificate vertex is represented separately by $\lambda$. Therefore, \[ \sum_{\ell=1}^{\gamma}q_\ell\le s\le2k\gamma. \]

\medskip 
\noindent \emph{Extracting signatures.} After processing the entire path $P$, let \[ \eta= \left( (\alpha_i)_{i\in[p]}, (s_i)_{i\in[p]}, c, \lambda, z, (q_\ell)_{\ell\in[\gamma]} \right) \] be a reachable final state. If $\sum_{i=1}^{p}s_i=0$, then the corresponding configuration is unaffected and its contribution is already represented by the baseline. We therefore do not extract an affected signature from this state.

Suppose that $\sum_{i=1}^{p}s_i\ge1$. The possible final open ordinary segment must now be closed. For every $\ell\in[\gamma]$, define  $\widehat q_\ell := q_\ell+\mathbf{1}[\lambda=\ell]$.
Let $b_P:=|V(P)|\bmod (\gamma+1)$ be the baseline remainder of $P$, and define  $d_\ell := \widehat q_\ell-\mathbf{1}[b_P=\ell]$  for every $\ell\in[\gamma]$.
The final value of $z$ is exactly the block correction $\delta(P,\varphi)$, since \[ z = a(P,\varphi) - \left\lfloor\frac{|V(P)|}{\gamma+1}\right\rfloor. \] Thus the final state realizes the signature \[ \left( (s_i)_{i\in[p]}, (\alpha_i)_{i\in[p]}, z, (d_\ell)_{\ell\in[\gamma]} \right). \] Let $\Sigma(P)$ be the set of all signatures extracted from reachable final states in this manner.

\begin{lemma} \label{lem:vdp-local-signatures} 
For every path component $P$ of $G-X$, the dynamic program computes exactly the set $\Sigma(P)$ of signatures realizable by affected configurations of $P$ satisfying $s_i(P,\varphi)\le2\gamma$ for every $i\in[p]$ and $s(P,\varphi)\le2k\gamma$. 
\end{lemma} 
\begin{proof} We prove the claim by induction on the number of processed vertices. The initial state represents the unique configuration of the empty prefix. Suppose that the claim holds after processing $P_j$. Every configuration of $P_{j+1}$ restricts to a configuration of $P_j$ represented by a reachable state. 
If $v_{j+1}$ is ordinary, the ordinary transition extends the current ordinary segment and extracts a complete block exactly when its residual order reaches $\gamma+1$. 
The update of $z$ records the difference between this actual block contribution and the corresponding change in the baseline. 
If $v_{j+1}$ is assigned to the $i$-th special part, the certificate transition increments $s_i$, adds the edges from $v_{j+1}$ to $X_i$, and counts the path edge $v_jv_{j+1}$ exactly when $v_j$ is also assigned to the $i$-th special part. 
It also closes and records the current ordinary segment, if one exists. Thus every valid configuration of $P_{j+1}$ gives rise to a reachable state with the stated interpretation. 
Conversely, every transition assigns $v_{j+1}$ either to no special part or to exactly one special part and updates all entries according to that assignment. Hence every reachable state corresponds to a valid configuration of the processed prefix.
After the entire path has been processed, closing the final ordinary segment and subtracting the fixed baseline remainder yields precisely the signature defined in Step~2. Therefore, the extracted signatures are exactly the realizable affected signatures of $P$.
\end{proof}

\medskip \noindent \emph{Running time.} At each path position, the number of states is at most \[ (\gamma+1)^p(2\gamma+1)^p(p+1)(\gamma+1)(2k\gamma+1)(2k\gamma+1)^\gamma. \] Indeed, the edge count vector has at most $(\gamma+1)^p$ possibilities, the certificate-count vector has at most $(2\gamma+1)^p$ possibilities, the last assignment has $p+1$ possibilities, the open residual order has $\gamma+1$ possibilities, the correction variable has at most $2k\gamma+1$ possibilities, and each residual-segment counter has at most $2k\gamma+1$ possibilities. Each state has at most $p+1\le k+1$ outgoing transitions. Hence $\Sigma(P)$ can be computed in time $f_1(k,\gamma)\cdot |V(P)|$ for some computable function $f_1$. In the next step, we combine the signature sets over all path components, verify the feasibility of the special parts, and construct the residual instance of bounded vertex integrity.

\medskip 
\noindent \textbf{Step 4: Combining the path signatures.}

We now combine the signature sets $\Sigma(P)$ over all paths $P\in\mathcal{P}$. Recall that the fixed baseline decomposition yields \[ A_0 := \sum_{P\in\mathcal{P}} \left\lfloor\frac{|V(P)|}{\gamma+1}\right\rfloor \] completed feasible parts. Moreover, for every $\ell\in[\gamma]$, the baseline contains \[ c^0_\ell := \bigl| \{P\in\mathcal{P}: |V(P)|\bmod(\gamma+1)=\ell\} \bigr| \] residual paths on exactly $\ell$ vertices. Let $\mathcal{P}=\{P_1,\ldots,P_m\}$. We process these paths one at a time. For every $j\in\{0,\ldots,m\}$, a state after processing $P_1,\ldots,P_j$ is a tuple \[ \zeta = \left( (\alpha_i)_{i\in[p]}, (s_i)_{i\in[p]}, \Delta A, (\Delta c_\ell)_{\ell\in[\gamma]} \right). \] The entries have the following meanings.

\begin{enumerate} \item For every $i\in[p]$, the value $\alpha_i\in\{0,\ldots,\gamma\}$ is the total edge count contributed by the certificate vertices selected from the processed paths to the $i$-th special part, truncated at $\gamma$. \item For every $i\in[p]$, the value $s_i\in\{0,\ldots,2\gamma\}$ is the total number of certificate vertices assigned to the $i$-th special part from the processed paths. \item The value $\Delta A$ is the accumulated correction to the number of completed $(\gamma+1)$-vertex blocks relative to the baseline contributions of the processed paths. \item For every $\ell\in[\gamma]$, the value $\Delta c_\ell$ is the accumulated correction to the number of residual paths on $\ell$ vertices relative to the baseline contributions of the processed paths. 
\end{enumerate}

We retain only states satisfying $s_i\le2\gamma$ for every $i\in[p]$ and $\sum_{i=1}^{p}s_i\le2k\gamma$. By the bounds established in Step~2, it suffices to consider $-2k\gamma\le\Delta A\le0$  and $-2k\gamma\le\Delta c_\ell\le4k\gamma$  for every $\ell\in[\gamma]$.

\medskip 
\noindent \emph{Initial state.} Before processing any path, the table contains only the state \[ \left( (0)_{i\in[p]}, (0)_{i\in[p]}, 0, (0)_{\ell\in[\gamma]} \right). \] 

\medskip 
\noindent \emph{Transitions.} Suppose that \[ \zeta = \left( (\alpha_i)_{i\in[p]}, (s_i)_{i\in[p]}, \Delta A, (\Delta c_\ell)_{\ell\in[\gamma]} \right) \] is reachable after processing $P_1,\ldots,P_{j-1}$. We distinguish two possibilities for $P_j$.

\smallskip 
\noindent \emph{Unaffected transition.} We may select no certificate vertex from $P_j$. The path then retains its baseline contribution, so the state remains unchanged.

\smallskip 
\noindent \emph{Affected transition.} Alternatively, choose a signature \[ \sigma = \left( (\widehat{s}_i)_{i\in[p]}, (\widehat{\alpha}_i)_{i\in[p]}, \delta, (d_\ell)_{\ell\in[\gamma]} \right) \in\Sigma(P_j). \] The signature certifies the existence of an affected configuration of $P_j$ with the effect recorded by $\sigma$. 
For every $i\in[p]$, set \[ s'_i:=s_i+\widehat{s}_i \qquad\text{and}\qquad \alpha'_i := \min\{\gamma,\alpha_i+\widehat{\alpha}_i\}. \]
The transition is allowed only if $s'_i\le2\gamma$ for every $i\in[p]$ and $\sum_{i=1}^{p}s'_i\le2k\gamma$. We further set $\Delta A':=\Delta A+\delta$  and, for every $\ell\in[\gamma]$, $\Delta c'_\ell:=\Delta c_\ell+d_\ell$.  The resulting state is retained only if all entries lie in the prescribed ranges.

\medskip 
\noindent \emph{Interpretation of a final state.} Let \[ \zeta = \left( (\alpha_i)_{i\in[p]}, (s_i)_{i\in[p]}, \Delta A, (\Delta c_\ell)_{\ell\in[\gamma]} \right) \] be reachable after all paths in $\mathcal{P}$ have been processed.

For every $i\in[p]$, let $R_i$ denote the certificate vertices assigned to the $i$-th special part by the configurations represented by $\zeta$. The set $X_i\cup R_i$ induces at least \[ \min\{\gamma,\,|E(G[X_i\cup R_i])|\} = \min\{\gamma,\,|E(G[X_i])|+\alpha_i\} \] edges, up to truncation at $\gamma$.
Accordingly, we call $\zeta$ \emph{special-feasible} if \[ |E(G[X_i])|+\alpha_i\ge \gamma \] for every $i\in[p]$. 
A final state that is not special-feasible is rejected. For a special-feasible state, define  $A:=A_0+\Delta A$.  
By the definition of the corrections, $A$ is the total number of completed $(\gamma+1)$-vertex blocks produced by the ordinary path segments. By exhaustive application of Reduction Rule~\ref{paths+u RR}, each such block forms one feasible non-special part. For every $\ell\in[\gamma]$, define \[ c_\ell:=c^0_\ell+\Delta c_\ell. \] We discard the state if $c_\ell<0$ for some $\ell\in[\gamma]$. Otherwise, $c_\ell$ is the number of residual path components on exactly $\ell$ vertices after all complete blocks have been extracted.

Let  $H_\zeta := \bigcup_{\ell=1}^{\gamma}c_\ell P_\ell$.  Thus, $H_\zeta$ consists of exactly $c_\ell$ pairwise vertex-disjoint copies of $P_\ell$ for every $\ell\in[\gamma]$. Every connected component of $H_\zeta$ has at most $\gamma$ vertices, and hence $\vi(H_\zeta)\le \gamma$.  Moreover,  $|V(H_\zeta)| = \sum_{\ell=1}^{\gamma}\ell c_\ell \le n$,  so $H_\zeta$ can be constructed in polynomial time.

\medskip
\noindent \emph{Acceptance conditions.}
Let  $q:=r-p$  be the required number of non-special parts.

\smallskip 
\noindent \emph{Case 1: $q=0$.} 
All required parts are special. Since $\zeta$ is special-feasible, the sets $X_i\cup R_i$, $i\in[p]$, provide $r=p$ feasible parts. Every remaining vertex can be assigned arbitrarily to one of these parts, since adding vertices cannot decrease the number of induced edges. We therefore accept the state.

\smallskip 
\noindent \emph{Case 2: $q>0$ and $A\ge q$.} Choose any $q$ of the $A$ completed blocks as initial non-special parts. 
Together with the $p$ special parts, they give exactly $p+q=r$ feasible parts. 
Assign the vertices of every additional completed block and every vertex of $H_\zeta$ arbitrarily among these parts. 
Since the selected parts already induce at least $\gamma$ edges, this preserves feasibility. We therefore accept the state. 

\smallskip
\noindent \emph{Case 3: $q>0$ and $A<q$.} 
Use the $A$ completed blocks as $A$ feasible non-special parts, and define $q':=q-A$.
The residual graph $H_\zeta$ must provide the remaining $q'$ non-special parts. We therefore invoke the algorithm for \textsc{ECGP} parameterized by vertex integrity on the instance $(H_\zeta,q',\gamma)$.  Since $\vi(H_\zeta)\le \gamma$, this invocation is fixed-parameter tractable in $\gamma$. We accept $\zeta$ if and only if $(H_\zeta,q',\gamma)$ is a YES-instance. If the residual instance is a YES-instance, its $q'$ feasible parts, together with the $A$ completed blocks and the $p$ special parts, give exactly  $p+A+q' = p+A+(q-A) = p+q = r$ feasible parts. The branch $X=X_1\dot\cup\cdots\dot\cup X_p$ is accepted if at least one reachable final state is special-feasible and satisfies one of the above acceptance conditions. The original instance is accepted if at least one branch is accepted.

\begin{lemma}
\label{lem:vdp-global-dp-interpretation}
Let $\zeta$ be a reachable final state. Then $A=A_0+\Delta A$ is exactly the number of completed feasible blocks produced by the ordinary path segments represented by $\zeta$, and $c_\ell=c^0_\ell+\Delta c_\ell$ is exactly the number of residual paths on $\ell$ vertices for every $\ell\in[\gamma]$.
\end{lemma} 
\begin{proof} 
We proceed by induction over the processed path components. Before any path is processed, all corrections are zero, so the statement holds. Suppose that it holds after processing $P_1,\ldots,P_{j-1}$. If $P_j$ is unaffected, its actual contribution equals its baseline contribution, and no correction is added. If $P_j$ is affected, the selected signature records precisely the difference between its actual contribution and its baseline contribution: $\delta$ is the correction to the number of completed blocks, and $d_\ell$ is the correction to the number of residual paths on $\ell$ vertices. Adding these values to the accumulated corrections therefore preserves the stated interpretation. The result follows after all paths have been processed. 
\end{proof}

\medskip 
\noindent \emph{Running time of the global dynamic program.}
The edge count vector has at most $(\gamma+1)^p$ possibilities, and the certificate-count vector has at most $(2\gamma+1)^p$ possibilities. The block correction has at most $2k\gamma+1$ possible values, while each of the $\gamma$ remainder-correction coordinates has at most $6k\gamma+1$ possible values. Therefore, the number of global states is at most $(\gamma+1)^p(2\gamma+1)^p(2k\gamma+1)(6k\gamma+1)^\gamma$. 
Since $p\le k$, this quantity is bounded by a function of $k+\gamma$. For every path $P$, the set $\Sigma(P)$ has size bounded by a function of $k+\gamma$. Hence the global dynamic program runs in FPT time.   The number of final states is bounded by a function of $k+\gamma$. Each residual instance $H_\zeta$ has at most $n$ vertices and vertex integrity at most $\gamma$. Therefore, by Theorem~\ref{thm:uccp-vin}, all residual instances can be processed in FPT time. It remains to prove the soundness and completeness of the algorithm.

\medskip
\noindent
\textbf{Step 5: Correctness and running time.}
We prove that the algorithm accepts if and only if $(G,r,\gamma)$ is a YES-instance.

\medskip
\noindent \emph{Soundness.} 
Suppose that the algorithm accepts a branch $X=X_1\dot\cup\cdots\dot\cup X_p$ through a special-feasible final state \[ \zeta = \left( (\alpha_i)_{i\in[p]}, (s_i)_{i\in[p]}, \Delta A, (\Delta c_\ell)_{\ell\in[\gamma]} \right). \] For every affected path, fix a configuration realizing the signature selected by the global dynamic program. For each $i\in[p]$, let $R_i$ be the set of certificate vertices assigned to the $i$-th special part, and define $S_i:=X_i\cup R_i$. The sets $S_1,\ldots,S_p$ are pairwise disjoint.

Let $\mu_i$ be the total, untruncated utility contribution of the certificate vertices in $R_i$. Since $\alpha_i=\min\{\gamma,\mu_i\}$ and $\zeta$ is special-feasible,  $|E(G[S_i])| = |E(G[X_i])|+\mu_i \ge \gamma$.  Thus $S_1,\ldots,S_p$ are feasible special parts.

By Lemma~\ref{lem:vdp-global-dp-interpretation}, the ordinary path segments produce exactly $A:=A_0+\Delta A$ pairwise vertex-disjoint completed blocks, each inducing exactly $\gamma$ edges, together with the residual graph $ H_\zeta = \bigcup_{\ell=1}^{\gamma}c_\ell P_\ell$, and $c_\ell:=c^0_\ell+\Delta c_\ell$.  
Let $q:=r-p$. If $q=0$, the sets $S_1,\ldots,S_p$ already provide all $r$ required parts. Assign every remaining vertex arbitrarily to one of them.

Suppose that $q>0$. If $A\ge q$, choose any $q$ completed blocks as non-special parts. Together with the $p$ special parts, they give exactly $r$ feasible parts. Assign all vertices of the additional completed blocks and $H_\zeta$ arbitrarily to these parts. Finally, suppose that $A<q$. The algorithm accepts only if $(H_\zeta,q-A,\gamma)$ is a YES-instance. Let $T_1,\ldots,T_{q-A}$ be a corresponding feasible partition of $H_\zeta$. The special parts $S_1,\ldots,S_p$, the $A$ completed blocks, and the parts $T_1,\ldots,T_{q-A}$ are pairwise disjoint, feasible, and their total number is \[ p+A+(q-A)=r. \] The certificate vertices belong to the special parts, while every ordinary vertex belongs either to a completed block or to $H_\zeta$. Hence these sets partition $V(G)$. Therefore $(G,r,\gamma)$ is a YES-instance.

\medskip \noindent \emph{Completeness.} Suppose that $(G,r,\gamma)$ admits a feasible partition $V(G)=V_1\dot\cup\cdots\dot\cup V_r$. After relabeling, let $V_1,\ldots,V_p$ be precisely the parts intersecting $X$, and let $X_i:=V_i\cap X$ for every $i\in[p]$. The algorithm considers the branch corresponding to \[ X=X_1\dot\cup\cdots\dot\cup X_p. \] 
By Lemma~\ref{lem:vdp-special-part-normalization}, for every $i\in[p]$ there exists a set $R_i\subseteq V_i\setminus X$ such that \[ |R_i|\le2\gamma \qquad\text{and}\qquad |E(G[X_i\cup R_i])|\ge \gamma. \] Assign every vertex of $R_i$ to the $i$-th special part in the corresponding path configuration, and declare every other vertex of $G-X$ ordinary. The total number of selected certificate vertices is at most $2p\gamma\le2k\gamma$.

Let $q:=r-p$. If $q=0$, the corresponding final state is special-feasible and is accepted by the first acceptance condition. Assume that $q>0$. Every vertex of  $\bigcup_{i=1}^{p} \bigl(V_i\setminus(X_i\cup R_i)\bigr) $ may be moved to an arbitrary non-special part. 
Each special part retains its certificate $X_i\cup R_i$, and adding vertices to a non-special part cannot decrease its utility. Consequently, the vertices outside the special certificates admit a partition into exactly $q$ feasible non-special parts. The restriction of this assignment to each path $P\in\mathcal{P}$ is either unaffected or defines an affected configuration of $P$.

By Lemma~\ref{lem:vdp-local-signatures}, its signature belongs to $\Sigma(P)$. The global dynamic program can therefore select the corresponding transition for every path and reaches a special-feasible final state $\zeta$. After fixing the certificate vertices, the residual graph is a disjoint union of ordinary path segments. 
Applying Reduction Rule~\ref{paths+u RR} exhaustively to these segments. By Lemma~\ref{lem:vdp-global-dp-interpretation}, these applications produce exactly $A=A_0+\Delta A$ completed feasible parts and leave precisely $H_\zeta$. If $A\ge q$, the algorithm accepts by the second acceptance condition. If $A<q$, correctness of the reduction rule implies that $(H_\zeta,q-A,\gamma)$ is a YES-instance. The vertex-integrity algorithm therefore returns YES, and the algorithm accepts by the third acceptance condition. Thus every YES-instance is accepted.

\medskip 
\noindent \emph{Running time.}
The number of partitions of $X$ considered in Step~1 is bounded by $k^k$. For each branch, Step~3 computes all local signature sets in FPT time. By Step~4, the number of global states is at most \[ (\gamma+1)^p(2\gamma+1)^p(2k\gamma+1)(6k\gamma+1)^\gamma, \] which is bounded by a function of $k+\gamma$ because $p\le k$. Since every path has a number of realizable signatures bounded by a function of $k+\gamma$, the global dynamic program runs in FPT time. For every final state, the residual graph $H_\zeta$ has at most $n$ vertices and satisfies $\vi(H_\zeta)\le \gamma$. The number of final states is bounded by a function of $k+\gamma$. Hence, by Theorem~\ref{thm:uccp-vin}, all calls to the vertex-integrity algorithm together take FPT time. Including the branching over partitions of $X$, the total running time is $?$. Therefore, \textsc{ECGP} is fixed-parameter tractable when parameterized by $\vdp+\gamma$. This completes the proof of Theorem~\ref{thm:vdp-plus-u}.
\end{proof}

\begin{theorem}
\label{thm:vds-plus-u}
\textsc{ECGP} is fixed-parameter tractable when parameterized by $\mathsf{vds}+\gamma$.
\end{theorem}

\begin{proof}
Let $(G,r,\gamma)$ be an instance of \textsc{ECGP}, and let $X$ be a minimum vertex deletion set to a disjoint union of stars, with $|X|=k$. Thus every connected component of $G-X$ is a star.

\begin{uccp}\label{rule:vds-leaves}
Let $X$ be a vertex deletion set to a disjoint union of stars, with $|X|=k$. Let $S$ be a star of $G-X$ with center $c$. For a fixed subset $T \subseteq X$, let $L_T$ be the set of leaves of $S$ whose neighborhood in $X$ is exactly $T$. If $|L_T|>(k+1)\gamma$, then delete arbitrary vertices from $L_T$ until exactly $(k+1)\gamma$ remain.
\end{uccp}

\begin{lemma}\label{lem:vds+u RR}
Reduction Rule~\ref{rule:vds-leaves} is correct.
\end{lemma}
\begin{proof}
Let $(G,r,\gamma)$ be an instance, and let $(G',r,\gamma)$ be the instance obtained after applying the rule to some set $L_T$. We prove that $(G,r,\gamma)$ is a YES-instance if and only if $(G',r,\gamma)$ is a YES-instance.

\medskip
\noindent
\emph{Forward direction.}
Suppose that $(G,r,\gamma)$ is a YES-instance, and let $C_1,\dots,C_r$ be a feasible partition of $V(G)$.

We first show that we may assume, without loss of generality, that every vertex of $L_T$ belongs to a part containing at least one of its neighbors. Indeed, every vertex of $L_T$ is adjacent precisely to the center $c$ and to the vertices of $T$. If some vertex $v \in L_T$ belongs to a part containing none of $c$ and none of the vertices of $T$, then $v$ is isolated in that part. Removing $v$ from that part does not decrease its number of induced edges. Now place $v$ into any part containing either $c$ or a vertex of $T$; this cannot decrease the number of induced edges in that part. Repeating this operation, we obtain a feasible partition in which every vertex of $L_T$ lies in a part containing at least one of its neighbors.

Now consider such a feasible partition. Any part containing a vertex of $L_T$ must contain either $c$ or a vertex of $T$. Since the parts are pairwise disjoint, there is at most one part containing $c$, and for each vertex of $T$ there is at most one part containing that vertex. Hence the number of parts that may contain vertices of $L_T$ is at most $|T|+1 \le |X|+1 = k+1$.

Fix one such part $C_j$. Each vertex of $L_T \cap C_j$ contributes at least one edge inside $C_j$, since it is adjacent only to $c$ and to the vertices of $T$. Hence the total number of edges in $G[C_j]$ that are incident to vertices of $L_T \cap C_j$ is at least $|L_T \cap C_j|$.
Since $C_j$ is feasible, we have $|E(G[C_j])| \ge \gamma$. Therefore, among the vertices of $L_T \cap C_j$, at most $\gamma$ are needed to contribute up to $\gamma$ edges toward this bound. Consequently, if $|L_T \cap C_j| > \gamma$, we may delete arbitrary vertices from $L_T \cap C_j$ until exactly $\gamma$ remain, without violating the condition $|E(G[C_j])| \ge \gamma$.

Applying this to every part containing vertices of $L_T$, we obtain a feasible partition in which at most $\gamma$ vertices of $L_T$ are used in each of at most $k+1$ parts. Thus, in total, at most $(k+1)\gamma$ vertices of $L_T$ are needed.
Hence the deleted vertices from $L_T$ are not needed for feasibility, and removing them yields a feasible partition of $G'$. Therefore $(G',r,\gamma)$ is a YES-instance.

\medskip
\noindent
\emph{Reverse direction.}
Suppose that $(G',r,\gamma)$ is a YES-instance. Since $G'$ is an induced subgraph of $G$, the same partition of $V(G')$ into $r$ feasible parts can be extended to a partition of $V(G)$ by placing all deleted vertices arbitrarily into any one of the existing parts. This does not decrease the number of induced edges in any part. Hence $(G,r,\gamma)$ is also a YES-instance.

\medskip

Thus $(G,r,\gamma)$ is a YES-instance if and only if $(G',r,\gamma)$ is a YES-instance, proving correctness of the rule.
\end{proof}

Apply the reduction rule exhaustively and let $G'$ be the resulting graph.
Consider any star $S$ in $G'-X$ with center $c$. The leaves of $S$ are partitioned into at most $2^k$ classes according to their neighborhoods in $X$, and each class has size at most $k\gamma+\gamma$. Hence $|V(S)| \le 1 + 2^k(k\gamma+\gamma)$.
Therefore $\mathsf{vi}(G') \le k + 1 + 2^k(k\gamma+\gamma)$.
Since \textsc{ECGP} is fixed-parameter tractable when parameterized by $\mathsf{vi}$, the reduced instance can be solved in time $f(k,\gamma)\cdot n^{\mathcal{O}(1)}$ for some computable function $f$.
Therefore \textsc{ECGP} is fixed-parameter tractable when parameterized by $\mathsf{vds}+\gamma$.
\end{proof}

\subsection{Algorithms via Dynamic Programming}

Next, we provide a dynamic programming algorithm for \textsc{ECGP} parameterized by $\mathsf{tw}+r+\gamma$.

\begin{theorem}
\label{thm:tw-dp-general}
\textsc{ECGP} can be solved in time
$r^{\mathcal{O}(k)} (\gamma+1)^{\mathcal{O}(r)} \cdot n$.
More precisely, given a nice tree decomposition of width $k$, the problem can be solved in time
$\mathcal{O}\!\left(r^{k+1} \cdot (\gamma+1)^{2r} \cdot n\right)$.
\end{theorem}
\begin{proof}
Let $(G=(V,E),r,\gamma)$ be an instance, and let $(T,\mathcal{B})$ be a nice tree decomposition of $G$ of width $k$. We use the standard node types: leaf, introduce-vertex, introduce-edge, forget, and join. We assume that the root bag is empty.

For a node $t \in V(T)$, let $B_t$ denote its bag. Let $V_t$ be the set of vertices appearing in bags of the subtree rooted at $t$, and let $E_t$ be the set of edges introduced in this subtree. Thus every edge of $G$ is introduced exactly once, and $G_t := (V_t,E_t)$ is the processed subgraph at node $t$.
We work with $r$ \emph{labeled} parts, indexed by $[r] := \{1,\ldots,r\}$. This is without loss of generality, since any feasible partition into $r$ unlabeled parts can be labeled arbitrarily.\\

\noindent \textbf{States.}
For every node $t$, we store a table indexed by tuples $(\chi,\alpha)$, where:
\begin{itemize}
    \item $\chi : B_t \to [r]$ assigns each vertex currently present in the bag to one of the $r$ parts;
    \item $\alpha = (\alpha_1,\ldots,\alpha_r) \in \{0,1,\ldots,\gamma\}^r$, where $\alpha_i$ stores the number of edges already introduced inside part $i$, truncated at $\gamma$.
\end{itemize}

We use the truncated addition operation $a \oplus_\gamma b := \min\{\gamma,a+b\}$ for $a,b \in \{0,\ldots,\gamma\}$, and extend it coordinatewise to vectors in $\{0,\ldots,\gamma\}^r$. \\

\noindent \textbf{Meaning of a state.}
We set $\mathrm{DP}_t[\chi,\alpha]=1$ if and only if there exists a mapping $\pi_t : V_t \to [r]$ such that:
\begin{enumerate}
    \item $\pi_t(v)=\chi(v)$ for every $v \in B_t$, and
    \item for every $i \in [r]$, the number of edges of $E_t$ whose endpoints are both mapped to part $i$ is, after truncation at $\gamma$, equal to $\alpha_i$.
\end{enumerate}

In other words, $\alpha_i$ records how many monochromatic edges of part $i$ have already been seen in the processed subgraph, but once the value reaches $\gamma$ we do not distinguish larger values.
Since $|B_t| \le k+1$, the number of possible functions $\chi$ is at most $r^{k+1}$. Hence the number of states per node is at most $r^{k+1}(\gamma+1)^r$.\\

\noindent \textbf{Leaf node.}
If $t$ is a leaf node, then $B_t=\emptyset$ and $V_t=\emptyset$. We set
$\mathrm{DP}_t[\emptyset,(0,\ldots,0)] = 1$,
and all other entries to $0$.\\

\noindent \textbf{Introduce-vertex node.}
Suppose $t$ is an introduce-vertex node with child $t'$, and $B_t = B_{t'} \cup \{v\}$.
For every state $(\chi',\alpha)$ with $\mathrm{DP}_{t'}[\chi',\alpha]=1$ and every color $c \in [r]$, define $\chi := \chi' \cup \{(v,c)\}$. Then set
$\mathrm{DP}_t[\chi,\alpha]=1$.
No edge is introduced at this step, so the vector $\alpha$ remains unchanged.\\

\noindent \textbf{Introduce-edge node.}
Suppose $t$ is an introduce-edge node with child $t'$, and $B_t=B_{t'}$, while the introduced edge is $e=\{x,y\}$ with $x,y \in B_t$.
For every state $(\chi,\alpha)$ with $\mathrm{DP}_{t'}[\chi,\alpha]=1$, define a new vector $\alpha'$ by
\[
\alpha'_i :=
\begin{cases}
\alpha_i \oplus_\gamma 1 & \text{if } i=\chi(x)=\chi(y),\\
\alpha_i & \text{otherwise.}
\end{cases}
\]
Then set $\mathrm{DP}_t[\chi,\alpha']=1$.
Thus, if the endpoints of the introduced edge lie in the same part, we increase the corresponding counter by one, truncated at $\gamma$.\\

\noindent \textbf{Forget node.}
Suppose $t$ is a forget node with child $t'$, and $B_t = B_{t'} \setminus \{v\}$.
For every state $(\chi,\alpha)$ at node $t$, we set $\mathrm{DP}_t[\chi,\alpha]=1$ if and only if there exists a color $c \in [r]$ such that
$\mathrm{DP}_{t'}[\chi',\alpha]=1$,
where $\chi'$ is the extension of $\chi$ to $B_{t'}$ defined by
$\chi'(v)=c$ and $\chi'(u)=\chi(u)$ for all $\gamma \in B_t$.
Here we simply forget the assignment of $v$ while keeping all accumulated edge counts.\\

\noindent \textbf{Join node.}
Suppose $t$ is a join node with children $t_1,t_2$, and $B_t = B_{t_1}=B_{t_2}$.
For a state $(\chi,\alpha)$, we set $\mathrm{DP}_t[\chi,\alpha]=1$ if and only if there exist vectors $\alpha^{(1)},\alpha^{(2)} \in \{0,\ldots,\gamma\}^r$ such that:
\begin{itemize}
    \item $\mathrm{DP}_{t_1}[\chi,\alpha^{(1)}]=1$,
    \item $\mathrm{DP}_{t_2}[\chi,\alpha^{(2)}]=1$, and
    \item $\alpha = \alpha^{(1)} \oplus_u \alpha^{(2)} \ominus d_\chi$ coordinatewise,
\end{itemize}
where $d_\chi[i]$ is the number of edges $\{x,y\} \in E(G[B_t])$ such that $\chi(x)=\chi(y)=i$.\\

\noindent \textbf{Correctness.}
At the root $r$, the bag is empty. Hence the instance is a YES-instance if and only if
$\mathrm{DP}_r[\emptyset,\alpha]=1$ for some vector $\alpha$ with $\alpha_i=\gamma$ for every $i \in [r]$.
Indeed, $\alpha_i=\gamma$ means that part $i$ induces at least $\gamma$ edges.\\

\noindent \textbf{Running time.}
At every node there are at most $r^{k+1}(\gamma+1)^r$ states. Leaf, introduce-vertex, introduce-edge, and forget nodes can be processed within this bound up to polynomial factors. The expensive step is the join node: for a fixed assignment $\chi$, we combine pairs of vectors $\alpha^{(1)},\alpha^{(2)} \in \{0,\ldots,\gamma\}^r$, giving a factor of $(\gamma+1)^{2r}$.
Hence the total running time is
$\mathcal{O}\!\left(n \cdot r^{k+1} \cdot (\gamma+1)^{2r}\right)$.
\end{proof}

Since $\gamma \le \binom{n}{2}$ in any instance, we have $(\gamma+1)^{2r} \le n^{\mathcal{O}(r)}$. 
Moreover, for \textsc{BECGP}, the dynamic programming can be extended by additionally storing, for each part, the number of assigned vertices (bounded by $n/r$), incurring an additional $n^{\mathcal{O}(r)}$ factor. 
The same approach also extends to signed variants by storing edge-count values.

\begin{corollary}
\textsc{ECGP}, \textsc{BECGP}, and their signed variants belong to XP when parameterized by $\mathsf{tw}+r$.
\end{corollary}

\begin{corollary} \label{thm:fpt-ru-via-dp} 
\textsc{ECGP} can be solved in time $n^{\mathcal{O}(1)} + \mathcal{O}\left( r\gamma^2 \cdot r^{2r\gamma}\cdot (\gamma+1)^{2r} \right)$,  and \textsc{BECGP} can be solved in time  $n^{\mathcal{O}(1)} + \mathcal{O}\left( (r\gamma^2+r^2)\cdot r^{2r\gamma}\cdot (\gamma+1)^{2r} \cdot (r\gamma^2+r^2+1)^{2r} \right)$. 
\end{corollary}  
\begin{proof} 
We first apply the corresponding kernelization algorithm. For \textsc{ECGP}, this produces an equivalent instance $(G',r,\gamma)$ with $\mathcal{O}(r\gamma^2)$ vertices, whereas for \textsc{BECGP} it produces an equivalent instance with $\mathcal{O}(r\gamma^2+r^2)$ vertices. 
The time required for kernelization is polynomial in $n$. 
In both kernelization algorithms, if the input graph contains a matching of size at least $r\gamma$, the instance is decided directly. We may therefore assume that $G'$ contains no matching of size $r\gamma$. Let $M$ be a maximal matching of $G'$, and let $C$ be the set of endpoints of the edges of $M$. Then $|M|<r\gamma$, and hence $|C|=2|M|<2r\gamma$. Since $M$ is maximal, $C$ is a vertex cover of $G'$. Consequently, $\tw(G')\le |C|<2r\gamma$.

For \textsc{ECGP}, the dynamic programming algorithm of Theorem~\ref{thm:tw-dp-general} solves an instance with $N$ vertices and a nice tree decomposition of width $k$ in time $\mathcal{O}\left( N\cdot r^{k+1}\cdot (\gamma+1)^{2r} \right)$.  Since $k<2r\gamma$ and both quantities are integers, we have $k+1\le 2r\gamma$. Substituting $N=\mathcal{O}(r\gamma^2)$ therefore gives a running time of  \[n^{\mathcal{O}(1)} + \mathcal{O}\left( r\gamma^2\cdot r^{2r\gamma}\cdot (\gamma+1)^{2r} \right).\]  

For \textsc{BECGP}, let $B:=N/r$ denote the required size of every part. We extend each dynamic programming state by a vector $\beta=(\beta_1,\ldots,\beta_r)\in\{0,\ldots,B\}^r$, where $\beta_i$ records the number of vertices assigned to part $i$ in the processed subgraph. 
At the root, we accept only states satisfying $\beta_i=B$ for every $i\in[r]$. The resulting dynamic programming algorithm runs in time  $\mathcal{O}\left( N\cdot r^{k+1}\cdot (\gamma+1)^{2r}\cdot (B+1)^{2r} \right)$.  The factor $(B+1)^{2r}$ accounts for combining pairs of size vectors at join nodes. For the reduced \textsc{BECGP} instance, we have $N=\mathcal{O}(r\gamma^2+r^2)$ and $B\le N$. Together with $k+1\le 2r\gamma$, this gives the running time \[ n^{\mathcal{O}(1)} + \mathcal{O}\left( (r\gamma^2+r^2)\cdot r^{2r\gamma}\cdot (\gamma+1)^{2r} \cdot (u^2+r+1)^{2r} \right). \] \end{proof}

\section{Hardness Results}

In this section, we establish hardness results for \textsc{ECGP} and \textsc{BECGP}. We use parameterized reductions from known problems, including \textsc{Tight Unary Bin Packing}, as described below.

\begin{center}
\fbox{%
\parbox{0.97\linewidth}{%
\textsc{Tight Unary Bin Packing}

\medskip

\textbf{Input:} Positive integers $a_1,\ldots,a_n$, $B$, and $r$, where all integers are given in unary and
$\sum_{i=1}^{n} a_i = rB$.

\textbf{Question:} Can the items be partitioned into $r$ bins such that the total weight of the items assigned to each bin is exactly $B$?
}}
\end{center}

Jansen et al. \cite{10.1007/978-3-642-13731-0_25} proved that the \textsc{Tight Unary Bin Packing} is W[1]-hard when parameterized by number of bins, that is, $r$. 
For a graph class $\mathcal{H}$, let $\mathsf{vdd}_{\mathcal{H}}(G)$ denote the minimum number of vertices whose removal transforms $G$ into a graph belonging to $\mathcal{H}$.
We next present a general reduction framework that yields W[1]-hardness for \textsc{ECGP} under a broad class of graph families.

\begin{theorem}
\label{thm:general-hardness-class}
Let $\mathcal{H}$ be a graph class such that:
\begin{enumerate}
    \item $\mathcal{H}$ is closed under disjoint union, and
    \item for every integer $q \ge 1$, one can construct in polynomial time a connected graph $H_q \in \mathcal{H}$ with exactly $q$ edges.
\end{enumerate}
Then \textsc{ECGP} is W[1]-hard when parameterized by $r+\mathrm{vdd}_{\mathcal{H}}$, even when $\mathrm{vdd}_{\mathcal{H}}=0$.
\end{theorem}
\begin{proof}
We give a parameterized reduction from the tight version of \textsc{Unary Bin Packing}, parameterized by the number $r$ of bins. An instance consists of positive integers $a_1,\dots,a_n$, a bin capacity $B$, and an integer $r$, and asks whether the items can be partitioned into $r$ bins such that the total size in every bin is at most $B$. We may assume without loss of generality that the instance is tight, that is, $\sum_{i=1}^n a_i = rB$.

Given such an instance, we construct an instance $(G,r,\gamma)$ of \textsc{ECGP} as follows. Set $\gamma := B$. For every item $a_i$, construct the connected graph $H_{a_i} \in \mathcal{H}$ having exactly $a_i$ edges, and let $G$ be the disjoint union of the graphs $H_{a_1},\dots,H_{a_n}$.

We claim that the \textsc{Unary Bin Packing} instance is a YES-instance if and only if $(G,r,\gamma)$ is a YES-instance.

\medskip
\noindent
\emph{Forward direction.}
Suppose that the items can be packed into $r$ bins $X_1,\dots,X_r$ such that for every $j \in [r]$, $\sum_{i \in X_j} a_i \le B$. Since the instance is tight, we in fact have $\sum_{i \in X_j} a_i = B$ for every $j \in [r]$.

For each bin $X_j$, let $V_j$ be the union of the vertex sets of the graphs $H_{a_i}$ with $i \in X_j$. Then $V_1,\dots,V_r$ is a partition of $V(G)$. Moreover, $G[V_j]$ contains exactly $\sum_{i \in X_j} a_i = B = \gamma$ edges. Hence each part induces at least $\gamma$ edges, and therefore $(G,r,\gamma)$ is a YES-instance.

\medskip
\noindent
\emph{Reverse direction.}
Suppose that $(G,r,\gamma)$ is a YES-instance, and let $V_1,\dots,V_r$ be a partition of $V(G)$ such that each $G[V_j]$ induces at least $\gamma=B$ edges.

Since $G$ is the disjoint union of the connected graphs $H_{a_1},\dots,H_{a_n}$, every edge of $G$ belongs to exactly one such graph. Furthermore,
$|E(G)|=\sum_{i=1}^n a_i = rB = r\gamma$.
Hence the total number of edges available in $G$ is exactly the minimum total number of edges required by the $r$ parts. It follows that every part must induce exactly $B$ edges, and no edge can be lost.

Now consider one connected gadget $H_{a_i}$. If its vertex set is split between two different parts, then since $H_{a_i}$ is connected, there exists an edge of $H_{a_i}$ with endpoints in different parts. Such an edge is not counted in any induced subgraph $G[V_j]$, contradicting the fact that no edge can be lost. Therefore each graph $H_{a_i}$ is entirely contained in a single part.
Consequently, each part corresponds to a collection of whole gadgets, and the number of edges in that part is exactly the sum of the corresponding item sizes. Since every part induces exactly $B$ edges, these collections define a packing of the items into $r$ bins of capacity $B$. Thus the \textsc{Unary Bin Packing} instance is a YES-instance.

\medskip

This proves correctness of the reduction. Since $\mathcal{H}$ is closed under disjoint union and each $H_{a_i}$ belongs to $\mathcal{H}$, the graph $G$ itself belongs to $\mathcal{H}$. Hence $\mathrm{vdd}_{\mathcal{H}}(G)=0$. Moreover, the reduction runs in polynomial time and maps the parameter $r$ to $r+\mathrm{vdd}_{\mathcal{H}}(G)=r$. Therefore this is a parameterized reduction. Since tight \textsc{Unary Bin Packing} is W[1]-hard parameterized by $r$, it follows that \textsc{ECGP} is W[1]-hard when parameterized by $r+\mathrm{vdd}_{\mathcal{H}}$, even when $\mathrm{vdd}_{\mathcal{H}}=0$.
\end{proof}

\noindent
Applying the theorem to disjoint unions of stars and paths respectively yields the following.

\begin{corollary}
\label{cor:vds-hard}
\textsc{ECGP} is W[1]-hard when parameterized by $r+\mathsf{fes}+\mathsf{vds}+\mathsf{mw}$, even when $\mathsf{fes}=0$, $\mathsf{vds}=0$, and $\mathsf{mw}=2$.
\end{corollary}

\begin{corollary}
\label{cor:vdp-hard}
\textsc{ECGP} is W[1]-hard when parameterized by $r+\mathsf{fes}+\mathsf{vdp}+\Delta$, even when $\mathsf{fes}=0$, $\Delta=2$, and $\mathsf{vdp}=0$.
\end{corollary}

\noindent
We now prove W[1]-hardness for \textsc{ECGP} parameterized by $r+\mathsf{cvd}$. Since cliques have a fixed number of edges, namely $\binom{b}{2}$, they cannot directly represent arbitrary item sizes. To overcome this, we use the following decomposition, which allows us to express any integer as a clique contribution plus a small correction term.

\begin{lemma}\label{lem:number decomposition}
\label{lem:triangular-decomposition}
For every integer $a \ge 1$, there exist unique integers $b \ge 1$ and $x$ such that
$a=\binom{b}{2}+x$ and $0 \le x < b$.
\end{lemma}
\begin{proof}
Let $T_b:=\binom{b}{2}$. Since $(T_b)_{b \ge 1}$ is strictly increasing and unbounded, there exists a unique integer $b \ge 1$ such that
$T_b \le a < T_{b+1}$.
Set $x:=a-T_b$. Then $x \ge 0$, and
$x < T_{b+1}-T_b = b$.
Hence $a=\binom{b}{2}+x$ with $0 \le x < b$.

For uniqueness, suppose
$a=\binom{b}{2}+x=\binom{b'}{2}+x'$
with $0 \le x < b$ and $0 \le x' < b'$. Then
$\binom{b}{2} \le a < \binom{b+1}{2}$ and
$\binom{b'}{2} \le a < \binom{b'+1}{2}$,
which implies $b=b'$, and therefore $x=x'$.
\end{proof}

\begin{theorem}
\label{thm:k-cvd-hard}
\textsc{ECGP} is W[1]-hard when parameterized by $r+\mathsf{cvd}$.
\end{theorem}

\begin{proof}
We give a parameterized reduction from the  \textsc{Tight Unary Bin Packing}, parameterized by the number $r$ of bins. 
By Lemma~\ref{lem:triangular-decomposition}, for every $i \in [n]$ there exist unique integers $b_i \ge 1$ and $x_i$ such that
$a_i=\binom{b_i}{2}+x_i$ and $0 \le x_i < b_i$.

\medskip
\noindent
\emph{Construction.}
Create $r$ vertices $z_1,\dots,z_r$, called \emph{bin vertices}. For every $j \in [r]$, attach exactly $\alpha$ pendant leaves to $z_j$, where
$\alpha := 1+\sum_{i=1}^n \binom{b_i}{2}$.
For every item $a_i$, create a clique $C_i$ on $b_i$ vertices. Choose an arbitrary set $D_i \subseteq V(C_i)$ of size $x_i$, and make every vertex of $D_i$ adjacent to all bin vertices $z_1,\dots,z_r$.
Set the number of parts to be $r$, and set the edge threshold to $\gamma:=\alpha+B$.
We claim that the \textsc{Unary Bin Packing} instance is a YES-instance if and only if $(G,r,\gamma)$ is a YES-instance.

\medskip
\noindent
\emph{Forward direction.}
Suppose that the items can be packed into $r$ bins $X_1,\dots,X_r$ such that $\sum_{i \in X_j} a_i \le B$ for every $j \in [r]$. Since the instance is tight, we have $\sum_{i \in X_j} a_i = B$ for every $j \in [r]$.

For every $j \in [r]$, let $V_j$ consist of:
\begin{itemize}
    \item the bin vertex $z_j$,
    \item all $\alpha$ pendant leaves adjacent to $z_j$, and
    \item all vertices of the cliques $C_i$ with $i \in X_j$.
\end{itemize}
Then $V_1,\dots,V_r$ is a partition of $V(G)$, and
$|E(G[V_j])|
=
\alpha + \sum_{i \in X_j}\left(\binom{b_i}{2}+x_i\right)
=
\alpha + \sum_{i \in X_j} a_i
=
\alpha+B
=
\gamma$
for every $j \in [r]$. Hence $(G,r,\gamma)$ is a YES-instance.

\medskip
\noindent
\emph{Reverse direction.}
Suppose that $(G,r,\gamma)$ is a YES-instance, and let $V_1,\dots,V_r$ be a partition of $V(G)$ such that each $G[V_j]$ induces at least $\gamma=\alpha+B$ edges.

We first show that every part contains exactly one bin vertex. Any set containing no bin vertex induces edges only inside the cliques $C_1,\dots,C_n$, and therefore induces at most $\sum_{i=1}^n \binom{b_i}{2} = \alpha-1$ edges. Since $\alpha-1<\gamma$, every part must contain at least one bin vertex. As there are exactly $r$ parts and exactly $r$ bin vertices, every part contains exactly one bin vertex.
Next, compute the total number of edges of $G$:

\begin{align*}
|E(G)|
&= r\alpha + \sum_{i=1}^n \left(\binom{b_i}{2}+rx_i\right)
 = r\alpha + \sum_{i=1}^n a_i + (r-1)\sum_{i=1}^n x_i 
 = r\gamma + (r-1)\sum_{i=1}^n x_i.
\end{align*}

Thus the graph contains exactly $(r-1)\sum_{i=1}^n x_i$ edges more than the minimum total number $r\gamma$ required by the $r$ parts.

For every $i \in [n]$ and every vertex $v \in D_i$, the vertex $v$ is adjacent to all $r$ bin vertices, while every part contains exactly one bin vertex. Therefore at most one of these $r$ edges can be internal to the part containing $v$, and at least $r-1$ of them are lost. Summing over all distinguished vertices, every feasible partition loses at least
$L_0 := (r-1)\sum_{i=1}^n x_i$
edges.

Since $|E(G)|=r\gamma+L_0$ and every feasible partition must realize at least $ru$ internal edges, no feasible partition can lose more than $L_0$ edges. Hence every feasible partition loses exactly $L_0$ edges, and all lost edges are precisely the unavoidable clique--bin edges incident with vertices in $D_1,\dots,D_n$.
We now show that no clique can be split.

\begin{claim}\label{claim:one bin vertex per cluster}
For every $j \in [r]$, all pendant leaves adjacent to $z_j$ lie in the same part as $z_j$.
\end{claim}
\begin{claimproof}
Let $v$ be a pendant leaf adjacent to $z_j$. Suppose $v$ lies in a part different from the one containing $z_j$. Then $v$ has no neighbors inside its part.

Moving $v$ to the part containing $z_j$ strictly increases the number of edges inside that part and does not decrease the number of edges in any other part. Hence we may assume that $v$ lies in the same part as $z_j$.
\end{claimproof}

\begin{claim}\label{Claim: cliques do not split}
For every $i \in [n]$, the clique $C_i$ is entirely contained in one part.
\end{claim}
\begin{claimproof}
Recall that every feasible partition loses exactly 
$L_0 = (r-1)\sum_{i=1}^n x_i$ edges, and these correspond precisely to the unavoidable edges between vertices of $D_1,\dots,D_n$ and bin vertices.

Consider a clique $C_i$. If $C_i$ is entirely assigned to the part containing some bin vertex $z_j$, then each vertex of $D_i$ contributes exactly $r-1$ lost edges (to the other bin vertices), and no edges inside $C_i$ are lost. Thus the total contribution of $C_i$ to the number of lost edges is exactly $(r-1)x_i$, which matches its share in $L_0$.

Suppose now that $C_i$ is split across at least two parts. Then, since $C_i$ is a clique, there exists at least one edge of $C_i$ whose endpoints lie in different parts. Such an edge is not counted in any part and is therefore lost. This loss is in addition to the unavoidable $(r-1)x_i$ edges contributed by the vertices in $D_i$.

Hence the total number of lost edges would be strictly greater than $L_0$, contradicting the fact that every feasible partition loses exactly $L_0$ edges.
Therefore, $C_i$ must be entirely contained in one part.
\end{claimproof}

By Claims~\ref{claim:one bin vertex per cluster} and \ref{Claim: cliques do not split}, every part consists of exactly one bin vertex, all its pendant leaves, and some whole item cliques. For every $j \in [r]$, let $I_j$ be the set of indices $i$ such that the clique $C_i$ lies in the same part as $z_j$. Then
\[
|E(G[V_j])|
=
\alpha + \sum_{i \in I_j}\left(\binom{b_i}{2}+x_i\right)
=
\alpha + \sum_{i \in I_j} a_i.
\]
Since the total number of induced edges over all parts is exactly $r\gamma$, every part must induce exactly $\gamma=\alpha+B$ edges. Therefore
$\sum_{i \in I_j} a_i = B
\qquad\text{for every } j \in [r]$.
Thus $I_1,\dots,I_r$ define a valid packing of the items into $r$ bins of capacity $B$.
This proves correctness of the reduction.

Finally, after deleting the $r$ bin vertices $z_1,\dots,z_r$, the remaining graph is a disjoint union of cliques, namely the cliques $C_i$ together with isolated vertices coming from the pendant leaves. Hence $\mathsf{cvd}(G)\le r$.
Thus $\mathsf{cvd}(G)=r$, and the parameter value of the constructed instance is exactly $2r$. The reduction is parameter preserving, so \textsc{ECGP} is W[1]-hard when parameterized by $r+\mathsf{cvd}$.
\end{proof}

\begin{theorem}\label{thm:uccp-to-buccp}
Let $\Pi$ be a graph parameter such that adding isolated vertices does not increase $\Pi$. Then there is a parameterized reduction from \textsc{ECGP} parameterized by $r+\Pi$ to \textsc{BECGP} parameterized by $r+\Pi$.  
\end{theorem}
\begin{proof}
Let $(G,r,\gamma)$ be an instance of \textsc{ECGP}, and let $n:=|V(G)|$. Construct a graph $G'$ from $G$ by adding exactly $(r-1)n$ isolated vertices. Then $|V(G')|=rn$, so in the corresponding instance of \textsc{BECGP} every part must have size exactly $n$.

We claim that $(G,r,\gamma)$ is a YES-instance of \textsc{ECGP} if and only if $(G',r,\gamma)$ is a YES-instance of \textsc{BECGP}.

\medskip
\noindent
\emph{Forward direction.}
Suppose $(G,r,\gamma)$ is a YES-instance. Then there is a partition of $V(G)$ into $r$ parts $C_1,\dots,C_r$ such that $|E(G[C_j])| \ge \gamma$ for every $j \in [r]$. For each $j$, add enough isolated vertices of $G'$ to $C_j$ so that its size becomes exactly $n$. Since the total number of added isolated vertices is
$r n - \sum_{j=1}^r |C_j| = rn - n = (r-1)n$,
this is possible. As isolated vertices contribute no edges, each part still induces at least $\gamma$ edges. Hence $(G',r,\gamma)$ is a YES-instance of \textsc{BECGP}.

\medskip
\noindent
\emph{Reverse direction.}
Suppose $(G',r,\gamma)$ is a YES-instance of \textsc{BECGP}. Then $V(G')$ can be partitioned into $r$ parts $D_1,\dots,D_r$, each of size exactly $n$, such that $|E(G'[D_j])| \ge \gamma$ for every $j \in [r]$. Let $C_j := D_j \cap V(G)$ for each $j \in [r]$. Since all added vertices are isolated, they contribute no edges, and therefore $|E(G[C_j])| = |E(G'[D_j])| \ge \gamma$ for every $j \in [r]$. Moreover, the sets $C_1,\dots,C_r$ form a partition of $V(G)$. Hence $(G,r,\gamma)$ is a YES-instance of \textsc{ECGP}.

\medskip

The reduction runs in polynomial time and preserves the values of $r$ and $\gamma$. Since adding isolated vertices does not increase $\Pi$, the parameter $r+\Pi$ is preserved. Therefore this is a parameterized reduction.
\end{proof}

\begin{corollary}
\label{cor:buccp-hardness}
\textsc{BECGP} is W[1]-hard under the following parameterizations:
\begin{itemize}
    \item $r+\mathsf{cvd}$,
    \item $r+\mathsf{fes}+\mathsf{vds}+\mathsf{mw}$, even when $\mathsf{fes}=0$, $\mathsf{vds}=0$, and $\mathsf{mw}=2$,
    \item $r+\mathsf{fes}+\mathsf{vdp}+\Delta$, even when $\mathsf{fes}=0, \Delta=2$ and $\mathsf{vdp}=0$.
\end{itemize}
\end{corollary}
\begin{proof}
Each statement follows from the corresponding hardness result for \textsc{ECGP} (Theorem~\ref{thm:k-cvd-hard}, Corollary~\ref{cor:vds-hard}, and Corollary~\ref{cor:vdp-hard}) together with Theorem~\ref{thm:uccp-to-buccp}. 
In each case, the reduction preserves the parameters, since adding isolated vertices does not increase $\mathsf{cvd}$, $\mathsf{fes}$, $\mathsf{vds}$, $\mathsf{vdp}$, or $\mathsf{mw}$.
\end{proof}

\begin{theorem} \label{thm:cw-hard-gamma-three} 
Both \textsc{ECGP} and \textsc{BECGP} are W[1]-hard parameterized by clique-width, even when $\gamma=3$. 
\end{theorem}
\begin{proof} We reduce from \textsc{$K_3$-Clique Partition}, which is W[1]-hard parameterized by clique-width~\cite{bojikian_et_al:LIPIcs.ESA.2026.115}. 
Given a graph $G=(V,E)$ with $|V|$ divisible by $3$, construct the instance $(G,r,\gamma)$, where $r:=|V|/3$ and $\gamma:=3$. 
If $G$ admits a partition into triangles, then the same partition is feasible for both \textsc{ECGP} and \textsc{BECGP}, 
since every triangle induces exactly three edges.

Conversely, let $V_1,\ldots,V_r$ be a feasible \textsc{ECGP} partition. Since every part induces at least three edges, it contains at least three vertices. As $|V|=3r$, every part contains exactly three vertices. A graph on three vertices induces at least three edges if and only if it is a triangle. Hence $V_1,\ldots,V_r$ form a $K_3$-clique partition. The same argument immediately applies to \textsc{BECGP}, where every part has size $|V|/r=3$. 
The graph is unchanged, so its clique-width is preserved. Therefore, both problems are W[1]-hard parameterized by clique-width, even for $\gamma=3$. \end{proof}

\section{Hardness on Signed Graphs}

\begin{theorem}\label{thm:NP-hard signed ECGP}
The \textsc{SECGP} is NP-hard even on disjoint union of two cliques, even when $r=3$ and $\gamma=0$.
\end{theorem}
\begin{proof}
We give a polynomial-time reduction from \textsc{3-Coloring}.
Let $G=(V,E)$ be an instance of \textsc{3-Coloring} with $|V|=n$. We construct a signed graph $H$ as follows.

\begin{itemize}
    \item Take a copy of $G$ on vertex set $V$, and color every edge of $G$ red.
    \item Add all missing edges on $V$, and color them green. Thus $V$ induces a clique in $H$.
    \item Add a new clique $K$ on $n$ vertices, and color all its edges red.
\end{itemize}

Thus $H$ is the disjoint union of two cliques. We set $r:=3$ and $\gamma:=0$.
We claim that $G$ is $3$-colorable if and only if $(H,r,\gamma)$ is a YES-instance.

\medskip
\noindent
\emph{Forward direction.}
Suppose that $G$ is $3$-colorable. Let $V_1,V_2,V_3$ be a partition of $V(G)$ into three independent sets, and let $n_i:=|V_i|$ for each $i \in [3]$, so $n_1+n_2+n_3=n$.
Partition the red clique $K$ into three parts $K_1,K_2,K_3$ such that $|K_i|=n_i$ for every $i \in [3]$, and define $D_i := V_i \cup K_i$ for every $i \in [3]$.

Since $V_i$ is an independent set in $G$, every pair of vertices of $V_i$ is joined by a green edge in $H$, so $V_i$ contributes exactly $\binom{n_i}{2}$ green edges. The set $K_i$ contributes exactly $\binom{n_i}{2}$ red edges, and there are no edges between $V$ and $K$. Hence the edge-count of $D_i$ is $\binom{n_i}{2} - \binom{n_i}{2} = 0$, so each part is feasible.

\medskip
\noindent
\emph{Reverse direction.}
Suppose that $(H,r,\gamma)$ is a YES-instance, and let $D_1,D_2,D_3$ be a feasible partition.
First consider the red clique $K$. Let $n_i := |K \cap D_i|$ for each $i \in [3]$. Then $n_1+n_2+n_3=n$.

\begin{claim}
For each $i \in [3]$, the part $D_i$ contains exactly $n_i$ vertices from $V$, that is, $|V \cap D_i| = n_i$.
\end{claim}

\begin{claimproof}
Let $x_i := |V \cap D_i|$. Since the vertices of $V$ induce a clique in $H$, the number of green edges inside $V \cap D_i$ is at most $\binom{x_i}{2}$. The number of red edges inside $K \cap D_i$ is exactly $\binom{n_i}{2}$, and there are no edges between $V$ and $K$.
Since $D_i$ has edge-count at least $0$, we must have $\binom{x_i}{2} \ge \binom{n_i}{2}$, and hence $x_i \ge n_i$.
Summing over all $i$, we get $\sum_i x_i = n = \sum_i n_i$, so it follows that $x_i = n_i$ for all $i \in [3]$.
\end{claimproof}

\noindent Thus each part $D_i$ contains exactly $n_i$ vertices from $V$ and exactly $n_i$ vertices from $K$.
We now show that each set $V \cap D_i$ induces no edge of $G$.

\begin{claim}
For each $i \in [3]$, the set $V \cap D_i$ is an independent set in $G$.
\end{claim}

\begin{claimproof}
Let $e_i$ be the number of edges of $G$ induced by $V \cap D_i$. Then the number of green edges inside $V \cap D_i$ is $\binom{n_i}{2} - e_i$, while the number of red edges inside $D_i$ is $\binom{n_i}{2} + e_i$ (coming from $K$ and from $G$).
Hence the edge-count of $D_i$ is $(\binom{n_i}{2} - e_i) - (\binom{n_i}{2} + e_i) = -2e_i$, which must be at least $0$. Therefore $e_i=0$.
\end{claimproof}

Thus each $V \cap D_i$ is an independent set, and $(V \cap D_1, V \cap D_2, V \cap D_3)$ is a proper $3$-coloring of $G$.
This proves that $G$ is $3$-colorable if and only if $(H,r,\gamma)$ is a YES-instance. Since \textsc{3-Coloring} is a well-known NP-hard problem, the result follows.
\end{proof}

If we replace the source problem \textsc{$3$-Coloring} with the \textsc{Equitable $3$-Coloring} in the reduction then we get NP-hardness for \textsc{SBECGP}.

\begin{corollary}\label{cor:NP-hard signed BECGP}
The \textsc{SBECGP} is NP-hard even when $r=3$ and $\gamma=0$ when the input graph is a disjoint union of two cliques.
\end{corollary}

\noindent Next, we show that the \textsc{SBECGP} remains hard even under strong structural restrictions.

\begin{theorem}
\label{thm:signed-buccp-tw}
The \textsc{SBECGP} is W[1]-hard when parameterized by $\mathsf{tw}+r+\gamma$, even when $\gamma=0$.
\end{theorem}
\begin{proof}
We give a parameterized reduction from the tight version of \textsc{Unary Bin Packing}, parameterized by the number $r$ of bins. Thus the input consists of positive integers $a_1,\dots,a_n$, a bin capacity $B$, and an integer $r$, such that $\sum_{i=1}^n a_i = Br$. As usual, we may assume that $a_i < B$ for every $i \in [n]$.

We construct a signed graph $G'$ as follows.

\begin{itemize}
    \item For each item $a_i$, add a cycle $C_i$ on $a_i$ vertices, and color all edges of $C_i$ green.
    \item Add exactly $Br$ pairwise vertex-disjoint copies of the clique $K_{r+1}$, and color all their edges red.
\end{itemize}

We ask whether $V(G')$ can be partitioned into exactly $r$ parts of equal size, each having edge-count at least $\gamma:=0$.
Let us first determine the common part size. The total number of vertices of $G'$ is
\[
\sum_{i=1}^n a_i + Br(r+1) = Br + Br(r+1) = Br(r+2),
\]
and therefore every part must have size exactly
$\frac{|V(G')|}{r} = B(r+2)$.
We claim that the tight \textsc{Unary Bin Packing} instance is a YES-instance if and only if the constructed signed \textsc{BECGP} instance is a YES-instance.

\medskip
\noindent
\emph{Forward direction.}
Suppose that the items can be packed into $r$ bins $X_1,\dots,X_r$ such that $\sum_{i \in X_j} a_i = B$ for every $j \in [r]$.
For each $j \in [r]$, place in part $D_j$ all vertices of the cycles $C_i$ with $i \in X_j$. Since each cycle $C_i$ has exactly $a_i$ green edges, the total number of green edges in $D_j$ is exactly
$\sum_{i \in X_j} a_i = B$.

Now distribute the $Br$ red cliques as follows. For each part $D_j$, designate exactly $B$ red cliques to contribute one red edge to $D_j$. In each such clique, place two vertices in $D_j$ and one vertex in every other part. Then that clique contributes exactly one red edge, and this edge lies in $D_j$.

Thus every part $D_j$ receives exactly $B$ red edges. Moreover, from each of the $Br$ red cliques, part $D_j$ receives at least one vertex, and from the $B$ cliques assigned to $D_j$ it receives one additional vertex. Hence the number of red-gadget vertices in $D_j$ is $Br + B$.
The number of green-gadget vertices in $D_j$ is
$\sum_{i \in X_j} a_i = B$.
Therefore
$|D_j| = (Br+B) + B = B(r+2)$,
so the partition is balanced. Finally, the edge-count of $D_j$ is
$B - B = 0$.
Hence the constructed signed \textsc{BECGP} instance is a YES-instance.

\medskip
\noindent
\emph{Reverse direction.}
Suppose that the constructed signed \textsc{BECGP} instance is a YES-instance, and let $D_1,\dots,D_r$ be a balanced feasible partition of $V(G')$.

For each $j \in [r]$, let $g_j$ denote the number of green edges induced by $D_j$, and let $q_j$ denote the number of red edges induced by $D_j$. Since the edge-count threshold is $0$, we have
\[
g_j - q_j \ge 0
\qquad\text{for every } j \in [r].
\]

We now establish several claims.

\begin{claim}
\label{cl:totals}
We have $\sum_{j=1}^r g_j = Br$, $\sum_{j=1}^r q_j = Br$, and $g_j=q_j$ for every $j \in [r]$.
\end{claim}

\begin{claimproof}
The total number of green edges in the whole graph is $\sum_{i=1}^n a_i = Br$,
hence
$\sum_{j=1}^r g_j \le Br$.
On the other hand, each red clique is a copy of $K_{r+1}$, and therefore is not $r$-colorable. Hence in any partition into $r$ parts, each such clique contributes at least one red edge to some part. Since there are exactly $Br$ red cliques, $\sum_{j=1}^r q_j \ge Br$.
Now $\sum_{j=1}^r (g_j-q_j) \ge 0$
because each term is nonnegative. Therefore
\[
0 \le \sum_{j=1}^r (g_j-q_j)
   = \sum_{j=1}^r g_j - \sum_{j=1}^r q_j
   \le Br - Br = 0.
\]
Hence equality holds throughout. In particular,
$\sum_{j=1}^r g_j = Br$,
$\sum_{j=1}^r q_j = Br$,
$g_j-q_j=0$ for every  $j \in [r]$.
Thus $g_j=q_j$ for every $j \in [r]$.
\end{claimproof}

\begin{claim}
\label{cl:cycles-whole}
Every green cycle $C_i$ is entirely contained in a single part.
\end{claim}

\begin{claimproof}
Suppose that some cycle $C_i$ intersects at least two parts. Then at least one edge of $C_i$ has endpoints in different parts, and therefore this green edge is not counted in any $g_j$. Hence
\[
\sum_{j=1}^r g_j < \sum_{i=1}^n a_i = Br,
\]
contradicting Claim~\ref{cl:totals}. Therefore every cycle $C_i$ is entirely contained in one cluster.
\end{claimproof}

\begin{claim}
\label{cl:red-clique-shape}
Every red clique is distributed among the $r$ parts in such a way that exactly one part receives two vertices of the clique, and every other part receives exactly one vertex.
\end{claim}

\begin{claimproof}
Fix one red clique $Q \cong K_{r+1}$. Since $Q$ contributes at least one red edge to the partition and Claim~\ref{cl:totals} implies that the total number of red edges contributed by all red cliques is exactly $Br$, it follows that each red clique contributes exactly one red edge.

Let $s_1,\dots,s_r$ be the numbers of vertices of $Q$ placed in the $r$ parts. Then
\[
s_1+\cdots+s_r = r+1,
\qquad s_j \ge 0.
\]
The number of red edges contributed by $Q$ is
$\sum_{j=1}^r \binom{s_j}{2}$.
Since this sum is exactly $1$, there is exactly one index $j$ with $s_j=2$, and for every other index we have $s_j \in \{0,1\}$. As the total number of vertices is $r+1$, it follows that in fact one part gets two vertices and every other part gets exactly one vertex.
\end{claimproof}

\begin{claim}
\label{cl:red-vertices-count}
For every part $D_j$, the number of red-gadget vertices contained in $D_j$ is exactly $Br + q_j$.
\end{claim}

\begin{claimproof}
There are $Br$ red cliques. By Claim~\ref{cl:red-clique-shape}, from each red clique part $D_j$ receives at least one vertex. Hence $D_j$ receives $Br$ baseline red-gadget vertices.

In addition, whenever a red clique contributes its unique red edge to $D_j$, part $D_j$ receives one extra vertex from that clique. Since $q_j$ is exactly the number of red cliques whose unique red edge lies in $D_j$, the number of extra red-gadget vertices received by $D_j$ is exactly $q_j$.
Therefore the total number of red-gadget vertices in $D_j$ is $Br+q_j$.
\end{claimproof}

We now finish the proof.

By Claim~\ref{cl:cycles-whole}, the green edges in a part $D_j$ come exactly from whole cycles placed inside $D_j$. Since a cycle $C_i$ has as many vertices as green edges, the number of green-gadget vertices in $D_j$ is exactly $g_j$.
By Claim~\ref{cl:red-vertices-count}, the number of red-gadget vertices in $D_j$ is exactly $Br+q_j$.
Hence the total size of $D_j$ is
$|D_j| = g_j + (Br+q_j$).
Since the partition is balanced, every part has size exactly $B(r+2)$, and therefore
$g_j + Br + q_j = B(r+2)$.
Using Claim~\ref{cl:totals}, namely $g_j=q_j$, we obtain
$Br + 2g_j = B(r+2)$,
and thus $g_j = B$
for every  $j \in [r]$.

Since each green cycle is intact and contributes exactly its item size in green edges, it follows that the cycles assigned to part $D_j$ correspond to items of total size exactly $B$. Therefore the items are partitioned into $r$ bins of capacity exactly $B$, and the original tight \textsc{Unary Bin Packing} instance is a YES-instance.
This proves correctness of the reduction.

Finally, the constructed graph is a disjoint union of cycles and copies of $K_{r+1}$. Hence its treewidth is
$\max\{2,r\} = r$.

The reduction preserves the parameter $r$, and sets $\gamma=0$. Therefore this is a parameterized reduction showing that \textsc{SBECGP} is W[1]-hard when parameterized by $\mathsf{tw}+r+\gamma$, even when $\gamma=0$.
\end{proof}

\begin{corollary}
\label{cor:signed-buccp-td}
The \textsc{SBECGP} is W[1]-hard when parameterized by $\mathsf{td}+r$, even when $=0$.
\end{corollary}
\begin{proof}
In the proof of Theorem~\ref{thm:signed-buccp-tw}, the only properties of the cycle $C_i$ used for an item of size $a_i$ are that it is connected and satisfies $|V(C_i)|=|E(C_i)|=a_i$. Therefore, the same proof remains valid if each cycle $C_i$ is replaced by any connected graph $H_{a_i}$ with $|V(H_{a_i})|=|E(H_{a_i})|=a_i$.
For every integer $s \ge 3$, let $H_s$ be the graph obtained from a triangle by attaching $s-3$ leaves to one of its vertices. Then $H_s$ is connected, $|V(H_s)|=|E(H_s)|=s$, and $\mathsf{td}(H_s)\le 3$.

In the proof of Theorem~\ref{thm:signed-buccp-tw}, we may replace every cycle $C_i$ by $H_{a_i}$. The resulting graph is a disjoint union of copies of $K_{r+1}$ and graphs of treedepth at most $3$. Hence the whole graph has treedepth at most $\max\{r+1,3\}=r+1$.
Thus the same reduction proves W[1]-hardness parameterized by $\mathsf{td}+r+\gamma$, even when $\gamma=0$.
\end{proof}

\section{Conclusion}

In this paper, we introduced and studied the computational and parameterized complexity of the \textsc{Edge-Constrained Graph Partitioning Problem} and its balanced and signed variants.
These problems provide a natural framework for coalition formation in which the effectiveness of each group is measured locally by the number of interactions induced by its members. 
Unlike many classical clustering and coalition-formation models that optimize a global objective, our formulations require every individual part to satisfy a prescribed minimum utility threshold.

We first established that both \textsc{ECGP} and \textsc{BECGP} remain NP-hard under several restricted settings. We then obtained polynomial kernels for both problems when parameterized by $r+u$, using the Expansion Lemma as the main ingredient. We also developed fixed-parameter algorithms for several structural parameterizations, including vertex deletion distance to a clique and vertex integrity. 
For \textsc{ECGP}, we obtained additional fixed-parameter algorithms for maximum leaf number, cluster vertex deletion number plus $u$, vertex deletion distance to stars plus $u$, vertex deletion distance to paths plus $u$, and treewidth plus $r+u$. 

In contrast, we proved that the problems remain $W[1]$-hard under several other structural parameterizations, including combinations involving the number of parts and deletion distances to clusters, stars, and paths. For signed graphs, we showed that the problems remain computationally hard even under strong restrictions on the input graph and for fixed values of the natural parameters.
In particular, the signed variants are NP-hard on disjoint unions of two cliques, and the balanced signed variant is $W[1]$-hard when parameterized by treedepth plus $r+u$, even when the utility threshold is zero. These results demonstrate that allowing both positive and negative interactions substantially changes the complexity of the problem. 

Several questions remain open. In particular, the parameterized complexity of \textsc{ECGP} and \textsc{BECGP} with respect to neighborhood diversity deserves further investigation. The main obstacle is that counting the edges induced by a part naturally gives rise to quadratic terms involving the numbers of vertices selected from different neighborhood types. 
Although we overcome a related difficulty for vertex deletion distance to a clique by guessing the relevant quantities from bounded ranges, this approach does not extend directly to multiple neighborhood types. 
It would also be interesting to determine whether both the problems admit FPT algorithm when parameterized by $\tw+u$. In fact, the problem is open even when $\tw(G)=1$, that is, when input graph is forest.\\

\noindent \textbf{AI Declaration.} During the preparation of this manuscript, Microsoft Copilot and Google Gemini were used to assist with language editing, the presentation of the results, and the construction of figures. All AI-assisted content was subsequently reviewed and verified by the authors, who assume full responsibility for the content of the manuscript.

\bibliography{lipics-v2021-sample-article}
\end{document}